\documentclass[]{hdsr}

\graphicspath{{figs/}}

\usepackage{nth}
\usepackage{color}
\usepackage{float}

\usepackage{mathrsfs}  %
\usepackage{stmaryrd}
\usepackage{bm}
\usepackage{appendix}
\usepackage[ruled]{algorithm2e}

\usepackage{bookmark}

\hypersetup{
	pdftitle={Invariant-Preserving Deployments of Differential Privacy},
	pdfauthor={James Bailie, Ruobin Gong, Xiao-Li Meng},
	pdfkeywords={confidentiality, data swapping, TopDown Algorithm, invariant statistics, statistical disclosure control},
	pdfdisplaydoctitle=true,
	pdfcreator={James Bailie, Ruobin Gong, Xiao-Li Meng}
}

\let\originalleft\left
\let\originalright\right
\renewcommand{\left}{\mathopen{}\mathclose\bgroup\originalleft}
\renewcommand{\right}{\aftergroup\egroup\originalright}

\let\amsnewtheorem\newtheorem %
\usepackage{mathtools} 
\DeclarePairedDelimiter\abs{\lvert}{\rvert}%
\DeclarePairedDelimiter\norm{\lVert}{\rVert}%
\makeatletter
\let\oldabs\abs
\def\abs{\@ifstar{\oldabs}{\oldabs*}}
\let\oldnorm\norm
\def\norm{\@ifstar{\oldnorm}{\oldnorm*}}
\makeatother

\renewcommand{\epsilon}{\varepsilon}

\newcommand{\X}{\mathcal{X}}

\newcommand{\T}{\mathcal T}
\newcommand{\U}{\mathcal{U}}

\newcommand{\sfP}{\mathsf P}

\newcommand{\sfPx}{\mathsf P\hspace{-0.07em}_{\xd}}

\newcommand{\sfPxdash}{\mathsf P\hspace{-0.07em}_{\xd'}}

\newcommand{\sfQ}{\mathsf Q}

\newcommand{\Xcef}{\X_{\mathrm{CEF}}}
\newcommand{\Xprod}{\X_{\times}}

\usepackage{bbm}

\usepackage{multirow, colortbl}
\usepackage{algorithmic}
\usepackage{xpatch}

\usepackage{todonotes}
\usepackage{xcolor}
\usepackage[normalem]{ulem} %

\newcommand{\floor}[1]{\left\lfloor #1 \right\rfloor}

\newcommand{\xd}{\bm{x}}

\newcommand{\Mult}{D_{{\normalfont \textsc{Mult}}}}

\newcommand{\Dnor}{D_{{\normalfont \textsc{NoR}}}}
\newcommand{\dX}{d_{\mathcal X}}

\newcommand{\dT}{\dPr}
\newcommand{\dPr}{D_{\Pr}}

\newcommand{\dHamS}{d_{\mathrm{HamS}}}
\newcommand{\dHamSh}{\dHamS^{hh}}
\newcommand{\dHamSr}{\dHamS^{r}}
\newcommand{\scD}{\mathscr{D}}
\newcommand{\D}{\mathcal{D}}
\newcommand{\DPSpec}{$\epsilon_{\mathcal D}$-DP\DPFlav}

\newcommand{\DPSpecInMath}{\epsilon_{\mathcal D}\text{-DP}\DPFlavInMath}
\newcommand{\DPFlavInMath}{(\mathcal X,\allowbreak \mathscr D,\allowbreak \dX,\allowbreak \dPr)}
\newcommand{\DPFlav}{$\DPFlavInMath$}

\newcommand{\FlavSwap}{$(\Xcef,\allowbreak \scD_{\cswap},\allowbreak \dHamSh,\allowbreak \Mult)$}
\newcommand{\SpecSwap}{$\epsilon_{\D}$-DP\FlavSwap}

\newcommand{\FlavSwapProd}{$(\Xprod,\allowbreak \DSwap,\allowbreak \dHamSh,\allowbreak \Mult)$}
\newcommand{\SpecSwapProd}{$\epsilon_{\D}$-DP\FlavSwapProd}

\newcommand{\FlavSwapGen}{$(\X,\allowbreak \scD_{\cswap},\allowbreak \dHamSr,\allowbreak \Mult)$}
\newcommand{\SpecSwapGen}{$\epsilon_{\D}$-DP\FlavSwapGen}

\newcommand{\FlavSwapProdGen}{$(\Xprod,\allowbreak \DSwap,\allowbreak \dHamSr,\allowbreak \Mult)$}
\newcommand{\SpecSwapProdGen}{$\epsilon_{\D}$-DP\FlavSwapProdGen}

\newcommand{\FlavTDA}{$(\Xcef,\allowbreak \scD_{\cTDA},\allowbreak \dHamS^p,\allowbreak \Dnor)$}
\newcommand{\FlavhTDA}{$(\Xcef,\allowbreak \scD_{\cTDA},\allowbreak \dHamSh,\allowbreak \Dnor)$}
\newcommand{\SpecTDA}{$\rho$-DP\FlavTDA}
\newcommand{\SpecTDAPrime}{$\rho$-DP$(\Xcef,\allowbreak \scD_{\bm c'},\allowbreak \dHamS^p,\allowbreak \Dnor)$}

\newcommand{\PLD}{L_{\D}}
\newcommand{\PLDdash}{L'_{\D'}}
\newcommand{\PLDz}{L_{\D_0}}
\newcommand{\PLDdashz}{L'_{\D_0'}}
\newcommand{\DPFlavDiffMultiverseInMath}{(\mathcal X,\allowbreak \mathscr D',\allowbreak \dX,\allowbreak \dPr)}
\newcommand{\DPFlavDiffMultiverse}{$\DPFlavDiffMultiverseInMath$}
\newcommand{\DPSpecDiffMultiverseInMath}{\epsilon'_{\mathcal D'}\text{-DP}\DPFlavDiffMultiverseInMath}

\newcommand{\M}{\mathcal{M}}

\newcommand{\epsinf}{\epsilon^{(\inf)}}

\newcommand{\Vars}{\bm V}
\newcommand{\Vmat}{\Vars_{\mathrm{Match}}}

\newcommand{\Vswap}{\Vars_{\mathrm{Swap}}}
\newcommand{\Vhold}{\Vars_{\mathrm{Hold}}}

\newcommand{\xdHold}{\xd_{\mathrm{Hold}}}
\newcommand{\xdSwap}{\xd_{\mathrm{Swap}}}
\newcommand{\xdMat}{\xd_{\mathrm{Match}}}
\newcommand{\cswap}{\bm c_{\mathrm{Swap}}}

\newcommand{\cTDA}{\bm c_{\mathrm{TDA}}}

\newcommand{\Df}{\mathcal D}%
\newcommand{\DSwap}{\scD_{\cswap}}
\newcommand{\DTDA}{\scD_{\cTDA}}

\newcommand{\Xh}{\xdh}
\newcommand{\xdh}{\xd_{hh}}
\newcommand{\Xp}{\xd_{p}}
\newcommand{\Dh}{\Sigma_{hh}}
\newcommand{\Dp}{\Sigma_{p}}
\newcommand{\Qh}{\bm Q_{hh}}
\newcommand{\Qp}{\bm Q_{p}}
\newcommand{\Th}{\bm T_{hh}}
\newcommand{\Tp}{\bm T_{p}}
\newcommand{\Zh}{\bm Z_{hh}}
\newcommand{\Zp}{\bm Z_{p}}
\newcommand{\Wh}{\bm W_{hh}}
\newcommand{\Wp}{\bm W_{p}}

\newcommand{\id}{\mathsf{id}}

\newtheorem{theorem}{Theorem}
\newtheorem{proposition}{Proposition}
\newtheorem{lemma}{Lemma}
\theoremstyle{definition}
\newtheorem{definition}{Definition}

\theoremstyle{remark}
\newtheorem{remark}{Remark}
\newtheorem{example}{Example}
\amsnewtheorem*{excont}{Example \continuation} %
\newcommand{\continuation}{??}
\newenvironment{continueexample}[1]
{\renewcommand{\continuation}{\ref{#1}}\excont[continued]}
{\endexcont}

\renewcommand{\algorithmicrequire}{\textbf{Input: }}

\SetKw{CONTINUE}{continue}
\SetKw{GOTO}{go to}
\SetKw{OUTPUT}{Output:}
\makeatletter
\xpatchcmd\algorithmic
{\newcommand{\IF}}
{%
	\newcommand\SCOPE{\begin{ALC@g}}%
		\newcommand\ENDSCOPE{\end{ALC@g}}%
	\newcommand{\IF}%
}
{}{\fail}
\makeatother

\letbibmacro{orig-doi+url}{doi+url}
\renewbibmacro*{doi+url}{%
	\iffieldundef{doi}%
	{\iffieldundef{url}%
		{\printfield{howpublished}\nopunct}%
		{\usebibmacro{orig-doi+url}}%
	}%
	{\usebibmacro{orig-doi+url}}%
}%

\AtEveryBibitem{\clearfield{month}}
\AtEveryCitekey{\clearfield{month}}
\AtEveryBibitem{\clearfield{labelmonth}}
\AtEveryCitekey{\clearfield{labelmonth}}

\begin{document}

\newgeometry{bottom=1.5in}

\volumeheader{0}{0}{00.000}

\begin{center}

  \title[Invariant-Preserving Deployments of Differential Privacy]{A Refreshment Stirred, Not Shaken: Invariant-Preserving Deployments of Differential Privacy for the U.S. Decennial Census}
  \maketitle

  \thispagestyle{empty}
  
  \vspace*{.2in}

  \begin{tabular}{cc}
    James Bailie\upstairs{\affilone,*}, Ruobin Gong\upstairs{\affiltwo}, Xiao-Li Meng\upstairs{\affilone}
   \\[0.25ex]
   {\small \upstairs{\affilone} Department of Statistics, Harvard University} \\
   {\small \upstairs{\affiltwo} Department of Statistics, Rutgers University} \\
  \end{tabular}
  
  \emails{
    \upstairs{*}jamesbailie@g.harvard.edu 
    }

\begin{abstract} 
    Protecting an individual's privacy when releasing their data is inherently an exercise in relativity, regardless of how privacy is qualified or quantified. This is because we can only limit the gain in information about an individual relative to what could be derived from other sources. This framing is the essence of differential privacy (DP), through which this article examines two statistical disclosure control (SDC) methods for the United States Decennial Census: the Permutation Swapping Algorithm (PSA), which resembles the 2010 Census's disclosure avoidance system (DAS), and the TopDown Algorithm (TDA), which was used in the 2020 DAS. To varying degrees, both methods leave unaltered certain statistics of the confidential data---their invariants---and hence neither can be readily reconciled with DP, at least as originally conceived. Nevertheless, we show how invariants can naturally be integrated into DP and use this to establish that the PSA satisfies pure DP subject to the invariants it necessarily induces, thereby proving that this traditional SDC method can, in fact, be understood from the perspective of DP. By a similar modification to zero-concentrated DP, we also provide a DP specification for the TDA. Finally, as a point of comparison, we consider a counterfactual scenario in which the PSA was adopted for the 2020 Census, resulting in a reduction in the nominal protection loss budget but at the cost of releasing many more invariants. This highlights the pervasive danger of comparing budgets without accounting for the other dimensions on which DP formulations vary (such as the invariants they permit). Therefore, while our results articulate the mathematical guarantees of SDC provided by the PSA, the TDA, and the 2020 DAS in general, care must be taken in translating these guarantees into actual privacy protection---just as is the case for any DP deployment.
\end{abstract}
\end{center}

\hspace{10pt}

\copyrightnotice

\vspace{-2pt}
\section*{Media Summary}

Preserving data privacy when publishing statistics is a permanent challenge for every organization that provides public use data files. This article tells two stories of how the U.S. Census Bureau dealt with this challenge in its 2010 and 2020 Decennial Censuses using two different systems. By incorporating both theses approaches into a broad mathematical framework called \emph{differential privacy}, the authors show how they can be formally understood. This study clarifies what promises these systems truly make, and what they do not, reminding us that mathematical assurances do not automatically ensure real-world privacy.

\section{Motivations and Contributions}
\label{sectionIntro}

\subsection{Privacy Protection Under the Constraint of Invariants}\label{sec:invariantRequirment}

In 2018, the United States Census Bureau (USCB) announced an overhaul of its disclosure avoidance system (DAS), retiring the data swapping methods that had been central to protecting the U.S. Decennial Census for the previous 30 years \citep{abowd2018us, mckennaDisclosureAvoidanceTechniques2018}. %
While the privacy provided by these methods had been justified with intuitive arguments, %
the DAS for the 2020 Census would, in contrast, 
be redesigned from the ground up, with the primary goal of supplying a mathematical guarantee of protection. Moreover, this guarantee, the USCB decided, must be some type of \emph{differential privacy} (DP) \citep{dwork2006calibrating}---a large family of technical standards \citep{desfontainesSoKDifferentialPrivacies2022} that conceptualize the `privacy' of a statistical disclosure control (SDC) method in terms of its sensitivity to counterfactual changes in its input data \citep{bailieRefreshmentStirredNot2026}.%

The USCB's adoption of DP was driven in part by the desire for a formal, quantitative, and measurable characterization of privacy protection. The bureau was concerned that the existing, swapping-based DAS lacked the rigorous basis supplied by such a characterization. Indeed, data swapping had not been analyzed by a formal system like DP, and thus a theoretical account of its SDC protection was limited. At the same time, the bureau's empirical evaluation---that is, their reconstruction and reidentification attack (see Section~\ref{secImpactInvariants} and \cite{abowd2010CensusConfidentiality2023})---concluded that the existing DAS was vulnerable, while in comparison, a DP-compliant DAS was expected to better protect against this and other emerging types of privacy attacks \citep{dworkExposedSurveyAttacks2017}.

Nevertheless, there were other priorities for the 2020 Census, many of which complicated a straightforward adoption of DP. In particular, state population counts are %
constitutionally mandated to be published exactly as counted%
, whereas DP---at least as originally defined in \citet{dwork2006calibrating}---requires that such counts be infused with random noise. %

Thus, even as DP offered the kind of formal guarantees the bureau sought, its implementation had to contend with the legal, operational, and statistical requirements of the Decennial Census---requirements that impose strict conditions on how SDC is applied. Indeed, the U.S. Decennial Census is a massive exercise in population enumeration, which is subject to numerous laws and regulations, as well as various pragmatic constraints and utility considerations. Taken together, these external criteria greatly complicated the design of the new 2020 DAS. A key challenge was the requirement that census publications must respect \emph{invariants}---exact summaries of the confidential data that must be released without modification. The most notable invariants in 2020 are the state population totals mentioned previously, but for operational and data-quality reasons the USCB also incorporated additional invariants into the 2020 Census, including counts of housing units at the lowest level of census geography (blocks) and various other statistics (summarized in Table~\ref{tab:compare_2020}).

While we have so far focused exclusively on the 2020 Census, invariants were not a novel problem in 2020. To the contrary, although their exact composition and method of imposition has changed from decade to decade, invariants have been a mainstay of every Decennial Census due to their constitutional and regulatory significance.
Moreover, many other types of data dissemination frequently feature key statistics that are not altered before publication---for example, the row and column margins of contingency tables are often published as is (see Example~\ref{exContingencyTable} in Section~\ref{secMultiverseAccommodatesInvariants}). Thus, invariants are inherent not only to the dissemination of U.S. Census data, but are also a common property of other statistical data products.

Yet the presence of invariants complicates SDC. By definition, the values of any invariant statistics cannot be modified by a DAS and as such are exempt from any SDC protection. Naturally, this creates potential statistical disclosure risks because exact knowledge of certain features of the confidential data, as is provided by invariants, may aid an attacker in reconstructing and reidentifying these data. %

Furthermore, standard formulations of DP---including those that the USCB has invoked (zero-concentrated DP), referenced (approximate DP), or at some point considered (pure DP)---do not allow for the specification of invariants. The issue is that DP, at least as it is typically understood, cannot measure the protection provided to the confidential data after taking into account the release of invariants. Even though DP is fundamentally an assessment of relative privacy---the privacy of an individual's unique information relative to knowing the rest of the population---it must be recast in a more general light in order to assess the privacy relative to knowing the invariants. %
Yet once one recognizes DP's relative nature, it becomes apparent that invariants do not contradict the fundamental essence of DP, but rather that their incorporation into DP greatly expand its applicability. 

This is because no method can protect absolute privacy \citep{kiferNoFreeLunch2011}. Hence any data protection standard, including DP, must take into account information that cannot be protected, whether that be invariants or other public information. As has largely gone unnoted, to spell this out is not only natural, but key to articulating the actual DP guarantee of any invariant-preserving SDC method, including those used for the 2020 Census. One thing that is well understood in the literature, however, is that invariants nullify the DP protection guarantees along the dimensions of the confidential data that the invariants implicate (see, e.g., \cite{gong2020congenial} and references discussed in Appendix~\ref{secRelatedWork}). 

Nevertheless, the compromising effect of invariants on the DP guarantees of the Decennial Census is a central motivation of this article. Of particular importance in this regard is a DP analysis of the TopDown Algorithm (TDA), which the USCB created for disseminating key 2020 Census data products. This algorithm runs in two steps: It first adds DP-calibrated noise to all of the 2020 Census data, then it removes this
noise from the invariants via a complex optimization procedure \citep{abowd2022topdown}. While a DP analysis of the first step is easy due to its use of an established procedure \citep{canonneDiscreteGaussianDifferential2022}, the second step, as the bureau's own assessment makes clear \citep{ashmead2019effective}, is particularly challenging to analyze from the perspective of DP because of the invariants it enforces. Yet a complete and rigorous assessment of the TDA's mathematical guarantee of protection must address the entire procedure---both the noise infusion in the first step and the noise removal in the second step. This assessment is, to the best of the authors' knowledge, missing from the literature thus far.

As we have mentioned, invariants were not just a requirement in the 2020 DAS but were also present in all other Decennial Censuses. Most relevant to a discussion of the USCB's overhaul of their DAS is the 2010 Census. Since it also respected a set of invariants, the SDC protection provided by the 2010 DAS inherits all the same complications outlined in the previous section. This makes it even more important to handle invariants in a unified way, especially if one wants to compare different invariant-preserving SDC methods from within the same theoretical system---as is one of the main goals of this article.

Like the 1990 and 2000 Censuses, the 2010 DAS primarily consisted of a \emph{data swapping} method \citep{daleniusDataswappingTechniqueDisclosure1982,fienbergDataSwappingVariations2004}.
While data swapping refers to a large class of methods and is widely employed by statistical offices around the world, in the context of the U.S. Decennial Census, swapping works by permuting the geographical data of a randomly selected subset of households \citep{mckennaDisclosureAvoidanceTechniques2018}. 
By definition, swapping keeps invariant all counts that are unaffected by this permutation operation. More specifically, as Section~\ref{sectionSwapping} spells out in detail, the invariants induced by swapping are population totals over various geographic and demographic stratification, whose exact composition are determined by the particular swapping parameters used. %
These invariants 
are inevitably much more numerous than the TDA's---an important observation when comparing data swapping with the TDA %
because, as the number of invariants increases, their impact ranges from negligible to completely nullifying any supposed guarantee of protection. 

Be that as it may, swapping does still provide some protection since it can alter any record's geographic information. According to the standard argument of how swapping provides SDC, it gives plausible deniability as to the reported location of any household, because those households that do in fact have their location changed are randomly sampled and kept secret. In this way, so the argument goes, swapping obfuscates the location, and hence the identity, of all households, thereby making it difficult for an attacker to learn any personal information.
(Appendix~\ref{appendixBackgroundSwapping} will critically examine this argument in detail.)
While swapping has this ad hoc and intuitive justification, %
it has not been given a rigorous theoretical foundation---that is, it has lacked the kind of formal, mathematical guarantees of protection provided by DP 
\citep{christDifferentialPrivacySwapping2022, abowdHowWillStatistical2017, slavkovicStatisticalDataPrivacy2023}.

\subsection{Our Main Contributions}

We have thus come to a puzzling revelation. On the one hand, the new DAS in 2020 was designed according to the rigorous mathematical principle of DP; but it also accommodates invariants, compromising its DP guarantee in ways that are poorly understood and have not been properly mathematically characterized. On the other hand, data swapping---upon which the old DAS was based---was abrogated 
in part because it had not been studied from the perspective of DP and accordingly lacked a DP guarantee. It would be prudent, therefore, to rectify the deficiency in our understanding of the actual DP guarantees for both the 2020 DAS as well as for data swapping.

The current article does exactly this. Casting both the new SDC methods employed in the 2020 DAS and the traditional SDC technique of data swapping within what we call our system of \emph{differential privacy specifications}, we show how each can be formally understood through the lens of DP. In this regard, this article makes four main contributions.%

First, this work presents a formal treatment of invariants in DP (Section~\ref{secInvariants}).  
Acknowledging invariants as essential to the validity, viability, and utility of the Decennial Census and other statistical disseminations, we describe why a data protection standard based only on the realized values of the invariants is not feasible, and why, in fact, a recognition of all potential and counterfactual values is a necessary part of any invariant-accommodating standard, including DP (Section~\ref{secMultiverseNecessary}). %
From this basis, we demonstrate how invariants can naturally be integrated into our system of DP specifications (Sections~\ref{secSystemDPSpecifications} and~\ref{secMultiverseAccommodatesInvariants}) and provide a preliminary, theoretical investigation into the impact of invariants on DP (Section~\ref{secUnderstandingImpactInvariants}). 

Second, this article shows that, contrary to previous views, it is possible to supply data swapping with rigorous guarantees of protection based on DP (Section~\ref{sectionSwapping}). 
To do this, we prove a formal description of the SDC protection---that is, a DP specification---for the Permutation Swapping Algorithm (PSA), which is a data swapping method with similarities to the 2010 DAS (Section~\ref{secPSAIsDP}). Intuitively, this specification can be understood as stating that the PSA satisfies pure DP \citep{dwork2006calibrating} subject to the invariants it induces. While this means the PSA's specification differs from conventional formulations of DP, as we have already argued such a departure is necessary to analyze any invariant-preserving DP method and does not contradict the essentially relative nature of DP. %

We note at the outset that it is impossible to properly assess the actual swapping algorithm used in the 2010 Census from the perspective of DP because the details of that algorithm are not entirely public. Yet we do know that the actual 2010 swapping procedure differs from the PSA in a number of ways (see Appendix~\ref{appendixCompareDAS2010Swapping})---ways that make it difficult to characterize that procedure using DP. Nevertheless, the PSA captures the 2010 DAS in its essence by mimicking its principal design features, and therefore, its DP specification is still informative for understanding the disclosure risk of the 2010 Census.
Indeed, by making the assumption that the PSA---with parameter settings chosen to resemble the 2010 data swapping algorithm---was used as the 2010 DAS, we can (and do) provide a reasonable estimate for the DP specification associated with the 2010 Census (Section~\ref{sec:2010PrivacyLoss}).

Third, this article conducts a comparative analysis of the protection afforded to the 2020 Census with the counterfactual scenario in which the PSA was used to publish the 2020 data (Section~\ref{sectionCensus}). To support this endeavor, we formulate and prove the first DP characterization of the TDA (Section~\ref{sectionCensus-TDA}). Because, much like data swapping, the TDA can only satisfy DP subject to its invariants, this result requires the use of our system of DP specifications to formally incorporate the TDA's invariants into its DP guarantee. In addition to the TDA's specification, we also compile DP specifications for all the primary 2020 Census data products (Section~\ref{sectionSpecifications2020}). Doing so requires identifying the \emph{protection} (or `\emph{privacy}') \emph{units} for the 2020 Census, which we determine to be `post-imputation persons'---a result that is important because, like invariants, post-imputation units negatively impact the actual SDC protection afforded by a DP specification.

We aggregate the DP specifications for the various 2020 publications into a single specification describing the SDC protection afforded to the 2020 Census across all its major publications. This DP specification is then compared to that of a hypothetical application of the PSA to the 2020 Census with various parameter settings (Sections~\ref{sec:whatif} and~\ref{sec:compare-das-psa}). This comparison is informative because it places mathematical summaries of the PSA's and the 2020 Census's SDC protection side-by-side. However, as these two specifications differ on several dimensions, drawing an overall conclusion about the relative strengths of their SDC protection remains difficult. Even so, just articulating the DP specifications for both systems is an important step forward because it clarifies and organizes the specific ways in which their SDC protection diverge.

Fourth, as part of this article's discussion (Section~\ref{sec:discussion}), we put forth a set of open issues with respect to the development of DP and the work needed going forward. These issues center on the difficulty just mentioned: While our system of DP specifications is a useful toolkit for describing the SDC protection of DP implementations, there are currently few tools to effectively compare two DP specifications that differ on multiple dimensions. %
To address this, future research is needed on exploring the trade-off between different dimensions of a DP specification---for example, trading off the presence of additional invariants with a reduced protection loss budget. We propose that such trade-offs might be assessed via their impact on disclosure risk, or via their efficiency against a privacy attack (Section~\ref{sec:compare-das-psa}). Additionally, as we do not attempt in this article to evaluate what substantive SDC protection an invariant-induced DP specification actually provides---this being a difficult question deserving its own dedicated study---we also outline in Section~\ref{sec:discussion} some future directions for understanding and mitigating the impact of invariants on disclosure risk.

Background on data swapping and other related work are provided in Appendices~\ref{appendixBackgroundSwapping} and~\ref{secRelatedWork} respectively. Appendix~\ref{appendix2010Swap} provides the most comprehensive (to the authors' knowledge) publicly available description of the 2010 DAS, along with a comparison between the 2010 DAS and the PSA (Appendix~\ref{appendixCompareDAS2010Swapping}) and a discussion of ways the PSA could be modified to further align with the 2010 DAS while still preserving its DP flavor (Appendix~\ref{appendixModifyPSA2010}).

\section{Invariants and Differential Privacy}\label{secInvariants}

\subsection{Integrating Invariants Into Differential Privacy---Intuitions and Subtleties}\label{secMultiverseNecessary}
Intuitively speaking, the impact of invariants on SDC is similar to conditioning in statistical inference, that is, constraining the possible states of the (confidential) data by known or assumed information. In other words, any invariant-respecting SDC criterion should only consider the counterfactual data sets that share the same values for the invariants as the confidential data set---just as any conditional probability only considers outcomes that agree with the conditioning information. In that sense, the procedure for infusing invariants into a DP formulation parallels the process of %
disintegrating a probability into a collection of conditional distributions.
The overall mathematical notion of a probability---or of DP---remains the same; the difference lies in the state spaces (i.e., the set of possible outcomes, or the set of possible data sets) to which it is applied.
	
At the same time, just as defining conditional distributions brings complications and subtleties (such as conditioning on probability-zero events), defining an invariant-accommodating formulation of DP requires addressing a variety of nuances, many of which are implicit in conventional DP definitions. Without making them explicit, comparing protection loss budgets across different forms of DP---especially those with different invariants---would be as meaningless as comparing the face values of different currencies without considering their conversion rates. Indeed, the problem is worse: the issue is not merely the conversion rates but, more importantly, the realizable purchasing power---that is, how much data privacy each method actually affords in terms of reducing statistical disclosure risk. 

As an example of the subtleties in incorporating invariants into a DP formulation, consider the scenario in which the U.S. population size $N$ is enumerated by the Decennial Census to be exactly 330 million. Once this value is made public, any counterfactual data set with a different value for $N$ becomes immediately distinguishable from the actual confidential census data. Hence, such data sets must be excluded from consideration under the DP paradigm, which relies on the notion of distinguishability, as we shall review later. 

However, it would be rather unwise---regardless of which DP specification we decide to adopt---to implement an SDC method that will satisfy the specification only when $N=330,000,000$. This is because the particular value of the enumerated national total is \textit{accidental}, in the sense that this specific total population count is, but need not be, a property of the Decennial Census \citep[see e.g.][]{sep-essential-accidental}. Even if it were completely accurate (which it never is), it reflects only the count at the time of the census. A DP method needs to work irrespective of the enumerated value of the population total. Likewise, a DP specification should guarantee protection irrespective of this value, or the realized value of any other invariant.

On the other hand, the fact that the national population count and other invariants were published exactly as enumerated is \emph{essential} to the Decennial Census because this occurs by design. This distinction between the essential invariant statistics and their accidental values leads respectively to the notions of the \textit{multiverse} and the \textit{universe}, which will be defined mathematically in Section~\ref{secSystemDPSpecifications}. For the current example, a universe is the collection of all plausible census data sets that share the same specific value for $N$. The multiverse is then the collection of these universes as $N$ varies within a specified range. (Not incidentally, specifying this range for $N$---or more generally, the admissible values of any invariant statistic---is yet another component that a rigorous theoretical formulation of DP must make explicit.)

Invoking the analogy to conditional distributions, one may view the distinction between the multiverse and a universe as analogous to the difference between the conditional probability $\sfP(\, \cdot \mid Y)$, conditioning on a random variable $Y$, and the conditional probability $\sfP(\,\cdot \mid Y=y)$, conditioning on the event $\{Y=y\}$. The former is a collection of probability distributions as $Y$ takes on different values, and the latter is a single probability distribution determined by the particular value $y$. Here the random variable $Y$ is the essential quality, and the event $\{Y=y\}$ is an accidental realization. Similarly, a DP specification should concern the essential nature of a data release mechanism, rather than its properties within some particular, but accidental, universe. 

Yet even setting aside the subtleties of properly accounting for invariants, the plethora of existing DP formulations differ along several other dimensions \citep{desfontainesSoKDifferentialPrivacies2022}. As such, DP can vary widely in form and spirit \citep{dworkDifferentialPrivacyPractice2019}, making it difficult to (1) understand what it means for an SDC method to be DP and (2) objectively compare different DP deployments in a systematic way—two tasks central to the goals of this article.

Therefore, to properly integrate invariants into DP, we must first be transparent about which formulation of DP the invariants will be incorporated into. This need ultimately led us to articulate a system of DP specifications, which is presented in the next section. The phrase ``a refreshment stirred, not shaken'' in our article title is intended to emphasize that this system is not new, but simply a synthesis of the existing literature on the many variations of DP. Indeed, this system is in essence the formalization of three principles that we believe are widely accepted in the DP community, as discussed below.

\subsection{Invariants in Our System of Differential Privacy Specifications}\label{secSystemDPSpecifications}

The first principle states that a DP formulation is a technical standard that requires the rate of change---or `derivative'---of an SDC method to be controlled (hence the epithet `differential'). The second principle asserts that the rate of change of an SDC method is defined as the change in the \emph{probability distribution} of the method's output per unit change in its input data. And the third principle observes that different DP formulations correspond to different choices for how and where to measure these changes, as well as how much to control the associated rate of change. 
        
We call these choices the \emph{building blocks} of DP because each of them formalizes a different dimension of DP and all of them are required to fully define a DP formulation. Since they are essential for establishing our theoretical results concerning the PSA's and the TDA's DP guarantees, we will describe the five building blocks in two ways---mathematically and intuitively:
\begin{itemize}
    \item The \emph{domain} $\X$: a set of data sets $\xd$
    \\ --- \textit{Who} is eligible for protection?
    \item The \emph{multiverse} $\scD \subset 2^{\X}$: a set of universes $\D \subset \X$
    \\ --- \textit{Where} does the protection extend to?  
    \item The \emph{input premetric} $\dX$: a `distance’ between any two data sets $\xd, \xd' \in \X$
    \\ --- \textit{What} is the granularity of protection? 
    \item The \emph{output premetric} $\dPr$: a `distance’ between probability distributions
    \\ --- \textit{How} are changes in output variations measured?
    \item The \emph{protection loss budget} (PLB) $\varepsilon_{\D}$: a function $\scD \to [0,\infty]$
    \\ --- \textit{How much} protection is afforded? 
\end{itemize}

We term a collection of choices for all five of the building blocks a \emph{DP specification}, while we call a collection of choices for the first four a \emph{DP flavor}. The last building block, the PLB, is more commonly known as the `privacy' loss budget, although we eschew this term to avoid implying that the PLB fully captures the complex concept of privacy \citep{nissenbaumPrivacyContextTechnology2010, seemanPrivacyUtilityDifferential2023, benthallIntegratingDifferentialPrivacy2024}. 
There is a subscript $\D$ in $\varepsilon_\D$ because adopting the same DP flavor for different universes does not mean that we must maintain the same PLB across these universes (for example, the USCB could decide that more protection is required for small values of $N$ than large values). Consequently, we allow the value the PLB takes to vary across universes.
The subscript $_\D$ therefore indicates the value of the PLB for the particular universe $\D$. (However, we may drop this subscript in some places for reasons explained later in Remarks~\ref{remarkbFunctionD} and~\ref{remarkRhoSquared}.) An extensive treatment of these notions and their nuances, implications, and applications is presented in \citet{bailieRefreshmentStirredNot2026}. Here, we only focus on aspects essential for stating and proving our main results regarding the PSA and the TDA. 

Through our system of DP specifications, we can unify many (though not all) DP formulations in the literature via the following formalism. 
    
\begin{definition}
A \emph{data release mechanism} (or, synonymously, an SDC method) $T$ is a function that takes as input a (confidential) data set $\xd \in \X$ and a random seed $U \in \U$, and outputs some noisy statistics $T(\xd, U)$ based on $\xd$. 

A mechanism $T$ \emph{satisfies the DP specification} \DPSpec\ if, for all universes $\D \in \scD$ and all pairs of data sets $\xd, \xd' \in \D$,
\begin{equation}\label{eqTSatisfiesDPDefn}
    \dPr \Big[ \sfP \big(T(\xd, U) \in \cdot\ \big), \sfP \big(T(\xd', U) \in \cdot\ \big) \Big] \le \epsilon_{\mathcal D} d_{\mathcal X} \big(\xd, \xd'\big),
\end{equation}
where $\sfP$ is the probability distribution induced by the random seed $U \in \U$, taking the input data set ($\xd$ or $\xd'$) as fixed.
\end{definition}

Under this system, there are two ways that invariants can naturally be integrated with DP. First, one can set $\dX(\xd, \xd')=\infty$ whenever $\xd$ and $\xd'$ disagree on the invariants. Second, invariants may be encoded through the multiverse $\scD$. We prefer the second approach for its better interpretability. The input premetric $\dX$ is already used to specify the \emph{units} (for example, people, households, or businesses) to which a DP specification provides protection. The units of a DP specification are those entities whose records differ between data sets $\xd$ and $\xd'$ with $\dX(\xd, \xd') = 1$ \citep[see][]{bailieRefreshmentStirredNot2026}. Overloading $\dX$ to also encode the invariants breaks the connection between a DP specification’s units and its input premetric and leads to confusion when both the units and the invariants are nontrivial---for example, as may occur in survey statistics \citep{drechslerWhoseDataIt2024}.

Nevertheless, these two approaches---incorporating invariants via $\dX$ or via $\scD$---have the same result: Both ensure that the DP condition (Equation~\ref{eqTSatisfiesDPDefn}) only bounds the `distance’ between the distributions $\sfP \big(T(\xd, U) \in \cdot \big)$ and $\sfP \big(T(\xd', U) \in \cdot \big)$ when the two counterfactual data sets $\xd$ and $\xd'$ agree on the invariants (see Section~\ref{secMultiverseAccommodatesInvariants}). Intuitively, this corresponds to conditioning on the invariants, except that a priori we do not know the realized value of the invariants. As discussed in Section~\ref{secMultiverseNecessary}, the DP specification must therefore account for all possible values through the multiverse $\scD$, rather than conditioning on any single particular value of the invariants (which would correspond to using a single universe $\D$). 

We show in Section~\ref{secUnderstandingImpactInvariants} that weakening the DP condition via a nonvacuous multiverse $\scD$—as happens whenever there are invariants—leads to a reduction in the actual protection guaranteed by a DP specification. (By `nonvacuous multiverse,' we mean one that is not equivalent to the multiverse $\scD = \{\X\}$.) However, this complication is necessary in many real-world applications of DP. In addition to examples from the literature outlined in \citet{bailieRefreshmentStirredNot2026}, we prove in Section~\ref{sectionCensus-TDA} that a nonvacuous multiverse is required to describe the DP protection provided to the 2020 U.S. Census. Furthermore, the practice of empirically restricting the data universe is typical in statistical disclosure control and in data analysis more broadly. Top-coding—setting a maximum limit on a continuous variable, usually after looking at the raw data—is one common example.

Before we proceed, we pause to address the astute reader who may wonder why we do not adopt a third, seemingly simpler approach to incorporating invariants into DP: Take an existing DP formulation and restrict its neighboring data sets to those that agree on the invariants. (For readers who are unfamiliar with the formulation of DP in terms of neighboring data sets, $\xd$ and $\xd'$ are called neighbors if $\dX(\xd, \xd') = 1$; DP is then defined as the requirement that $\dPr [ \sfP (T(\xd, U) \in \cdot\ ), \sfP (T(\xd', U) \in \cdot\ ) ] \le \epsilon_{\D}$ holds for all neighbors $\xd, \xd' \in \D$.) The problem with this approach is that, after excluding neighboring data sets that have different values for the invariants, there may be no neighbors left, resulting in a completely vacuous formulation of DP. Indeed, we will see that this is what happens with data swapping. This is one reason to replace the concept of neighboring data sets with the more general notion of an input premetric, in addition to being the reason not to encode invariants through neighboring data sets.

\subsection{How the Multiverse Accommodates Invariants}\label{secMultiverseAccommodatesInvariants}

As we illustrated in Section~\ref{secMultiverseNecessary}, the multiverse $\scD$ of a DP flavor \DPFlav\ is a collection of subsets of the domain $\mathcal X$, and these subsets are called the universes of the DP flavor. %
The idea behind the concept of a universe $\D$ is that any two data sets $\xd, \xd' \in \D$ should be `mutually plausible,' in the sense that an attacker, upon observing the released statistics, should not be able to determine (with certainty and correctly) whether the true confidential data set is $\xd$ or whether it is $\xd'$. %
(Here, and in the next paragraph, we assume for simplicity that $\dX(\xd, \xd')$ and $\dX(\xd', \xd)$ are both finite for every $\xd, \xd' \in \D$.)

In practice, it is often the case that the multiverse $\mathscr D$ is defined via a set-valued function that we call the \emph{universe function} $\Df(\cdot) : \mathcal X \to 2^{\mathcal X}$, that associates every potential data set $\xd \in \mathcal X$ with a universe $\Df(\xd) \subset \mathcal X$. %
In this case, the multiverse $\scD = \{\Df(\xd)\}_{\xd \in \mathcal X}$ is the image of this universe function $\Df$. 
The resulting DP specification would then ensure that every $\xd' \in \D (\xd)$ is plausible if $\xd$ is also plausible---that is, an attacker would not be able to distinguish with certainty the true confidential data set to be $\xd'$ and not $\xd$, for any $\xd' \in \D(\xd)$.%

A set of invariants can be encoded by a universe function of the form: %
\begin{equation}\label{eq:data-universe-invariants}
    \mathcal D_{\bm c}(\xd) = \Big\{ \xd' \in \mathcal X: \bm c(\xd') = \bm c(\xd) \Big\},
\end{equation}
for a given deterministic function $\bm c: \mathcal X \to \mathbb R^k$. Here the function $\bm c$ describes the features $\bm c(\xd)$ of the data set $\xd$, which are taken to be invariant. For example, $\bm c(\xd)$ could be the state population totals calculated from the data set $\xd$ of census responses. We call $\mathcal D_{\bm c}(\cdot)$ the \emph{invariant-induced universe function} and its image $\{\mathcal D_{\bm c}(\xd)\}_{\xd \in \mathcal X}$ the \emph{invariant-induced multiverse} $\mathscr D_{\bm c}$.

By construction, the invariants take the same values across every data set in a universe $\D_{\bm c}(\xd)$. As such, invariants give rise to an equivalence relation $\sim$ over the domain $\mathcal X$, which is defined by $\xd \sim \xd'$ if $\bm c(\xd) = \bm c(\xd')$. The data universe function  therefore induces a \emph{partition} of $\mathcal X$ indexed by the image of the invariant function $\bm c$, splitting $\X$ into universes of mutually plausible data sets that share the same values for the invariants.

\begin{example}\label{exContingencyTable}
Let the data set be an contingency table of $m\times n$ records taking non-negative integer values:  $\mathcal X = \mathbb N^{m\times n}$. Suppose the function 
$\bm c: \mathbb N^{m\times n} \to \mathbb N^{m + n}$ tabulates the  row- and column-margins:
\begin{equation}\label{eqMarginInvariants}
    \bm c (\xd) = \bigg(\sum_{i=1}^{m}x_{i1},\ldots,\sum_{i=1}^{m}x_{in}, \sum_{j=1}^{n}x_{1j},\ldots, \sum_{j=1}^{n}x_{mj}\bigg).
\end{equation}
The data curator may treat the row and column margins of the confidential data set as invariant \citep[see][and references therein]{dobraBoundsCellEntries2000}. This would be equivalent to employing the universe function $\mathcal D_{\bm c}(\cdot)$ as defined in Equation~\ref{eq:data-universe-invariants} using the function $\bm c$ from Equation~\ref{eqMarginInvariants}, since $\mathcal D_{\bm c}(\cdot)$ ensures that only pairs of data sets $\xd, \xd'$ with the same row and column margins are subject to the DP condition (Equation~\ref{eqTSatisfiesDPDefn}).
\end{example}

More generally, the DP condition is required only for data sets that belong to the same universe $\D \in \scD$. In the case of an invariant-induced multiverse $\scD_{\bm c}$, this means 
the DP condition will only be applied to data sets $\xd$ and $\xd'$, which share the same values for the invariants. 
As we will prove in Section~\ref{secUnderstandingImpactInvariants}, this allows the invariants to be released as is, without contributing to the PLB.

In some applications, there are also inequality invariants \citep{abowd2022topdown}. As an example of such an invariant, the 2020 U.S. Decennial Census requires that the reported number of group quarters in any geographical unit is no larger than the number of persons in that unit. More generally, an inequality invariant is of the form $f(\xd) \le 0$ for some function $f : \mathcal X \to \mathbb R$. Such an invariant can be incorporated in the above framework by defining 
\begin{equation}\label{eqInequalityInvariant}
    c(\xd) = \begin{cases}
    1 & \mathrm{if\ } f(\xd) \le 0, \\
    0 & \mathrm{if\ } f(\xd) >0 .
    \end{cases}
\end{equation}

\subsection{Understanding the Impact of Invariants on Differential Privacy}\label{secUnderstandingImpactInvariants}

As we have repeatedly emphasized, invariants reduce the SDC protection provided by a DP specification because they restrict where the DP condition (Equation~\ref{eqTSatisfiesDPDefn}) must hold. This does not mean that invariants are antithetical to the fundamental idea of DP. %
To assert otherwise would impose an unduly restrictive interpretation of DP, one that would not only greatly reduce its applicability, but would also categorize the original formulation of DP as not satisfying DP. Indeed, this formulation \citep[given in][]{dwork2006calibrating} uses the Hamming distance as its input premetric, as does any formulation of \emph{bounded} DP more generally. Because the Hamming distance between two data sets is infinite whenever they have a different number of records, these formulations only require the DP condition to hold for data sets of the same size. Thus, having the data set size as an invariant has been a part of DP from its conception. 

It should be clear then that the mere presence of invariants is not necessarily an issue; what matters is the extent that they impact actual SDC protection. 
To this end, this section provides some preliminary, theoretical results on invariants' effect. (A broader discussion on their impact is provided in Section~\ref{secImpactInvariants}.) Since releasing the invariants without modification was the aim of restricting the DP condition in the first place,
we begin by demonstrating the necessity and sufficiency of the invariant-induced multiverse to achieving this aim. %
(All results in this section are proved in Appendix~\ref{secInvariantsProofs}.)

\begin{proposition}\label{propReleaseF}
Fix a domain $\X$ and some invariants $\bm c : \X \to \mathbb R^k$. For any $\dX$ and $\dPr$, the mechanism $T(\xd) = \bm c(\xd)$ that releases the invariants exactly satisfies $\epsilon_{\D}$-DP$(\X, \scD_{\bm c}, \dX, \dPr)$ with PLB $\epsilon_{\D} = 0$ for all $\D \in \scD_{\bm c}$.

Now suppose $\dPr(\sfP, \sfQ) = \infty$ whenever the total variation distance between $\sfP$ and $\sfQ$ is one. %
(This assumption is satisfied by most common choices of $\dPr$.) 
Let $\scD$ be a multiverse such that there exists some universe $\D_0 \in \scD$ and some $\xd, \xd' \in \D_0$ with $\dX(\xd, \xd') < \infty$ and $\bm c(\xd) \ne \bm c(\xd')$. Then $T$ does not satisfy $\epsilon_{\D}$-DP$(\X, \scD, \dX, \dPr)$ whenever $\epsilon_{\D_0} < \infty$.
\end{proposition}

The first part of Proposition~\ref{propReleaseF} shows that the invariants can be released without modification `for free'---that is, at no cost to the PLB. Hence, the invariant-induced multiverse is sufficient for achieving the aim of accommodating invariants. This result also holds if $\mathscr D$ is any multiverse with $\bm c$ constant within every universe $\D \in \scD$ (i.e., if $\bm c(\xd) = \bm c(\xd')$ for all $\xd, \xd' \in \D$ and all $\D \in \scD$). 

The second part of Proposition~\ref{propReleaseF} demonstrates the necessity of the invariant-induced multiverse. It shows that a set of invariants cannot be published as is by a DP mechanism whenever $\scD$ does not conform to these invariants.
More precisely, if there are data sets $\xd, \xd' \in \mathcal D_0$, which take different values on the invariants (and satisfy $\dX(\xd, \xd') < \infty$), then releasing the invariants exactly would require $\epsilon_{\mathcal D_0} = \infty$.

The following result demonstrates that the invariant-induced multiverse does not allow for the release of additional information, above and beyond the information contained in the invariants, without incurring some loss of protection. This is important because while we want to allow the invariants to be released as is, we do not want this to imply that other information can also be released for free.

\begin{proposition}\label{propReleaseFConverse}
    Suppose that a data release mechanism $T$ varies within some universe $\D_0 \in \scD$ in the sense that there exists $\xd, \xd' \in \D_0$ with $\dX(\xd, \xd') < \infty$ but {$\sfP (T(\xd, U) \in \cdot\ ) \ne \sfP (T(\xd', U) \in \cdot\ )$}. When $\dPr$ is a metric, $T$ satisfies $\epsilon_{\D}$-DP$(\X, \scD, \dX, \dPr)$ only if $\epsilon_{\D_0} > 0$.
\end{proposition}

A mechanism $T$ varies within a universe $\D_0 \in \scD_{\bm c}$ whenever it includes information that is not logically equivalent to the invariants $\bm c$. In this case, Proposition~\ref{propReleaseFConverse} shows that releasing this information will require a nonzero PLB, even while the invariants can be released for free.

Beyond describing exactly what information can be released under an invariant-induced multiverse, the above two propositions more generally illustrate that the interpretation of the PLB cannot be isolated from the multiverse, and indeed this complicates the comparison of budgets across different DP flavors. %

As a concrete example of how the meaning of the PLB $\epsilon_{\mathcal D}$ changes with $\scD$, consider evaluating the same mechanism $T$ against two DP flavors $(\X, \scD_{\bm c}, \dX, \dPr)$ and $(\X, \scD_{\bm c'}, \dX, \dPr)$, which differ only on their invariants. Suppose the second set of invariants are nested within the first; that is, $\bm c$ is strictly more constraining than $\bm c'$. (For example, $\bm c$ are population counts at the block level and $\bm c'$ are counts at the county level.) Then Proposition~\ref{propNestedscD2} below proves that $T$'s protection loss under $\scD_{\bm c'}$ cannot be smaller and may be strictly larger than its loss under $\scD_{\bm c}$.%

By comparing $T$'s two PLBs on their own, a reader might arrive at the seemingly paradoxical conclusion: we can increase the protection provided by $T$ simply by increasing the number of invariants allowed by our DP flavor. Yet this ignores the critical fact that the PLB is only a `within-system' evaluation of SDC. It is not an absolute, unitless measure of protection, but only a relative measure whose units are determined by the DP flavor. (See \cite{bailieRefreshmentStirredNot2024c}, for an explanation of how the units of the PLB depend on the other three components of the DP flavor, not just $\scD$.) It is dangerous therefore to think that the ${\bm c}$-release is actually afforded with less SDC protection than the ${\bm c'}$-release because there is privacy leakage due to specifying additional invariants, which is not captured by the within-system evaluation $\epsilon_{\D}$. Indeed in the extreme example where $\bm c$ is an injective function so that the universes $\mathcal D$ are singletons, %
the entire data set can be published, and hence there is no actual protection afforded by the DP specification $\epsilon_{\D}$-DP$(\X, \scD_{\bm c}, \dX, \dPr)$ even though in such cases $\epsilon_{\D}$ can be set to zero. This point is crucial to understanding the comparative analysis between the PSA and the 2020 Census data releases as presented in Section~\ref{sectionCensus}.

\begin{proposition}\label{propNestedscD2}
    Suppose that $\scD$ is refinement of $\scD'$---that is, for all $\D \in \scD$, there exists some $\D' \in \scD'$ such that $\D \subset \D'$. Define the protection loss of a data release mechanism $T$ under the flavor \DPFlav\ as $\PLD = \inf \{ \epsilon_{\D} : T \mathrm{\ satisfies\ } \DPSpecInMath \}$. Similarly define $\PLDdash = \inf \{ \epsilon'_{\D'} : T \mathrm{\ satisfies\ } \DPSpecDiffMultiverseInMath \}$. Then, the protection loss $\PLDdash$ under \DPFlavDiffMultiverse\ is no smaller than the loss $\PLD$ under \DPFlav:
    \[\PLDdash \ge \PLD,\]
    for any $\D \in \scD$ and $\D' \in \scD'$ with $\D \subset \D'$.
\end{proposition}

This proposition is an initial tool for comparing DP specifications, which vary on their invariants. The two sets of invariants in this result have to take a very particular form---one has to be nested inside the other. In this case, we can order the two sets of invariants by their impact on SDC protection. %
Adding invariants reduces the quality of a DP flavor, because it refines the multiverse from $\scD'$ to $\scD$, and in so doing shrinks the scope of protection and `waters down' the DP flavor. Moreover, the new, lower quality DP flavor brings about an apparent `saving' of protection loss---a change from  $L'_{\D'}$ to $ L_{\D}$. Section~\ref{sec:whatif} dissects this phenomenon using a hypothetical swapping analysis of the 2020 Census. %
As a general matter, comparing the relative impact of non-nested sets of invariants is a much more difficult task because it requires weighting the disclosiveness of different invariants against each other. We suggest ways to approach this more difficult task in Section~\ref{secImpactInvariants}.

\section{A Differentially Private Analysis of Data Swapping}\label{sectionSwapping}

\subsection{Data Swapping}

Given a data set $\xd$, partition its set of variables $\Vars$ into two nonempty subsets: the \emph{swapping variables} $\Vswap$ and the \emph{holding variables} $\Vhold$. A data swapping method randomly selects some records of $\xd$ and interchanges the values of their swapping variables $\Vswap$, while leaving their $\Vhold$ unchanged. %
This creates a new data set consisting of individual records whose $\Vhold$  values are as originally observed and whose $\Vswap$ values are possibly different. %
The exact procedure for selecting records and interchanging their $\Vswap$ values varies between different data swapping methods. %
(References to literature on various swapping methods can be found in Appendix~\ref{appendixBackgroundSwapping}, which also summarizes the use of data swapping by national statistical offices across the world.)

Sometimes, swapping is restricted to records that share the same values on a subset of the holding variables $\Vhold$, called the \emph{matching variables} $\Vmat$. More exactly, whenever $\Vmat$ is nonempty, records are partitioned into strata according to their $\Vmat$ values and data swapping is repeated independently within each stratum. Also referred to as the \emph{swap key} \citep{mckennaDisclosureAvoidanceTechniques2018, abowd2023confidentiality}, the matching variables often capture important demographic information that the data custodian would like to preserve. (Because swapping records with different $\Vmat$ values is prohibited, this information is indeed preserved by data swapping.) %

\begin{example}\label{exampleUSCensus2010DAS}
	This example is a simplification of the DAS for the 2010 U.S. Decennial Census. Represent the 2010 Census data as a list of household records, whose variables include all the household's characteristics, as well as the questionnaire responses from each individual associated with that household. The matching variables $\Vmat$ (i.e., swap key) include both the number of voting age persons and the total number of persons in the household. $\Vmat$ also includes a geographic variable %
	$V_g$ (see \cite{uscensusbureauGuidanceGeographyUsers2021}), either the census tract, county, or state of the household. (To the best of our knowledge, the exact choice of $V_g$ has never been made public by the USCB.) $\Vswap$ are the geographic variables nested underneath $V_g$.
	For example if $V_g$ is the county, then $\Vswap$ is the block and tract of the household. All other variables belong to $\Vhold$---in particular, the household and person characteristics. One can imagine the 2010 DAS as digging up pairs of houses of the same size in the same geographic area and swapping their locations but not changing the houses and their occupants.
	In the 2010 DAS, each household is assigned a risk score based on the USCB's assessment of how unique the household is within its neighborhood. These risk scores are used to compute each household's probability of being swapped. Every (non-imputed) household has a nonzero swap probability. Selected households are then swapped with one of their neighbors. (See Appendix~\ref{appendix2010Swap} for a detailed description of the 2010 DAS and references for this information.)
\end{example}

\subsection{What Invariants Does Swapping Preserve?}\label{secSwapInvariants}

Swapping is, very loosely, a synthetic data generation mechanism. Given a data set $\xd$ as input, 
swapping produces a `privacy enhanced' version $\bm Z$ of $\xd$. Both $\xd$ and $\bm Z$ contain the same variables 
as well as the same number of records. Hence, the invariants of swapping are determined by examining what swapping does, and does not, change in the data.

Consider the data set $\xd$ as a matrix whose rows correspond to the records of $\xd$ and whose columns correspond to the variables $\Vars$ of $\xd$. %
Without loss of generality, the holding variables are ordered before the swapping variables so that $\xd$ can be partitioned as $[\xdHold, \xdSwap]$. 
A swapping algorithm randomly selects a permutation $\sigma$ of the rows of $\xd$ and interchanges the rows of the matrix $\xdSwap$ according to $\sigma$. This operation yields $\xdSwap^\sigma$, whose $i$th row is given by the $\sigma(i)$-th row of $\xdSwap$. This defines the swapped data set $\bm Z$ as the matrix $[\xdHold, \xdSwap^\sigma]$, and
the swapping mechanism releases as its output the fully saturated contingency table generated by $\bm Z$. %

One can see that after swapping, 
any statistic generated by only the matrix $\xdHold$ is invariant.
Moreover, since $\Vmat$ is identical among swapped records, any statistic generated by only $\xdMat$ and  $\xdSwap$ is also preserved by swapping. Only statistics that depend nontrivially on both variables $\Vswap$ and $\Vhold \setminus \Vmat$ can be altered by swapping.

\begin{proposition}\label{propInvariantsOfSwap}
	Suppose that $\Vhold \setminus \Vmat$ and $\Vswap$ are nonempty. Then, 
	without loss of generality, we may assume that each of $\Vmat, \Vswap$, and $\Vhold \setminus \Vmat$ are univariate. %
	Denote a value of the matching variable $\Vmat$ by $m$. Similarly, let $h$ and $s$ be values of $\Vhold \setminus \Vmat$ and $\Vswap$ respectively.
	
	Disregarding the ordering of records, the data set $\xd$ can be represented as a 3-dimensional contingency table $H(\xd) = [n_{mhs}^{\xd}]$ of counts in each combination of possible values for $m$, $h$, and $s$. (We will omit the superscript $^{\xd}$ when it is clear from the context.) In general, interior cell counts $n_{mhs}$ are not preserved under swapping and neither are the margins $n_{\cdot hs} = \sum_{m} n_{mhs}$. But swapping does keep $n_{m\cdot s} = \sum_h n_{mhs}$ and $n_{mh\cdot} = \sum_s n_{mhs}$ invariant.
\end{proposition}	

In other words, there are two contingency tables that remain unchanged by swapping: 1) $\Vmat \times \Vswap$: the cross-classification of the matching variables by the swapping variables; and 2) $\Vhold$: the cross-classification of all the holding variables; while the interior of the contingency table, $(\Vhold\setminus\Vmat) \times \Vswap$ is perturbed by swapping. (Here, as elsewhere in this work, ``$\times$'' and ``$\setminus$'' are respectively the Cartesian product and set difference operators.)

\begin{proof}
	First we justify why we can assume that $\Vmat$, $\Vswap$ and $\Vhold \setminus \Vmat$ are univariate (i.e., that these variable sets are singletons). If $\Vmat$ is empty, replace it with a set consisting of a new variable taking the same value on every record. And if either of $\Vmat$, $\Vswap$ or $\Vhold \setminus \Vmat$ has more than one variable, then cross-classify these variables into a single variable. Neither of these two operations will change the behavior of a swapping method, %
	so we may use them to ensure $\Vmat, \Vswap$ and $\Vhold \setminus \Vmat$ are univariate.
	
	Since every permutation $\sigma$ can be written as the composition of swaps (i.e., two-cycles), it suffices to show that all possible swaps preserve $n_{m\cdot s}$ and $n_{mh\cdot}$ but not necessarily $n_{\cdot hs}$. A swap pairs a record $a$ in categories ${mhs}$ with a record $b$ in ${mh's'}$. It moves $a$ to $mh's$ and $b$ to $mhs'$. The matching category $m$ is the same in $a$ and $b$ by construction. %
	Unless $m = m'$ or $s = s'$, after the swap $n_{mhs}$ and $n_{mh's'}$ decrease by one, and $n_{mh's}$ and $n_{mhs'}$ increase by one. Hence, $n_{m\cdot s}$ and $n_{mh\cdot}$ remain unchanged but $n_{\cdot hs}$ changes whenever $h \ne h'$ and $s \ne s'$.
\end{proof}

\begin{continueexample}{exampleUSCensus2010DAS}
	In the 2010 U.S. Census DAS, the number of adults, children, and households in each block are invariant. (This is the $n_{m\cdot s}$ margin.)  The counts of all the person and household characteristics inside each $V_g$ are also invariant. (This is the $n_{mh \cdot}$ margin.) For example, if $V_g$ is the county, then the aggregate characteristics at the county level remain unchanged by swapping, but these aggregates at the block and tract level are perturbed.
\end{continueexample}

\begin{definition}%
	Under the setup of Proposition \ref{propInvariantsOfSwap}, define the \emph{swapping invariants} $\cswap(\xd)$ for a given choice of $\Vmat, \Vswap$, and $\Vhold$ as the vector of all margins $n_{mh\cdot}$ and $n_{m\cdot s}$, for all possible values of $m,h$ and $s$. For example, if $\Vmat$, $\Vhold \setminus \Vmat$, and $\Vswap$ take values in $\{1,\ldots, \mathcal M\}$, $\{1, \ldots, \mathcal H\}$, and $\{1,\ldots,\mathcal S\}$ respectively, %
	then
	\[\cswap (\xd) = \begin{bmatrix}
		n_{11\cdot} & n_{12\cdot} & \cdots & n_{\mathcal M\mathcal H\cdot} & n_{1\cdot1} & n_{1\cdot2} & \cdots & n_{\mathcal M\cdot \mathcal S}
	\end{bmatrix}^{\mathsf{T}}.\]
\end{definition}

As the following example illustrates, we do not have complete flexibility in choosing the invariants of swapping.

\begin{example}\label{exTDAInvariantsNotSwapInvariants}
	In the 2020 TDA, there are three invariants: 1) the number of people in each state; 2) the number of housing units in each block; and 3) the count of each type of occupied group quarters (e.g., residence halls, nursing facilities, prisons) in each block \citep{u.s.censusbureauDisclosureAvoidance20202021}. We cannot design a swapping algorithm that preserves these---and only these---invariants. In other words, the 2020 U.S. Census invariants do not correspond to any swapping invariants $\cswap$, regardless of the choice of $\Vmat, \Vswap$ and $\Vhold$. Why? Swapping always preserves the one-dimensional marginals: $n_{m \cdot \cdot}, n_{\cdot h \cdot}$ and $n_{\cdot \cdot s}$; but the 2020 DAS does not. For example, the number of 25- to 34-year-old people in the United States is not invariant under the 2020 TDA, but it must necessarily be invariant under any swapping algorithm. %
\end{example}

\subsection{Permutation Swapping Satisfies Pure Differential Privacy Subject to Its Invariants}\label{secPSAIsDP}

In this section, we design a specific data swapping algorithm---called the \emph{Permutation Swapping Algorithm} (PSA) to distinguish it from other data swapping methods---which satisfies the DP flavor \FlavSwapGen. Here $\X$ denotes any set of data sets that all have the same common set of variables, and $\dHamS^r$ denotes the Hamming distance at the resolution $r$ of the PSA's swapping procedure. That is, $\dHamS^r(\xd, \xd')$ denotes the number of records that differ between data sets $\xd$ and $\xd'$ (ignoring the ordering of the records in $\xd$ and $\xd'$):
\begin{equation}\label{eqHamSDistDefn}
   \dHamS^r (\xd, \xd') = \begin{cases}
      \frac{1}{2} \abs{ \xd \ominus \xd'} & \text{if } \abs{\xd} = \abs{\xd'}, \\
     \infty & \text{otherwise},
    \end{cases}
\end{equation} 
where $\ominus$ is the symmetric multiset difference and the resolution $r$ is the type of record in consideration. For example, if the PSA swapped records that correspond to individual persons, then the input premetric of the PSA's DP flavor would be the Hamming distance $\dHamS^p$ on person-records. Alternatively, the PSA could swap household-records, in which case its input premetric would be the Hamming distance $\dHamSh$ at the resolution of households. (This distinction will become important when we compare the PSA with the TDA.) The output premetric of the PSA's DP flavor is the multiplicative distance, which corresponds to the notion of %
pure DP \citep{dwork2006calibrating} and is defined as:
\begin{equation}\label{eqMultDistDefnPart2}
	\Mult(\sfP, \sfQ) = \sup_{E \in \mathscr F} \abs{ \ln \frac{\sfP(E)}{\sfQ(E)}},
\end{equation}
for two probabilities $\sfP$ and $\sfQ$ on the measurable space $\T$ with $\sigma$-algebra $\mathscr F$.

While the PSA was not used in 2010, %
a specific instantiation of it does reflect the essential features of the 2010 DAS's data swapping algorithm (Section~\ref{sec:2010PrivacyLoss}). However, certain aspects of the PSA were made with the specific goal of satisfying DP. For example, a swapping method cannot satisfy \FlavSwapGen\ if the number of swaps it makes is fixed. (To be clear, based on the available public information, we do not believe the 2010 DAS fixes the number of swaps, although it does appear to control this number to some degree.) To see this, suppose that a possible output data set $\bm z$ differs from $\xd \in \D_0$ by $m$ swaps and from $\xd' \in \D_0$ by $m+1$ swaps. If the swapping methods allows a maximum of $m$ swaps, %
then $\bm z$ has nonzero probability given $\xd$ as input but zero probability given $\xd'$, thereby violating \SpecSwapGen\ for any finite $\epsilon_{\mathcal D_0}$. More generally, a necessary condition for a swapping method to satisfy \SpecSwapGen\ for finite $\epsilon_{\D_0}$ is that, given input $\xd \in \D_0$, any data set $\xd' \in \D_0$ has a nonzero probability of being outputted (up to reordering of the rows of $\xd'$).

To ensure this condition, rather than swapping rows of $\xdSwap$ in the same matching category $m$, the PSA instead randomly permutes these rows, a type of data swapping method introduced in \citet{depersioNcycleSwappingAmerican2012}. %
Since we do not want to permute every row of $\xdSwap$, rows are randomly selected, independently with probability $p$, and only these selected rows are shuffled. Or, more accurately, after selecting rows of $\xdSwap$ with matching value $m$, the PSA samples uniformly at random a permutation $\sigma_m : \{1, \ldots, n_{\cdot\cdot\cdot}\} \to \{1, \ldots, n_{\cdot\cdot\cdot}\}$, which fixes nonselected rows (i.e., $\sigma_m(i) = i$ for all nonselected $i$), and deranges selected rows (i.e., $\sigma_m(i) \ne i$ for all selected $i$).
This process is repeated for all values of $m$ so that the final data set, after all permutations have been applied, is given by $\bm Z = [\xdHold, \xdSwap^\sigma]$, where $\sigma$ is defined by $\sigma(i) = \sigma_m(i)$ for record $i$ with matching category $m$. In the case that only one record was selected, there are no possible $\sigma_m$ and so records are reselected. Hence, the probability that a record with matching category $m$ is swapped is $$p \sum_{j=1}^{n_{m \cdot \cdot}-1} \binom{n_{m \cdot \cdot} - 1}{j} p^j (1-p)^{n_{m \cdot \cdot}-1-j}.$$ When $n_{m \cdot \cdot} \gg 1$, the expected fraction of records that will have their swapping variables interchanged is approximately $p$. For this reason, we call $p$ the swap rate.%

Pseudocode for the PSA is provided in Algorithm \ref{algoPermute}. The output is a fully saturated contingency table $C(\bm Z) = [n_{mhs}^{\bm Z}]$ (i.e., a 3-way tensor) computed on the swapped data set $\bm Z$. When $\Vmat$, $\Vhold\setminus\Vmat$, and $\Vswap$ all take a finite number of values, $C(\bm Z) = [n_{mhs}^{\bm Z}]$ is a collection of $\mathcal M$ matrices $C_m(\bm Z)=[n_{mhs}^{\bm Z}]$, for $m=1, \ldots, \mathcal M$, each of which has dimension $\mathcal H \times \mathcal S$. This contingency table $C(\bm Z)$ fully determines $\bm Z$ up to reordering of the rows of $\bm Z$.

\begin{theorem}\label{thmSwapDP}
	Suppose the domain $\X$ is such that every data set $\xd \in \X$ shares the same common set $\Vars$ of %
	variables. Partition $\Vars$ into %
	swapping variables $\Vswap$ and holding variables $\Vhold$, and let $\Vmat \subset \Vhold$ be the (possibly empty) set of matching variables.
    Let $\mathcal B$ be the set of matching strata $m$ for which there exist at least two distinct records in stratum $m$. If $\mathcal B$ is not empty, define $b = \max_{m \in \mathcal B} n_{m \cdot \cdot}$; otherwise define $b = 0$.
    
	Suppose that the PSA (Algorithm \ref{algoPermute}) permutes records at resolution $r$. Then it satisfies \SpecSwapGen\ with
	\begin{equation}\label{eqMainTheorem}
		\epsilon_{\mathcal D} = \begin{cases}
			0 & \mathrm{if\ } b = 0, \\
			\ln (b + 1) - \ln o & \mathrm{if\ } 0 < p \le \frac{\sqrt{b+1}}{\sqrt{b+1}+1} \mathrm{\ and\ } b > 0, \\
			\ln o & \mathrm{if\ } \frac{\sqrt{b+1}}{\sqrt{b+1}+1} \le p < 1 \mathrm{\ and\ } b > 0, \\
			\infty & \mathrm{if\ } p \in \{0,1\} \mathrm{\ and\ } b > 0,
		\end{cases}
	\end{equation}
	where $o = p/(1-p)$.
\end{theorem}

\begin{algorithm}[ht]
	\caption{The Permutation Swapping Algorithm (PSA)}%
\label{algoPermute}
\algorithmicrequire A data set $\xd \in \X$ whose set of variables $\Vars$ is partitioned into holding variables $\Vhold$ and swapping variables $\Vswap$; along with a set of matching variables $\Vmat \subset \Vhold$ that define the matching strata. \\
\begin{algorithmic}[1]
	\STATE Set $\bm Z \gets \xd$
	\FORALL{%
		matching strata $m$}
	\IF{$n_{m\cdot\cdot} = 0$ or $n_{m\cdot\cdot} = 1$}
	\STATE \CONTINUE
	\ENDIF
	\FORALL{records $i$ in stratum $m$}\label{lineSelectRecords}
	\STATE Select $i$ with probability $p$
	\ENDFOR
	\IF{0 records selected}
	\STATE \CONTINUE
	\ELSIF{exactly 1 record selected}
	\STATE Deselect all records
	\STATE \GOTO line \ref{lineSelectRecords}
	\ENDIF
	\STATE Sample uniformly at random a permutation $\sigma_m$, which fixes the unselected records and deranges the selected records
	\STATE \COMMENT{Permute the swapping variables according to $\sigma_m$:}
	\STATE $\bm Z \gets [\bm Z_{\mathrm{Hold}}, \bm Z^{\sigma_m}_{\mathrm{Swap}}]$
	\STATE Deselect all records
	\ENDFOR
	\RETURN the fully saturated contingency table $C(\bm Z)$%
\end{algorithmic} 	
\end{algorithm}

It is worth noting that the monotonic increase of $\epsilon_{\mathcal D}$ with $b$ may seem counterintuitive, until one realizes that the PLB quantified in Theorem~\ref{thmSwapDP} does not include the loss due to the invariants themselves. In other words, the more invariants the PSA imposes---which tends to lead to smaller $b$---the less information there is left for the PSA to protect, and hence it is easier to achieve smaller $\epsilon_{\mathcal D}$. This phenomenon is not unique to the PSA, but reflects the fundamentally \textit{relative} nature of DP. %

A proof of Theorem~\ref{thmSwapDP} is presented in Appendix~\ref{sectionAppendixProofMainTheorem}. Here we give a broad sketch for the case $0 < p \le 0.5$ %
and $b > 0$. Because $\sqrt{b+1}/(\sqrt{b+1}+1) > 0.5$, we need to show, for fixed data sets $\xd, \xd'$, and $\bm z$ in the same universe $\mathcal D \in \scD_{\cswap}$, that the budget $\epsilon_{\D} = \ln (b+1) - \ln o$ satisfies the inequality
\begin{equation}\label{eqSketchProofWTS}
\sfP [C([\xdHold, \xdSwap^\sigma]) = C(\bm z)] \le \exp (k \epsilon_{\D}) \sfP [C([\xdHold', \xdSwap'^{\sigma'}]) = C(\bm z)],
\end{equation}
where $k = \dHamSr(\xd, \xd')$. The probabilities in Equation~\ref{eqSketchProofWTS} are over the random sampling of the permutations $\sigma$ and $\sigma'$ in Algorithm \ref{algoPermute}. %
We can show that there exists a derangement $\rho$ of $k$ records such that $C(\xd) = C([\xdHold', \xdSwap'^\rho])$ (Lemma~\ref{lemmaDSteps}). (A derangement is a permutation that does not fix any rows.) Moreover, there is a bijection between the possible $\sigma$ and $\sigma'$ given by $\sigma' = \sigma \circ \rho$. Hence, if $k_{\sigma}$ is the number of records deranged by $\sigma$, we have 
\begin{equation}\label{eqInequalityOfM}
k_{\sigma} - k \le k_{\sigma'} \le k_{\sigma} + k.
\end{equation}
For such pairs of possible $\sigma$ and $\sigma'$, the ratio $\sfP(\sigma)/\sfP(\sigma')$ can be bounded in terms of $o^{k_{\sigma}-k_{\sigma'}}$ and the ratio between the number of derangements of size $k_{\sigma'}$ and of size $k_{\sigma}$. For $o \le 1$, this can in turn be bounded by $(b+1)^k o^{-k}$ using the inequality in Equation~\ref{eqInequalityOfM}. Hence $\epsilon_{\D} = \ln (b+1) - \ln o$ does indeed satisfy Equation~\ref{eqSketchProofWTS}. %

In Appendix~\ref{appendixOptimality}, we prove that the PLB $\epsilon_{\D}$ for the PSA given in Theorem~\ref{thmSwapDP} is tight in the weak sense that under some mild assumptions the difference between the right and left sides of the inequality in Equation~\ref{eqSketchProofWTS} is arbitrarily close to zero for some choice of $\xd$, $\xd'$, and $\bm z$.

\begin{remark}\label{remarkbFunctionD}
Since $n_{m\cdot\cdot}$ is an invariant, $n_{m \cdot\cdot}^{\xd} = n_{m\cdot\cdot}^{\xd'}$ for all $\xd$ and $\xd'$ in the same universe. Thus, $b$ is a function of $\D$ and hence so is the PLB $\epsilon_{\D}$ given in Equation~\ref{eqMainTheorem}. In the context of the PSA, we will use $\epsilon$ to denote the value of $\epsilon_{\DSwap(\xd^*)}$ under the universe $\DSwap(\xd^*)$ corresponding to the realized confidential data $\xd^*$. %
We will also report the PSA's protection loss budget in terms of this value $\epsilon$ and omit the values of $\epsilon_{\D}$ for other universes $\D$. 
Even though it is a function of the realized data $\xd^*$, the value of $\epsilon$ can still be publicly reported under the PSA's DP specification without additional protection loss \citep{bailieRefreshmentStirredNot2026}.
\end{remark}

\subsection{A Numerical Demonstration: The 1940 Census Full Count Data}\label{sec:numerical_demo}

We demonstrate the PSA using the 1940 U.S. Decennial Census full count data, obtained from the IPUMS USA Ancestry Full Count Database \citep{ipumsdata}. For the 1940 Census, the smallest geography level is county, hence swapping is performed among household units across counties within each state---that is, $\Vswap$ is set to be each household's county indicator. The matching variables (or swap key) $\Vmat$ are the number of persons per household and the household's state.  Our analysis is focused on the ownership status of household dwellings, an indicator variable taking value of either owned (including on loan) or rented. This is our $\Vhold \setminus \Vmat$. The invariants $\cswap$ induced by this swapping scheme include 1) the total number of owned versus rented dwellings at each of the household sizes at the state level, and 2) the total number of dwellings at each of the household sizes at the county level. In our notation, these are the $n_{m\cdot s}$'s and the $n_{mh\cdot}$'s, respectively. 

We restrict our illustration to the state of Massachusetts. Table~\ref{tab:swapped_table} compares the two-way tabulations of dwelling ownership by county based on the original data and one instantiation of the swapping mechanism using a high swap rate of $p=50\%$. The row margin of either table is the county-level total dwellings and is invariant due to $n_{\cdot h \cdot} = \sum_m n_{mh \cdot}$. The column margin is the total number of owned versus rented dwellings in Massachusetts and is invariant due to $n_{\cdot\cdot s} = \sum_m n_{m \cdot s}$.

\begin{table}[t!]
\centering
\caption{A comparison of two-way tabulations of dwelling ownership by county based on the 1940 Census full count for the state of Massachusetts (left) and one instantiation of the PSA at $p = 50\%$ (right).}

\begin{tabular}{l@{\hskip 36pt}llr|llr}
	\toprule
	County & Owned & Rented & Total & Owned & Rented & Total \\ 
	&&&& (swapped) & (swapped) & (swapped) \\
	\midrule
	Barnstable & 7,461 & 3,825 & 11,286 & 5,907 & 5,379 & 11,286 \\ 
	Berkshire & 14,736 & 18,417 & 33,153 & 13,770 & 19,383 & 33,153 \\ 
	Bristol & 33,747 & 63,931 & 97,678 & 35,537 & 62,141 & 97,678 \\ 
	Dukes & 1,207 & 534 & 1,741 & 946 & 795 & 1,741 \\ 
	Essex & 53,936 & 81,300 & 135,236 & 52,631 & 82,605 & 135,236 \\ 
	Franklin & 7,433 & 6,442 & 13,875 & 6,337 & 7,538 & 13,875 \\ 
	Hampden & 30,597 & 58,166 & 88,763 & 32,267 & 56,496 & 88,763 \\ 
	Hampshire & 9,427 & 8,630 & 18,057 & 8,145 & 9,912 & 18,057 \\ 
	Middlesex & 104,144 & 147,687 & 251,831 & 100,372 & 151,459 & 251,831 \\ 
	Nantucket & 593 & 432 & 1,025 & 471 & 554 & 1,025 \\ 
	Norfolk & 44,885 & 40,285 & 85,170 & 38,566 & 46,604 & 85,170 \\ 
	Plymouth & 24,857 & 23,882 & 48,739 & 21,549 & 27,190 & 48,739 \\ 
	Suffolk & 49,656 & 176,553 & 226,209 & 67,357 & 158,852 & 226,209 \\ 
	Worcester & 53,126 & 78,535 & 131,661 & 51,950 & 79,711 & 131,661 \\
	\midrule
	Total & 435,805 & 708,619 & 1,144,424 & 435,805 & 708,619 & 1,144,424 \\ 
	\bottomrule
\end{tabular} 
 \flushleft\footnotesize 
  { \textit{Note:} Total dwellings per county, as well as total owned versus rented units per state, are invariant; the values of these statistics are not shown.
    }
\label{tab:swapped_table}
\end{table}

Table~\ref{tab:plb_ma} supplies the conversion between different swap rates to the protection loss $\epsilon$ of the PSA. Under the current swapping scheme, the largest stratum size delineated by $\Vmat$ is $b = 264,331$, consisting of all two-person households of Massachusetts. Therefore by Equation~\ref{eqMainTheorem}, we see that a low swap rate of $1\%$ corresponds to $\epsilon = 17.08$, whereas a high swap rate of $50\%$ corresponds to $\epsilon = 12.48$. It is worth noting that since the invariants $\cswap$ are fixed in this analysis, the different values of $\epsilon$ presented in this table can be directly interpreted as SDC guarantees of different quantified strengths. On the other hand, as we alluded to earlier, the protection losses corresponding to different invariants $\cswap$ are not directly comparable---see the discussion in Section~\ref{sec:compare-das-psa}. 

\begin{table}[t!]
\centering
\caption{Conversion of (expected) swap rate $p$ to protection loss $\epsilon$.}
\begin{tabular}{lrrrr}
	\toprule
	$p$ & 0.01 & 0.05 & 0.10 & 0.50 \\ 
	$\epsilon$ & 17.08 & 15.43 & 14.68 & 12.48 \\ 
	\bottomrule
\end{tabular}
 \flushleft\footnotesize 
  { \textit{Note:} Under this swapping scheme, the largest stratum size is $b = 264,331$, the number of all two-person households of Massachusetts.} 
\label{tab:plb_ma}
\end{table}

We also examine the accuracy of the two-way tabulation as a function of swap rate. Figure~\ref{fig:mape_ma} shows the mean absolute percentage error (MAPE) in the two-way tabulation induced by swapping at different swap rates from $1\%$ to $50\%$. The variability across runs is small: each boxplot reflects $20$ independent runs of the PSA.

Here, the MAPE of a swapped table from its true table is defined as the cell-wise average of the ratio between their absolute differences and the true table values. The MAPE in Figure~\ref{fig:mape_ma} is with respect to the contingency table of county by dwelling ownership in Massachusetts and is defined in the notation of Proposition~\ref{propInvariantsOfSwap} as 
\[\frac{1}{\mathcal H \mathcal S} \sum_{h,s} \frac{\abs{n_{\cdot hs}^{\xd} - n_{\cdot hs}^{\bm Z}}}{n_{\cdot hs}^{\xd}},\]
where $\xd$ is the true table, $\bm Z$ is the swapped table, $h$ is the county indicator, and $s$ the indicator of whether the house was rented or owned.

\begin{figure}
\centering
\includegraphics[width = .8\textwidth]{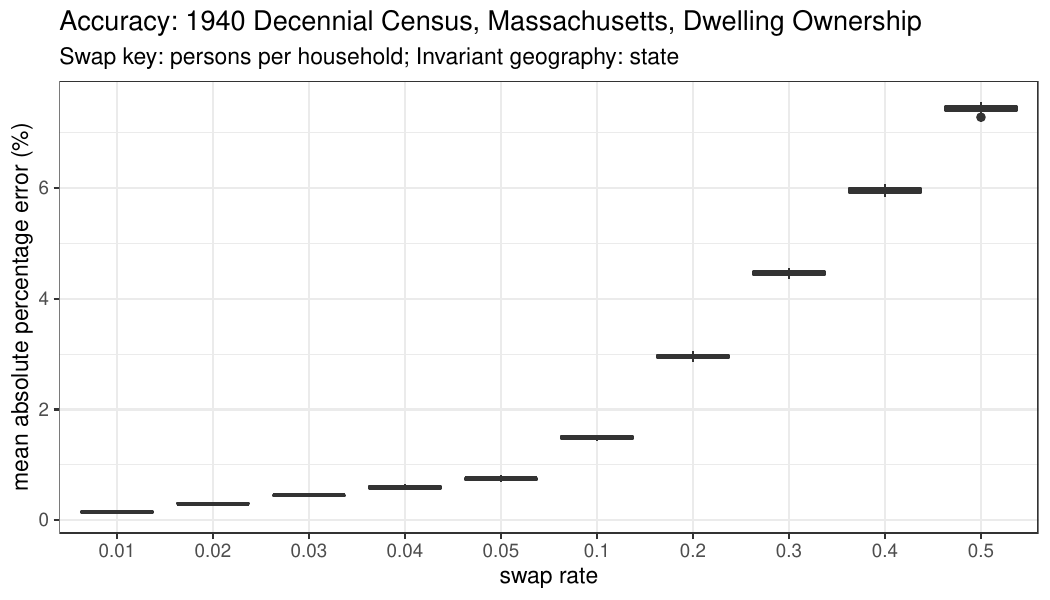}
\caption{Mean absolute percentage error in the two-way tabulation of dwelling ownership by county induced by the Permutation Swapping Algorithm applied to the 1940 Census full count data of Massachusetts, at different swap rates from $1\%$ to $50\%$. Each boxplot reflects $20$ independent runs of the PSA at that swap rate.}
\label{fig:mape_ma}
\end{figure}

The accuracy assessment we demonstrate here is highly limited. The analysis above assesses only cell-wise departures of the swapped two-way marginal table from its confidential counterpart. It does not capture potential loss of data utility in terms of multivariate relational structures. It is well understood in the literature that swapping erodes the correlation between $\Vswap$ and $\Vhold \setminus \Vmat$ \citep[see, e.g.,][]{slavkovic2010synthetic, drechslerSamplingSynthesisNew2010, mitra2006adjusting}. For the current example, this means the countywide characteristics of household dwellings (other than their size) are not preserved, but other multivariate relationships are. While an in-depth investigation into the utility of swapping is out of the scope of this article, we return to the subject of data utility in \citep{bailieRefreshmentStirredNot2024c} to discuss the implication this work may have on that line of inquiry.

\subsection{Estimating the Differential Privacy Specification of the 2010 DAS}\label{sec:2010PrivacyLoss}

If we entertain the assumption that the 2010 DAS used the PSA, we could obtain a crude sketch of the SDC guarantee afforded to the 2010 Census data. (We examine the validity of this assumption in detail in Appendix~\ref{appendixCompareDAS2010Swapping}.) As detailed in Example~\ref{exampleUSCensus2010DAS}, the 2010 DAS swapped household-records. Therefore, the DP flavor for the 2010 DAS would be \FlavSwap. Here the domain $\Xcef$ is the set of all possible Census Edited Files (CEFs), where the term CEF refers to the data set that consists of the census data after editing and imputation and is the input into the USCB's DAS.

The 2010 DAS utilized a swap key that includes the household size as well as the household voting age population and some geography (either tract, county, or state) (Example~\ref{exampleUSCensus2010DAS}). %
As we are unable to locate the 2010 Census data that allows for the precise calculation of $b$ pertaining to this particular swapping scheme, the swap key we consider here is coarser: we set the matching variables $\Vmat$ to be the household's size and state (so we do not include the third matching variable of the 2010 DAS, the household count of voting age persons). %
This results in a value of $b = 3.65$ million, which---as we use a coarser swap key---is an upper bound for the actual $b$ of the 2010 Census. %
Combining $b = 3.65$ million with a purported swap rate $p$ between 2--4\% \citep{boyd2022differential} implies that %
the nominal budget $\epsilon$ is between $18.29$ and $19$. 

The range 18.29--19 is an overestimate for the PLB; using the actual 2010 swap key would decrease the value of $b$ and hence result in a smaller value of $\epsilon$. However, we emphasize that the range 18.29--19 for the value of $\epsilon$ does not necessarily reflect the PLB of the 2010 DAS, but rather the PLB of the PSA when we choose its parameters to reflect what we know about the 2010 DAS. Moreover, as is always the case, this protection loss budget must be interpreted within the context of its DP flavor. Crucially, the DP flavor \FlavSwap\ for the above instantiation of the PSA includes the invariants $\Vhold$ and $\Vswap \times \Vmat$ (Proposition~\ref{propInvariantsOfSwap}). Under the 2010 parameter choices, these invariants are the counts of households by number of occupants at the block level, and all cross-classifications of nongeographical variables at the state level. Hence, the measure of protection loss provided by the above values of $\epsilon$ are modulo any SDC leakage caused by the release of these invariants.

\subsection{What Does the Permutation Swapping Algorithm's Budget Look Like?}

\begin{figure}
\centering
\includegraphics[width=.9\textwidth]{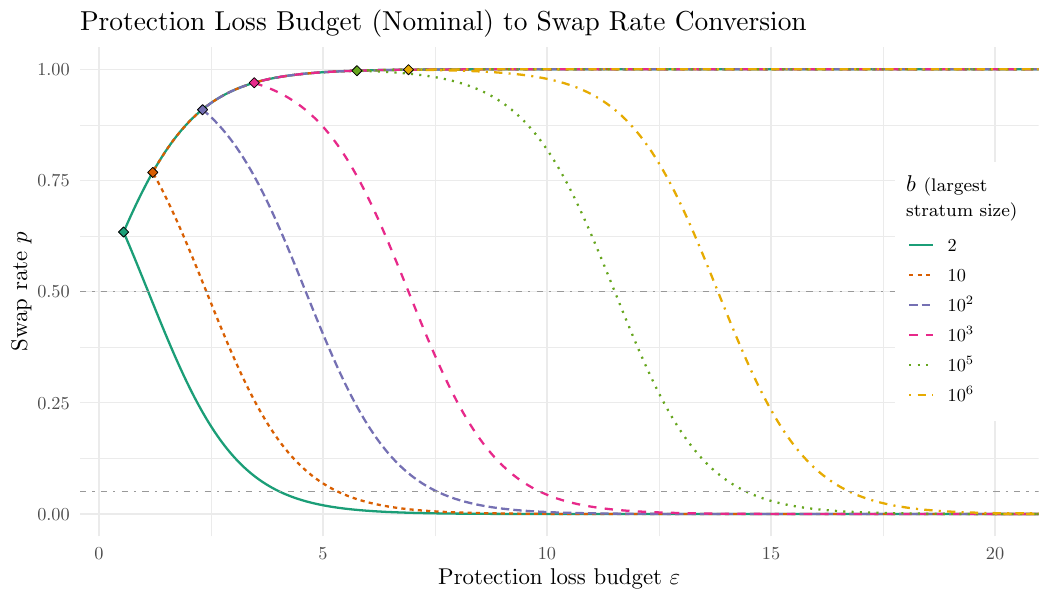}
\caption{Conversion between the nominal protection loss budget $\epsilon$ and the swap rate $p$ for the Permutation Swapping Algorithm. Color and line type encode different values of $b$, the size of the largest stratum delineated by $\Vmat$ (from $2$ to $1$ million). Outlined diamonds indicate the smallest $\epsilon$ attainable for each $b$. Grey dotted horizontal lines correspond to swap rates of 5\% and 50\% respectively. 
The protection loss budgets are nominal in that the statistical disclosure control guarantee they afford must be understood in the context of $\cswap$ (and hence the values of $\epsilon$ across different values of $b$ are not immediately comparable).}
\label{fig:p-by-epsilon}
\end{figure}

Figure~\ref{fig:p-by-epsilon} provides a visual illustration of Theorem~\ref{thmSwapDP}, connecting the PLB $\epsilon$ to the swap rate $p$ for a number of values for the largest stratum size $b$. 
Three observations are worth noting. First, for each $b$, there exists a smallest $\epsilon$, call it $\epsilon^{(b)}$, which lower bounds the PLB of the PSA (regardless of its swap rate $p \in [0, 1]$). The minimum budget $\epsilon^{(b)} = \ln(b+1)/2$ is achieved by the swap rate $p^{(b)} = \sqrt{b+1}/(\sqrt{b+1}+1)$. %
For each $b$ in Figure~\ref{fig:p-by-epsilon}, this quantity $\epsilon^{(b)}$ is marked by an outlined diamond. Importantly, the larger the $b$, the larger the minimum possible budget $\epsilon^{(b)}$. For example, when $b = 10$, $\epsilon^{(b)}$ is $1.20$ (attained at $p^{(b)} = 77\%$);
whereas when $b = 10^6$, $\epsilon^{(b)}$ is  $6.91$ (at $p^{(b)} = 99.9\%$). %

That some PLBs are not attainable for a fixed $b$ follows from the fact that the ratio $\sfP(\bm z \mid \bm x)/\sfP(\bm z \mid \bm x')$ of probabilities of a swapped data set $\bm z$ from two different input data sets $\xd$ and $\xd'$ depends not just on the swap rate $p$ but also the ratio $r$ of the number of derangements of size $\dHamSh(\xd, \bm z)$ to the number of derangements of size $\dHamSh(\xd', \bm z)$. (This is because the PSA selects $\dHamSh(\xd, \bm z)$ records and then samples derangements of size $\dHamSh(\xd, \bm z)$ uniformly at random.) This ratio $r$ is upper bounded by $(b+1)^{\dHamSh(\xd, \xd')}$, which means $\epsilon$ must be at least $\ln(b+1)-\ln p + \ln(1-p)$.

Second, for every $b$ and every budget $\epsilon > \epsilon^{(b)}$, two different swap rates can achieve that budget $\epsilon$, with the higher rate often being very close to $100\%$. For example at $b = 10$, a swap rate of either $35.4\%$ 
or $95.2\%$ %
achieves the nominal budget of $\epsilon = 3$. The mathematical reason behind this is that, for large $p$ (i.e., $p > p^{(b)}$) the ratio $r$ is dominated by the odds $o = p/(1-p)$, in which case $[\ln \sfP(\bm z \mid \bm x)-\ln\sfP(\bm z \mid \bm x')]/\dHamSh(\xd, \xd')$ is maximized when $\dHamSh(\xd', \bm z)$ %
is as small as possible. This results in $\epsilon = \ln o$, while, as explained in the previous paragraph, $\epsilon = \ln(b+1) - \ln o$ for $p \le p^{(b)}$. Since the former $\epsilon$ is monotone increasing in $p$ and the latter monotone decreasing, there are two swap rates $p$ corresponding to any $\epsilon > \epsilon^{(b)}$.
This is akin to the behavior of the randomized response mechanism, where a large probability $p_{\mathrm{RR}} > 0.5$ of flipping the binary confidential answer inadvertently preserves statistical information, thereby achieving the same budget $\epsilon = \abs{\ln o_{\mathrm{RR}}}$ as $1-p_{\mathrm{RR}}$.%

Third and most importantly, we emphasize that the budgets visualized in Figure~\ref{fig:p-by-epsilon} are \emph{nominal} in the sense that the SDC guarantee they afford must be understood with respect to the full context as outlined by the PSA's DP specification. %
An aspect of this context 
is $b$, the size of the largest stratum of $\Vmat$, and as a result, the same value of $\epsilon$ across different $b$'s should not be equated to be the same SDC guarantee. Indeed, the ordering of the $b$ curves in the figure suggests a seemingly peculiar fact that, for a larger $b$, a higher $p$ is needed to achieve the same $\epsilon$. This apparent contradiction is explained by a point we have repeatedly made: for a fixed data set, a change in the value of $b$ requires that the swapping invariants, and hence the PSA's SDC guarantee, also change.

\section{A Differentially Private Analysis of the TopDown Algorithm}\label{sectionCensus-TDA}

This section provides a DP specification for the TDA \citep{uscensusbureauTopDown2023,abowd2022topdown}. The TDA was used to produce the P.L. 94-171 Redistricting Summary File (PL) \citep{uscensusbureauPLnational2021,uscensusbureauPLstate2021} and the Demographic and Housing Characteristics File (DHC) \citep{uscensusbureauDHC2023, cumings-menonDisclosureAvoidance20202025} for the 2020 U.S. Decennial Census. Four other products---the Demographic Profile \citep{census2023DemographicProfile}, the Privacy-Protected Microdata Files (PPMF) \citep{uscensusbureauPPMF2024}, the Redistricting and DHC Noisy Measurement Files (NMF) \citep{uscensusbureau2020CensusRedistricting2023, uscensusbureauNMF2023}, and 118th Congressional District Summary File \citep{uscensusbureau2020Census118th2023}---were also derived during the production of the PL and the DHC. Hence the publication of these four additional data products do not contribute to additional privacy loss, and our DP specifications for the PL and the DHC automatically extend to cover the release of all six products.

We prove in Theorem~\ref{thmTDASatisfiesDP} that the TDA satisfies zero-concentrated DP (zCDP) \citep{bunConcentratedDifferentialPrivacy2016a} subject to its invariants. By this we mean that the TDA satisfies a DP specification whose output premetric is the normalized R\'enyi metric $\Dnor$---the output premetric corresponding to zCDP (see Appendix~\ref{appZCDP})---and whose multiverse contains the TDA's invariants. More exactly, we show that the TDA satisfies the DP flavor \FlavTDA, where %
$\scD_{\cTDA}$ is the multiverse induced by the TDA's invariants: the state population totals; the total number of housing units in each census block; and the count of each type of occupied group quarters in each block. By proving that the TDA cannot satisfy zCDP (with input premetric $\dHamS^p$) for any finite protection loss budget without conditioning on these invariants, we will also show that the TDA's DP flavor must have $\scD_{\cTDA}$ (or a refinement of $\scD_{\cTDA}$) as its multiverse.

The TDA, summarized in Algorithm~\ref{algoTDA}, was run twice for the 2020 Census---once to produce the PL and then a second time for the DHC. The PLBs associated with releasing these files are given in Table~\ref{tab:budgets_2020}. It is a two-step procedure: The first step (called the ``measurement phase'' in \citeauthor{abowd2022topdown}, \citeyear{abowd2022topdown}) produces the NMF $\bm T_p(\xd_p)$ and $\Th(\xdh)$. Here $\xd_p$ and $\xdh$ denote the representations of the CEF at the person and household levels, respectively. The NMF are noisy versions of tabular summaries $\bm Q_p(\xd_p)$, at the person level, and $\Qh(\xdh)$, at the household level, respectively. (In this section, we will include group quarters as households for the purposes of conciseness.) %
The tabular summaries $\bm Q_p(\xd_p)$ and $\Qh(\Xh)$ are different for the PL and the DHC, but, roughly, they are the raw statistics that the USCB would like to include in each of these files. For example, when releasing the PL, $\Qp(\Xp)$ and $\Qh(\Xh)$ are the statistics in this file, but aggregated directly from the census microdata without any noise. (However, to improve accuracy, the bureau adds some additional statistics to $\Qp(\xd_p)$ and $\Qh(\xdh)$, which do not appear in the PL.) %
Discrete Gaussian noise is added to $\bm Q_p(\xd_p)$ and $\Qh(\xdh)$ to produce the NMFs $\bm T_p(\xd_p)$ and $\Th(\xdh)$.

\begin{table}[ht]
\centering
\caption{The protection loss budgets of the mechanisms $\Tp$ (person) and $\Th$ (household) used in the first step of the TopDown Algorithm to produce the P.L. 94-171 Redistricting Summary File (PL) and Demographic and Housing Characteristics File (DHC).}
{\small
	\begin{tabular}{r l |cc}
		\toprule
		& & $\rho^2$%
	& $\epsilon$ (with $\delta = 10^{-10}$) \\
	\midrule
	PL & Household & 0.07 & 2.70 \\
	& Person & 2.56 & 17.90 \\
	DHC & Household & 7.70 & 34.33 \\
	& Person & 4.96 & 26.34 \\
	\midrule
	Total && 15.29 & 52.83 \\
	\bottomrule
\end{tabular}
\normalsize}
\flushleft\footnotesize 
   {\textit{Note:} The value of $\epsilon$ for each row is computed using the conversion $\epsilon = \rho^2 + 2 \rho \sqrt{-\ln \delta}$ given in \citet{bunConcentratedDifferentialPrivacy2016a} and adopted by the U.S. Census Bureau. (Hence the aggregate loss of 52.83 is not the sum of the individual $\epsilon$'s.) We follow the U.S. Census Bureau's choice of $\delta = 10^{-10}$. \\
   {Source: \citet{uscensusbureau20230403PrivacylossBudget2023}. }}
    \label{tab:budgets_2020}
\end{table}

In the second step (called the ``estimation phase'' in \citeauthor{abowd2022topdown}, \citeyear{abowd2022topdown}), the PPMF $\Zp$ and $\Zh$ are produced by solving a complex optimization problem. (The PPMF is also called the Microdata Detail File by \citeauthor{abowd2022topdown}, \citeyear{abowd2022topdown}.) The PPMF $\Zp$ and $\Zh$ agree with the CEF $\Xp, \Xh$ on the invariants $\cTDA$. %
In addition, the PPMF $\Zp$ and $\Zh$ for the DHC are consistent with related statistics in the PL \citep{uscensusbureauFactsheetDisclosureAvoidance2023}. To enforce this consistency, the PL $\bm P$ is passed as input into the TDA when producing the DHC and a constraint $\bm H(\Zp, \Zh) = \bm P$ is added to the optimization problem. %
(The input $\bm P$ is not used by the TDA in producing the PL.)

The PL and the DHC are tabulations of the PPMF data sets $\Zp$ and $\Zh$. In addition to the PL and the DHC, the USCB released the NMF $\Tp(\Xp)$ and $\Th(\Xh)$ produced for the PL and the DHC \citep{uscensusbureauReleaseDatesSet2023}, and the PPMF $\Zp$ and $\Zh$ produced for the DHC \citep{uscensusbureau2020CensusData}. The Demographic Profile and the 118th Congressional District Summary File are retabulations of the DHC \citep{santos2020CensusDemographic, uscensusbureau2020Census118th2023}.

\begin{algorithm}
\caption{Overview of the TopDown Algorithm \citep{abowd2022topdown}, focusing on aspects salient to statistical disclosure control.} 
\label{algoTDA}
\begin{algorithmic}
\REQUIRE
\STATE A CEF $\xd \in \Xcef$ with representations $\xd_p$ and $\xdh$ at the person and household levels respectively;
\STATE Person and household queries $\Qp$ and $\Qh$;
\STATE Privacy noise scales $\Dp$ and $\Dh$;
\STATE Constraints $\cTDA^+$ (including invariants $\cTDA$, edit constraints and structural zeroes);
\STATE (Optional) previously released statistics $\bm P$ along with an aggregation function $\bm H$, which specifies the relationship between $\bm P$ and the PPMF $\bm Z_{p}$ and $\Zh$.%
\end{algorithmic}
\begin{algorithmic}[1]
\STATE Step 1: Noise Infusion
\SCOPE
\STATE Sample discrete Gaussian noise \citep{canonneDiscreteGaussianDifferential2022}:
\SCOPE
\STATE $\Wp \sim \mathcal N_{\mathbb Z}(\bm 0, \Dp)$
\STATE $\Wh \sim \mathcal N_{\mathbb Z}(\bm 0, \Dh)$
\ENDSCOPE
\STATE Compute the NMF:
\SCOPE
\STATE $\Tp (\Xp) \gets \Qp(\Xp) + \Wp$
\STATE $\Th (\Xh) \gets \Qh(\Xh) + \Wh$
\ENDSCOPE
\ENDSCOPE
\STATE Step 2: Post-Processing
\SCOPE
\STATE Compute the PPMF $\Zp$ and $\Zh$ as a solution to the optimization problem:
\SCOPE
\STATE Minimize loss between $[\Tp(\Xp), \Th(\Xh)]$ and $[\Qp(\Zp), \Qh(\Zh)]$
\SCOPE
\STATE subject to constraints $\cTDA^+ (\Zp, \Zh) = \cTDA^+(\Xp, \Xh)$ and $\bm H(\Zp, \Zh) = \bm P$.
\ENDSCOPE
\ENDSCOPE
\ENDSCOPE
\end{algorithmic} 
\begin{algorithmic}
\ENSURE 
\STATE The PPMF $\Zp$ and $\Zh$;
\STATE The NMF $\Tp(\Xp)$ and $\Th(\Xh)$ at the person and household levels.
\end{algorithmic}
\end{algorithm}

\begin{theorem}\label{thmTDASatisfiesDP}
The TDA satisfies the DP specification \SpecTDA\ with PLB $\rho^2 = 2.63$ for the PL and $\rho^2 = 15.29$ for the DHC. (Note that these budgets do not vary with the universe $\D \in \DTDA$.) %

In the opposite direction, let $\bm c'$ be any proper subset of TDA's invariants. %
Then the TDA does not satisfy \SpecTDAPrime\ with any finite budget $\rho$.
\end{theorem}

A proof of Theorem~\ref{thmTDASatisfiesDP} is given in Appendix~\ref{appProofTDASatisfiesDP}. 

\begin{remark}\label{remarkRhoSquared}
Because the standard parametrization of zCDP's protection loss budget is equal to the square of our parametrization (Appendix~\ref{appZCDP}), throughout this article we report zCDP budgets in terms of $\rho^2$ to maintain consistency with the values reported in existing publications. We also drop the subscript $_{\D}$ when reporting 2020 Census budgets since these budgets do not vary with the universe $\D$. 
\end{remark}

\section{What if the 2020 Census Used Swapping?}\label{sectionCensus}

In this section we ask the counterfactual question: what if the PSA %
was applied to the 2020 Decennial Census? What would its protection guarantee look like? The answer is not a straightforward one for two reasons. First, the %
set of invariants implemented by the TDA for the 2020 Census does not conform to a valid invariant specification under any PSA regime (see Example~\ref{exTDAInvariantsNotSwapInvariants}). Second, the PSA may be deployed using different swapping schemes at varying swap rates, and we cannot know in the counterfactual world what swapping scheme and swap rate would have been chosen for the 2020 Census. Thus, our analysis proceeds in an exploratory fashion. We ask specifically: what would the disclosure risk for the 2020 Census look like under different choices for invariants, swapping schemes, and swap rates for the PSA, when these choices are made to mimic the design choices of the 2020 DAS, subject to reasonable and necessary departures? 

To answer this question, we begin by spelling out the DP specifications of the 2020 Census and its components, and examine them alongside the DP specifications that can be achieved with the PSA under a range of hypothetical choices for its parameters. As the ensuing analysis will make clear, the distinct DP flavors of the 2020 DAS and the PSA hinder the extent to which their DP specifications may be directly compared against one another---most crucially due to the fact that the invariants of one of these flavors are neither strictly stronger nor weaker than the other.

\subsection{An Overview of the 2020 Disclosure Avoidance System}\label{secOverview2020DAS}
For the `counterfactual PSA' to be at all relevant to the 2020 Decennial Census, we must first understand the DP specification of the 2020 DAS. This will allow us to ascertain the extent of compatibility between the PSA and the 2020 DAS, in order to facilitate an analysis of the DP specification of a would-be 2020 PSA deployment.

An overview of the 2020 DAS and its data products lays the foundation for our analysis of its DP specification. The USCB 
divides the 2020 Census data releases into three groups.
Group 1 encompasses the two principal data products that we have already discussed, namely, the PL and the DHC. (The Demographic Profile is also included in Group 1 but as it is simply a subset of DHC's tabulations, we do not consider it as a standalone product.) As detailed in the previous section, both the PL and the DHC were protected using %
the TDA~\citep{uscensusbureauTopDown2023,abowd2022topdown}. %
Group 2 encompasses the Detailed DHC-A~\citep{uscensusbureauDDHCA2023} and Detailed DHC-B Files~\citep{uscensusbureauDDHCB2024}, which were produced by the SafeTab-P and SafeTab-H Algorithms respectively \citep{uscensusbureauSafeTabP2023,uscensusbureauSafeTabH2024,tumultlabsSafeTabDPAlgorithms2022}. (Here `-P' and `-H' stand for `persons' and `households'.) Group 2 also includes the Supplemental DHC File (S-DHC)~\citep{uscensusbureauSDHC2024}, protected using the PHSafe Algorithm~\citep{uscensusbureauPHSafe2024}. %

Group 3 contains the additional products derived from the 2020 Census data, most notably the PPMF~\citep{uscensusbureauPPMF2024}, the \nth{118} %
Congressional District Summary File~\citep{uscensusbureau2020Census118th2023}, and NMFs~\citep{uscensusbureau2020CensusRedistricting2023, uscensusbureauNMF2023}. As explained in the previous section, because these data products are derived either from the Group 1 products or the privacy-protected intermediate outputs pertaining to those products, their production does not contribute to the overall 2020 PLB. As a result, we need not consider these Group 3 products in our analysis. Other Group 3 data releases, such as publications from researchers with access to census microdata, may be released in the future \citep{hawes2020CensusDisclosure}. 
Our analysis does not account for these releases, nor for products derived from 2020 Census data that are not listed here.

\subsection{Understanding the Differential Privacy Specifications of the 2020 Census}\label{sectionSpecifications2020}

The first five rows of Table~\ref{tab:compare_2020} summarize the DP specifications of the 2020 DAS overall, as well as its constituent algorithms, including the TDA \citep{uscensusbureauTopDown2023,abowd2022topdown}, the SafeTab Algorithms \citep{uscensusbureauSafeTabP2023,uscensusbureauSafeTabH2024,tumultlabsSafeTabDPAlgorithms2022}, and the PHSafe Algorithm \citep{uscensusbureauPHSafe2024}. By way of contrast, the last row of Table~\ref{tab:compare_2020} summarily presents the DP specification of the hypothetical application of the PSA to the 2020 Decennial Census, the details for which is expanded upon in Table~\ref{tab:whatif2020} of Section~\ref{sec:whatif}. 

The following remarks accompany the pertinent entries in Table~\ref{tab:compare_2020}:

\begin{table}%
\caption{The differential privacy (DP) specifications of the 2020 disclosure avoidance system (DAS) 
and of the hypothetical application of the Permutation Swapping Algorithm (PSA) to the 2020 Decennial Census.}
\centering
{\small
\begin{tabular}{l|c|c|c|c}
	\toprule
	& $\dPr$ & $\dX$  (Resolution)  & Invariants$^{\ref{noteEditConstraints}}$ & Protection Loss Budget$^{\ref{noteZCDPBudget}}$ \\
	\midrule
	\multirow{4}{*}{TopDown} &   \multirow{8}{*}{$\Dnor$} & \multirow{5}{*}{$\dHamS^p$  (person)}  &  Population (state)  & PL \& DHC:   \\
	&&&  Total housing units (block) & $\rho^2 = 15.29$ \\ %
&&&  Occupied group quarters  &  $(\epsilon = 52.83, \delta = 10^{-10}$) \\
&&& by type (block) %
& See Table~\ref{tab:budgets_2020} \\
\cmidrule{1-1}\cmidrule{4-5}
SafeTab-P &   &  & \multirow{2}{*}{Total housing units (block)} & DDHC-A: $\rho^2 = 19.776$ \\
\cmidrule{1-1}\cmidrule{3-3}\cmidrule{5-5}
SafeTab-H &  &  \multirow{2}{*}{$\dHamSh$  (household)} &  & DDHC-B: $\rho^2 = 17.79$ \\
\cmidrule{1-1}\cmidrule{4-5}
PHSafe    &   &  & $\ge 1$ housing unit (block)$^{\ref{notePHSafeInvariant}}$ & S-DHC: $\rho^2 = 2.515$ \\
\midrule
Overall (to date) & \multirow{2}{*}{$\Dnor$} & \multirow{2}{*}{$\dHamS^p$  (person)} & \multirow{2}{*}{Same as TopDown} & $\rho^2 = 55.371$ \\
2020 DAS$^{\ref{noteOverall2020Spec}}$ &&& & ($\epsilon = 126.78, \delta = 10^{-10}$) \\
\midrule
\midrule
\multirow{2}{*}{Swapping (PSA)} &  \multirow{2}{*}{$\Mult$} &  \multirow{2}{*}{$\dHamSh$ (household)} & Varies but much greater & \multirow{2}{*}{$\epsilon$ between$^{\ref{noteSwapBudget}}$ 8.42--19.36} \\
& & & than TopDown$^{\ref{noteSwapInvariants}}$%
&   %
\\  
\bottomrule
\end{tabular}
\normalsize}

\flushleft\footnotesize 
   {\textit{Note:} The 2020 DAS consists of the TopDown Algorithm, %
which produced the P.L. 94-171 Redistricting Summary File (PL) %
and the Demographic and Housing Characteristics File (DHC); %
the SafeTab Algorithms, %
which produced the Detailed DHC-A File (DDHC-A) and Detailed DHC-B File (DDHC-B);
and the PHSafe Algorithm, which produced the Supplemental DHC File (S-DHC). For each DP specification, the protection domain is the set $\Xcef$ of all possible Census Edited Files and the multiverse is induced by the listed invariants. %
$\dHamS^p$ and $\dHamSh$ denote the Hamming distance at the resolution of person- and household-records, respectively (Equation~\ref{eqHamSDistDefn}); $\Dnor$ the normalized R\'enyi metric (Appendix~\ref{appZCDP}), which is the output premetric underlying zCDP \citep{bunConcentratedDifferentialPrivacy2016a}; and $\Mult$ the multiplicative distance (Equation~\ref{eqMultDistDefnPart2}), which is pure DP's output premetric \citep{dwork2006calibrating}. The numbered superscripts \ref{noteEditConstraints}--\ref{noteSwapBudget} in the table refer to the explanatory Remarks~\ref{noteEditConstraints}--\ref{noteSwapBudget}.}%
\label{tab:compare_2020}
\end{table}

\begin{remark}[TDA's additional constraints]\label{noteEditConstraints}
In addition to invariants, the TDA also enforces that the PPMFs $\bm Z_p$ and $\bm Z_h$ satisfy edit constraints and structural zeroes \citep{abowd2022topdown}. Edit constraints are rules that the USCB applies in their editing procedure to correct implausible or impossible census responses. One such rule is that a mother must be a certain number of years older than any of her children. If a census record does not satisfy an edit constraint, the USCB will modify it so that the edited record does. Structural zeroes are similar---they describe rules that cannot be broken by the census data because of constraints embedded in data collection. For example, if a block has an occupant, then one of the households in the block must have an occupant. As such, the CEF will always satisfy edit constraints and structural zeros. Because every possible $\xd \in \Xcef$ satisfies these rules by construction, these requirements need not be included as invariants. %
\end{remark}

\begin{remark}[zCDP budget conversion]\label{noteZCDPBudget} Because we reparametrize zCDP, we report its budgets in terms of $\rho^2$ (rather than $\rho$) to be consistent with other literature on the 2020 DAS (see Remark~\ref{remarkRhoSquared}). %
Moreover, following the USCB, we use the formula from \citet{bunConcentratedDifferentialPrivacy2016a}, $\epsilon = \rho^2 + 2\rho \sqrt{-\ln \delta}$, to convert from a zCDP budget $\rho^2$ to an approximate DP budget $(\epsilon,\delta)$. %
\end{remark}

\begin{remark}[Inequality invariants of PHSafe] \label{notePHSafeInvariant} The PHSafe Algorithm has inequality invariants (see Equation~\ref{eqInequalityInvariant}). Specifically, its invariant function $\bm c(\xd)$ is the vector of indicators for whether each census block has at least one housing unit or not.

\end{remark}

\begin{remark}[DP specification of the overall 2020 DAS] \label{noteOverall2020Spec} This DP specification covers all of the primary 2020 Census releases \citep{bureauCensusBureauReleases} discussed previously in this section, but not other data products that were derived from the 2020 CEF, such as the 2020 DAS accuracy metrics \citep{uscensusbureau2020CensusDisclosure2023}, the Population and Housing Unit Estimates \citep{uscensusbureauMethodologyUnitedStates2023}, and the National Population Projections \citep{uscensusbureauMethodologyAssumptionsInputs2023}. We were unable to locate information on the DP specifications associated with these data products. Nevertheless, as with any data release, %
they  
necessarily increase the total PLB associated with the 2020 Census. They could also possibly weaken the 2020 Census's DP flavor by increasing the invariants, weakening the output premetric, or increasing the resolution of the input premetric. %
Moreover, the USCB may make additional releases in the future, such as the Surname File \citep{uscensusbureauDecennialCensusSurname2016} or research papers generated with access to census microdata \citep{hawes2020CensusDisclosure}. These releases would further weaken the DP specification for the 2020 Census. 
In comparison, under data swapping, the PLB and DP flavor covers all data releases (Section~\ref{sec:whatif}).

\end{remark}

\begin{remark}[PSA invariants] \label{noteSwapInvariants} Depending on the swap key $\Vmat$ and the swapping variables $\Vswap$, the invariants induced by the PSA are all (multivariate) household characteristics at either the state, county, or block group levels, and optionally the household size at the corresponding geography one level lower. See Section~\ref{sec:whatif} for details.
\end{remark}

\begin{remark}[PSA budget] \label{noteSwapBudget} The exact PLB $\epsilon$ of the PSA depends on the swap rate $p$ and the swap key $\Vmat$, with the combination of a higher swap rate and finer geography-household strata giving rise to the lower range and vice versa (Table~\ref{tab:whatif2020}).
\end{remark}

Last but not least, we comment on the choice of protection units for both the 2020 DAS and the hypothetical PSA. The protection units (also known as the `privacy' units) of a DP specification are the basic entities that are protected under that DP specification. More exactly, a specification's PLB restricts how much a mechanism can change when a single protection unit's data changes. One might imagine that the protection units of a DP specification correspond to the resolution of its input premetric $\dX$---and this is true, but only in simplistic examples. Here by `the resolution of $\dX$,' we mean the size of the change between $\xd$ and $\xd'$ when $\dX(\xd, \xd') = 1$. For example, if $\dX(\xd, \xd') = 1$ whenever $\xd$ and $\xd'$ differ on a person-record, then the resolution of $\dX$ is a person. Common resolutions, in order from high to low, are: single transactions or interactions, persons, households, and businesses.

However, data preprocessing can create complications, so that the protection units of a specification are not always given by the resolution of $\dX$. In the case of the Decennial Census, an individual contributor's data can be used for multiple records in the CEF $\xd \in \Xcef$ because the USCB's imputation procedure replaces missing records with copies of nonmissing records. As such, the protection units of $(\Xcef, \scD, \dX, \dPr)$ do not correspond to the resolution of $\dX$. Rather, the protection units of the 2020 DAS are `post-imputation persons'---that is, those (fictional) entities with data that is exactly one record in the CEF. Similarly, the PSA's protection units are the `post-imputation households' rather than actual households.

The distinction between actual and post-imputation protection units is not simply a matter of semantics. From the perspective of a data contributor, the resolution of $\dX$ is not particularly informative in determining the SDC protection provided to them. Rather, a more accurate measure of a contributor's actual protection loss is given by the nominal PLB multiplied by the number of person-records the contributor contributed to. %
For example, if the 2020 imputation process duplicates a contributor's record once, then their substantive PLB is $\rho^2 \ge 221.48$ (or $\epsilon \ge 364.31$ with $\delta = 10^{-10}$), rather than $\rho^2 \ge 55.37$%
. (Here the $\rho^2$ budgets are inflated by a factor of four since doubling $\rho$ quadruples $\rho^2$. Also, we write $\ge$ rather than $=$ because the contributor's protection loss will increase due to data releases that we have not accounted for, as explained in Remark~\ref{noteOverall2020Spec}.) 
In general, the conversion from a DP flavor with post-imputation persons as units to a DP flavor with persons as units requires an inflation of the PLB by a factor equal to the maximum number of times a record can be duplicated (for a proof of this, see the section on group privacy in \cite{bailieRefreshmentStirredNot2026}). To avoid this complication, we have reported post-imputation budgets in Table~\ref{tab:compare_2020}, but we caveat this with the important observation that these budgets are relative to unusual protection units.

\subsection{The Protection Guarantee of the 2020 Census Under Swapping}\label{sec:whatif}

Table~\ref{tab:whatif2020} shows the total nominal PLB $\epsilon$ that would be achieved by applying the PSA to the 2020 Decennial Census for a variety of possible parameter choices. For the purpose of illustration, we stipulate the swapping variable $\Vswap$ to be the block, tract, or county membership of each household, and the matching variable $\Vmat$ to be the geography one level higher than $\Vswap$, either alone or crossed with the household size variable. These choices justify what Table~\ref{tab:compare_2020} claims, that the invariants of the PSA ``varies'' but are ``much greater than'' those of the TDA (see Remark~\ref{noteSwapInvariants}).

Note that if the PSA were applied to the 2020 Decennial Census%
, the nominal $\epsilon$ reported in Table~\ref{tab:whatif2020} would be the \emph{total} PLB across all data products derived from the swapped data set $\bm Z$, including the PL, the DHC, the DDHC, and the S-DHC, for both persons and household product types. This is because swapping is performed on the full CEF, and hence produces a synthetic version of it from which all data products can be generated. Therefore, when comparing the PLBs in Table~\ref{tab:whatif2020} with those reported for the 2020 DAS in Table~\ref{tab:compare_2020}, it should be understood that the protection loss for the PSA covers all the 2020 data products. %
This characteristic of swapping leads to an additional desirable property that is not necessarily enjoyed by mechanisms based on output noise infusion (such as those used in the 2020 DAS): the \emph{logical consistency} between, and within, multiple data products is automatically preserved under swapping without the need for any additional processing, because all these data products are produced from the same post-swapped microdata. %

The $\Vswap$ levels in Table~\ref{tab:whatif2020} are ordered in increasing granularity of geography. Within each level of $\Vswap$, the two $\Vmat$ levels are nested, in the sense that the swapping scheme represented in the latter row (i.e., crossed with household size) induces a logically stronger and more constrained set of invariants than the former one. 
These $\Vmat \times \Vswap$--level combinations result in largest strata of varying sizes, as can be seen from $b$ ranging from as large as $13.47$ million (the total number of households in California) to as small as $4,549$ (the total number of two-person households in a Florida block group).%

\begin{table}[ht]
\centering
\caption{The total nominal protection loss budget $\epsilon$ for the Permutation Swapping Algorithm applied to the 2020 Decennial Census for a variety of $\Vmat$, $\Vswap$, and swap rate choices.}
{\small
\begin{tabular}{llrrrl}
\toprule
$\Vmat$ & $\Vswap$ & $b$ & Total $\epsilon$  & Total $\epsilon$  & Largest Stratum \\ 
&&& $p = 5\%$ & $p= 50\%$ & \\
\midrule
State & county & 13,475,623 & 19.36 & 16.42 & California \\ 
State $\times$ household size & county & 3,948,028 & 18.13 & 15.19 & California, two-household \\ 
County & tract & 3,420,628 & 17.99 & 15.05 & Los Angeles County \\ 
County $\times$ household size & tract & 939,185 & 16.70 & 13.75 & Los Angeles County, two-household \\ 
Block group & block & 6,204 & 11.68 & 8.73 & a California block group  \\ 
Block group $\times$ household size & block & 4,549  & 11.37 & 8.42 & a Florida block group, two-household \\ 
\bottomrule
\end{tabular} 
\normalsize}

\flushleft\footnotesize 
   {\textit{Note:} The column $b$ is the number of households in the largest stratum, obtained from the Demographic and Housing Characteristics File. (The California and Florida block groups identified in rows 5 and 6 have 2020 Census GEOIDs 060730187001 and 121199114024, respectively.)}
\label{tab:whatif2020}
\end{table}

This analysis highlights an important, yet perhaps counterintuitive, observation.   %
When the swap rate $p$ is fixed, including more invariants decreases the nominal PLB $\epsilon$ of the PSA. As Table~\ref{tab:whatif2020} shows, when swaps are performed freely across counties in a state, even a high swap rate of $50\%$ renders a nominal $\epsilon$ that is much larger than that pertaining to swaps among households of the same size within a block group at a low swap rate of $5\%$ ($\epsilon = 16.42$ and $11.37$, respectively). If these nominal $\epsilon$'s are taken at face value, one may be tempted to conclude that  swapping schemes with finer invariants should be preferred from an SDC standpoint. Furthermore, one may find it convenient to also recognize that finer invariants are desirable from a data utility standpoint, for the obvious reason that more exact statistics about the confidential data are made known. However, as we warned right after presenting Theorem~\ref{thmSwapDP}, such a conclusion---that finer invariants should benefit both utility \emph{and} privacy---would be dangerously mistaken, for it overlooks the privacy leakage, in an ordinary sense of the phrase, due to the invariants alone. Indeed, it is the loss of protection due to releasing more invariants that results in less information remaining that needs protection, thereby creating the illusion that we can achieve better DP protection with finer invariants. 
This illusion highlights the criticality of interpreting $\epsilon$ within the context of its DP flavor, and the necessity of treating the invariants as an integral part of any protection guarantee.

\subsection{Comparing the 2020 DAS With Swapping: The Need for a Disclosure Risk-Based Assessment}\label{sec:compare-das-psa}

With the aid of our system of DP specifications, we are able to articulate both the flavor and the intensity of the DP guarantees of the 2020 Census as well as the hypothetical application of swapping to it. A side-by-side examination of the building blocks of these DP specifications in Tables~\ref{tab:compare_2020} and~\ref{tab:whatif2020} elucidates the similarities and differences between the protection provided by the PSA and the 2020 DAS.  But we caution against reading these tables as a \emph{direct} comparison between the two SDC approaches. In our view, these tables illustrate precisely a \emph{lack} of comparability between these approaches, due to their distinct DP flavors, which in turn render comparisons of their PLBs as effectively one of `apples' versus `oranges.'

To be sure, the DP specifications of the 2020 DAS and the PSA share similarities that facilitate their comparison on a conceptual level. Firstly, both the 2020 DAS and the PSA have the same protection domain: the set of all possible CEFs $\Xcef$. This means that the PSA and the 2020 DAS protect the data $\xd \in \Xcef$ as it exists after collection, coding, editing, and imputation, rather than as it exists at other stages in its life cycle.  As
such, it is not the contributors’ data (i.e., their `raw’ census responses) that are directly protected, but rather it is the edited and imputed data (i.e., the CEF) that receives the DP guarantee.

Secondly, because both the 2020 DAS and the PSA have invariants, each of their DP specifications partition the protection domain $\Xcef$ into multiple universes. This operation constrains the scope of the 2020 DAS and the PSA's SDC protection to data sets that agree on their invariants. Therefore, for the same reasons that the 2020 DAS cannot satisfy the original specification of zCDP given in \citet{bunConcentratedDifferentialPrivacy2016a}, data swapping cannot satisfy the original pure DP specification of \citet{dwork2006calibrating}. In this sense, both the 2020 DAS and the PSA are DP only insofar as their invariants allow.

Thirdly, the input premetrics for the PSA and the 2020 DAS are both Hamming distances, although with differing resolutions---household-records for the PSA and person-records for the 2020 DAS. This means the protection units are post-imputation households and post-imputation persons, respectively. Since the input premetric is the yardstick for measuring change in the input data,  %
using a lower resolution like household-records provides more protection than a higher resolution like person-records (all else being equal). %
That is, a household-level distance is a stronger notion than a person-level distance, since if the record of a single household changes part of its value, the multiple persons residing in a same household may all change their records.

Fourthly, the PSA's output premetric is also stronger than the 2020 DAS's. The PSA uses the multiplicative distance $\Mult$---as used in pure DP \citep{dwork2006calibrating}---while the 2020 DAS uses the normalized R\'enyi metric $\Dnor$---as used in zCDP \citep{bunConcentratedDifferentialPrivacy2016a}. There exist probabilities $\sfP$ and $\sfQ$ with $\Dnor(\sfP, \sfQ)$ arbitrarily small but $\Mult(\sfP, \sfQ) = \infty$. As such, $\Mult$ ensures a greater level of SDC protection than $\Dnor$ (again, assuming that all else is equal).

On the other hand, however, there are crucial differences between the DP specifications of the 2020 DAS and the PSA that ultimately render their comparison a wrought endeavor.
The most damning obstacle to a meaningful comparison between the DP flavors of the 2020 DAS and the hypothetical PSA is that they do not, and cannot (on the hypothetical PSA's account) induce invariants in the data product that are equal, or even that are nested with respect to one another. As Proposition~\ref{propInvariantsOfSwap} makes clear, a swapping regime maintains invariant two marginal tables, that is, the cross-classifications of the matching variables by the swapping variables and that of all the holding variables, while perturbing the interior cells of the multiway contingency table. The 2020 DAS, on the other hand, used a list of invariants and other constraints on an as-needed basis. They need not, by design, accord to some marginal tabulations of the underlying microdata. 
As a matter of fact, all of the SDC methods used in 2020 have invariants, but the PSA has many more invariants than any of these methods and, as such, places more restrictions on the scope of protection. 
On the flip side, while swapping almost always has stricter invariants for most variables, it does not necessarily have the TDA's group quarter invariants. Therefore, the 2020 DAS DP flavor is not strictly stronger than the PSA's flavor, nor visa versa---although the 2020 DAS places less restrictions on the scope of protection, these are not nested within the restrictions induced by the PSA's invariants.

Even with nested invariants---for example, in the case that the 2020 DAS's invariants were compatible with some swapping regime, the extent of comparison is still limited to a qualitative degree. We should all agree that when all else are held the same, a smaller set of invariants is strictly less disclosive than a larger superset of invariants (Section~\ref{secUnderstandingImpactInvariants}). But just how much less disclosive? We are still left to appeal to our intuition to reason with the privacy leakage induced by the additional invariants, with no quantitative guidance available.

As a consequence of the incommensurability between their DP flavors, the PLBs of the PSA and of the 2020 DAS are not directly comparable because a budget's `unit of measurement' is determined by its DP flavor. That is to say, a PLB is a nominal measure of SDC protection, which is always relative to---and hence can only be understood within the context of---the four other building blocks. The DP flavors for the PSA and the 2020 DAS are different and so their budgets have different units of measurement. 

An astute reader may point out that just like the USCB, we can convert the 2020 DAS zCDP budget $\rho^2 = 55.371$ to the approximate DP budget of $\epsilon = 126.78$ with $\delta = 10^{-10}$, which is more comparable with a pure DP budget $\epsilon$ under the view that the latter has an $\delta=0$. This calculation would appear to reveal that the PLB of the 2020 DAS is an order of magnitude larger than that of the PSA. 
We caution, again, that a budget conversion does not render 
the PSA and the 2020 DAS comparable because their DP flavors have different invariants and different input premetrics. While the 2020 DAS's budget would substantially increase under a household-level input premetric (Appendix~\ref{appProofTDASatisfiesDP}), removing even one invariant from the PSA's DP specification would result in an infinite PLB (Section~\ref{secUnderstandingImpactInvariants}.

It seems to us that the ultimate comparison between two competing SDC algorithms that sport different DP specifications should be phrased in terms of disclosure risk. The focus on disclosure risk is appropriate as it has been traditional in the SDC literature; see, for example, \citet{kazanAssessingStatisticalDisclosure2024} for a demonstration of such assessment in the context of the 2020 Census. Articulating the DP specifications of SDC algorithms and recognizing the differences among competing options are essential to understanding their likely impact on disclosure risk. Yet, most crucially, disclosure risk stands to be the most promising, if not the only viable, yardstick for measuring the efficacy of SDC protection in the presence of invariants. Methodologically, we are encouraged by the observation that releasing a statistic under a large PLB is pragmatically equivalent to making that statistic an invariant. Hence, in principle it should be possible to effectively trade off invariants with large budgets, thus making the comparison between the budgets of SDC algorithms with distinct flavors, such as the PSA and the 2020 DAS, a more tractable one. 
Another possible equating principle is the vulnerability of a DP mechanism (and therefore by extension, its specification) to a reconstruction attack. For example, \citet{abowd2010CensusConfidentiality2023} show that swapping at a rate of 50\% had about the same protection against reconstruction as the TDA (on the DHC). Hence, the increase from the TDA's invariants to the 2010 invariants might be roughly comparable to the increase in budget from $\epsilon \approx 8.42$--19.36 (the swapping budget, at the household level) to $\epsilon \approx 52.83$ (the DHC budget, with $\delta=10^{-10}$, at the person-level). Comparisons between the performance of swapping and the 2020 DAS against other privacy attacks could potentially provide additional heuristics on how invariants can be compared to large budgets \citep{steedQuantifyingPrivacyRisks2025a, ballesterosEvaluatingImpactsSwapping2025, christDifferentialPrivacySwapping2022}. Sections~\ref{secImpactInvariants} and~\ref{secMitigateInvariants} offer some of our preliminary thoughts on the impact of invariants on disclosure risk, though we leave this question largely as a subject for future research.

\section{Discussion}\label{sec:discussion}

This article continues an existing line of research \citep{rinottConfidentialityDifferentialPrivacy2018, bailieABSPerturbationMethodology2019, sadeghiConnectionABSPerturbation2023, neunhoefferFormalPrivacyGuarantees} examining traditional SDC methods---which are typically regarded as ad hoc and are motivated by intuitive notions of protection or specific attacker models
---through the lens of DP. This body of literature shows that---even though they were designed without DP in mind---traditional SDC methods can still be fruitfully analyzed from the perspective of DP. By providing another example---data swapping---that can be studied theoretically via the lens of DP, we hope to inspire 
further formal analyses of other traditional SDC methods. 
This type of analysis improves our understanding of such methods by supplying mathematical descriptions of the level and substance (or, in our terminology, the intensity and flavor) of the methods' SDC. Such descriptions are important: they can provide assurance to data providers and custodians that their data is adequately protected; or, conversely, they can reveal inadequate SDC and spur additional protection. 

However, it can be challenging to assess whether a given DP specification provides an adequate level of protection. To do so, we must understand how choices for each of the five building block can affect SDC---both individually and in conjunction with choices for the other building blocks.  
This requires answering a range of difficult sociotechnical questions. 
For instance, taking the other four building blocks as fixed, what PLB (if any) is sufficient for adequate SDC---adequate for whom, and who should decide this adequacy? 
Also, what is the practical impact of the protection units being post-imputation persons, and how should, if at all, this impact be mitigated? 
Furthermore---and most relevant for this article---what is the effect of invariants on SDC?

While we know that increasing the invariants strictly weakens the DP specification (Section~\ref{secUnderstandingImpactInvariants}), it is more difficult to determine how they affect an attacker's ability to make disclosures. \citet{ashmead2019effective} have investigated the effect of the 2020 DAS invariants, but there is a need for future work that studies the effect of invariants at the scale of those induced by data swapping. 
In addition to building technical understanding---and parallel to studies that survey preferences on appropriate settings for the PLB, protection domain and input and output premetrics---it could be beneficial to gauge %
public opinion on the acceptability of specific invariants.

Nevertheless, by providing DP specifications for both the PSA and the TDA, we demonstrate the feasibility of mathematically comparing and contrasting on fair grounds traditional SDC methods 
with DP-based mechanisms. With these two algorithms as prime examples, the paper points to the possibility of similar comparative analyses  
between other SDC methods, both those that were explicitly inspired by DP and those that were designed and motivated from non-DP perspectives. 
By explicating the five building blocks for the PSA and the TDA, we hope to promote nuanced assessments of DP deployments that go beyond discussion of the PLB.

\subsection{Understanding the Impact of Invariants on Disclosure Risk}\label{secImpactInvariants}

A major criticism of the swapping method implemented in the 2010 Census is that it induces too many invariants. One salient consequence of a plurality of %
invariants is that it severely constrains the permissible values for the confidential data. %
Indeed, 
the larger the number of invariants, 
the more data sets an attacker can rule out as impossible,
and, consequently, the higher the risk of disclosure. %
As \citet{abowd2023confidentiality} discuss, the invariants in the 2010 Census elevate disclosure risk because not only are they numerous but they also include information at a very fine granularity, for example, the total and voting age populations at the block level.%

The USCB's reconstruction and reidentification attack against their 2010 Census \citep{abowd2010CensusConfidentiality2023}
provides some understanding of the impact of its invariants on disclosure risk, although, as we will see, there were also other confounding factors at play.
As its name suggests, a reconstruction and reidentification attack consists of two steps: a reconstruction attack, which creates a plausible version of the confidential microdata; followed by a reidentification attack, which links this `reconstructed' data set to an external source in order to attach names or other personally identifying information to (some or all of) the reconstructed records. Because this second step assigns identities to the linked microdata records, an attacker can use the resulting data to infer some information about some specific census respondents---such as their race or ethnicity---although they cannot be certain that this inference is correct due to both the potential for linkage errors and the uncertainty introduced by swapping.

A reconstruction attack works by collating many publicly available, aggregate statistics about the confidential (unknown) microdata. Then it constructs a data set that agrees with these statistics---or, in the case where the published statistics are noisy functions of the confidential microdata, it \textit{infers} a data set based on a statistical model of the published data.\footnote{This model has the unknown microdata as its parameters, the released statistics as its data, and the SDC method that generated the published statistics from the confidential microdata as its data-generating process. For example, the reconstructed data set could be the maximum likelihood estimate under this model, or it could be a draw from a Bayesian posterior that is compatible with this model.} This reconstructed data set is a plausible guess for the confidential microdata, since, in the case where the statistics are deterministic functions of the confidential microdata, it %
generates identical statistics to the ones generated by the microdata. When the statistics are noisy, the situation %
is more complex: the reconstructed data set does not necessarily reproduce the published statistics exactly due to their stochasticity, but nevertheless it still could plausibly be the microdata that generated these statistics.
In any case, the larger the number of published statistics and the more accurate they are, the more heavily they constrain the possible configurations of the reconstructed data set, and hence the more likely it is that this reconstruction agrees with the true confidential microdata. 

For the reconstruction attack on the 2010 Census, the underlying microdata the USCB targeted was not the CEF$\xd_{\mathrm{CEF}}$, but the post-swapped data---that is, the resulting records after swapping had been applied to $\xd_{\mathrm{CEF}}$. %
This avoids the complication of designing a reconstruction attack that accounts for the noise introduced by swapping. Yet, %
because of the low swap rate and the large number of invariants in the 2010 DAS, there is a high degree of alignment between the reconstructed data and $\xd_{\mathrm{CEF}}$. As a result, linking the reconstructed data to an external data set containing personally identifiable information allows for the possibility of learning potentially sensitive data: the race and ethnicity of census respondents.
Indeed, the USCB found that an attacker could predict the race and ethnicity of about 3.4 million vulnerable individuals (i.e., about 1\% of the U.S. population) with 95\% accuracy \citep{abowd2010CensusConfidentiality2023}, although an actual attacker would not be able to verify their level of accuracy or identify which individuals are vulnerable since they do not have access to the confidential data like the USCB does. 
As discussed in Section~\ref{sec:invariantRequirment}, it is partly from these observations that the bureau concluded an urgent need to revamp their swapping-based SDC. However, 
while the bureau's later experiments further suggest that the rate of swapping must be significantly increased to achieve what it deemed as an acceptable level of protection 
\citep{abowd2023confidentiality},
the question remains open: in what ways does imposing a specific set of invariants impact the disclosure risk of the resulting data product?%

To understand the relationship between swapping's invariants and disclosure risk, two caveats are worth noting at the outset. First, the degree of vulnerability %
to a 
reconstruction attack 
is a measure of \textit{absolute disclosure risk} \citep{duncan1986disclosure,reiter2005estimating}, defined as the degree of certainty with which an attacker can make inferences about confidential information from the published data. Unless strong assumptions about the attacker's prior knowledge are made, DP does not directly translate into any quantifiable degree of control over the absolute disclosure risk; see, for example, \citet{dworkDifficultiesDisclosurePrevention2010, kiferNoFreeLunch2011, mcclure2012differential, hotz2020assessing}. Similarly, absolute measures of reconstruction attack success (such as percentages of successfully reconstructed records, as reported by the USCB) are not controlled by DP \citep{kenny2021use, francisNoteMisinterpretationUS2022}. %
Second, invariants are not unique to swapping, nor should they be viewed as a static byproduct. %
The final choice of invariants used in the 2020 TDA was arrived at by the USCB through an iterative process.
For example, 
block-level populations were once considered as an invariant, but were ultimately not included 
\citep{ashmead2019effective, abowd2022topdown, kiferDesignPrinciplesTopDown2019}. %

Notwithstanding these caveats, it is worthwhile to inquire, to the extent possible, about the impact of invariants on disclosure risk through the lens of DP. Such an inquiry can be challenging within the standard DP paradigm, because the impact of invariants cannot be captured by the PLB. 
By contrast, our system of DP specifications is more dexterous because invariants can be explicitly included through the multiverse $\scD$. An analysis of the impact of invariants on disclosure risk therefore amounts to a five-dimensional comparison between alternative DP specifications that differ on $\scD$ and potentially on other building blocks as well. A comprehensive description of the five-way dynamics remains open for future research, although investigations with a restricted scope (e.g., only varying two or three dimensions at a time, rather than all five) can already be informative. %
For example, it can be shown that reconstruction attacks can be increasingly successful if applied to DP-protected data when more invariants are imposed on them \citep{protivashReconstructionAttacksAggressive2022}. Analysis in Section~\ref{sec:whatif} also indicates that the granularity of swapping's invariants has a large numerical impact on the PLB $\epsilon$ (through the largest stratum size $b$), while the swap rate has comparatively little influence. This %
suggests that a reduction of invariants may have a larger impact on SDC protection compared to an increase in the swap rate---at least when protection is measured by a DP flavor. 

Finally, crude comparisons can be made between swapping and the TDA based on their vulnerability to reconstruction attacks. From experiments in \citet{abowd2010CensusConfidentiality2023}, swapping with a high swap rate of 50\% is roughly comparable in its protection against reconstruction attacks as the TDA (using the 2020 production settings). This suggests that the reduction in invariants from the 2010 data swapping method's invariants (which are numerous) to the TDA's invariants (which are few) is approximately equivalent to an increase in PLB from $\epsilon = 15.19$ at the household level to $\epsilon = 52.83$ (with $\delta=10^{-10}$) at the person level.

\subsection{Mitigating the Impact of Invariants on Disclosure Risk}\label{secMitigateInvariants}

Because some small number of invariants is frequently mandated, while a large number may have an adverse impact on disclosure risk,
statistical agencies 
need methodologies that allow for the specification of invariants in a flexible and precise manner. %
To this end, swapping---as instantiated either in 2010 or in our work---does not suffice because its invariants are largely hardwired into its mechanics. %
However, %
several extensions to 
data swapping 
enable more customization in the choice of invariants, %
thereby allowing for a better balance between SDC protection and accuracy targets. (On this topic, it is also worth noting that the TDA allows for a range of invariant choices.)

One such extension is \textit{probabilistic unit matching}, which was considered by the USCB as part of its comparative analysis between data swapping and the TDA. Instead of using the swapping key to form hard strata within which swapping is confined, this method permits units across different strata to be swapped with a small probability, which could be inversely proportional to some distance metric on the strata. 
For example,
consider using county as the swapping variable, with state and household size as the matching variables. 
Suppose that, for some $\alpha > 0$, a household chosen for swapping would have a $(1-\alpha)\%$ chance of being swapped with another household of the same size, but an $\alpha\%$ chance of being swapped with a differently sized household. Doing so retains the countywide household counts as invariant, but the countywide total populations would no longer be invariant. 

Two other approaches to remove some of swapping's invariants are 
\textit{pre-}  and  \textit{post-swap perturbation}.  As their names suggest, the former infuses noise into the confidential records prior to applying swapping 
\citep[p. 23]{hawesDeterminingPrivacylossBudget2021}, whereas the latter perturbs an intermediate data product after applying swapping. Notably, data swapping followed by tabular perturbation is a common SDC strategy; for example, it is the approach taken by the Office of National Statistics (ONS) for the protection of the 2021 UK Census \citep{officefornationalstatisticsonsProtectingPersonalData2023}. In this case, the cell key method (CKM) \citep{fraserProposedMethodConfidentialising2005, thompsonMethodologyAutomaticConfidentialisation2013, marleyMethodConfidentialisingUserdefined2011} is employed to perturb the cells of contingency tables after targeted swapping has been applied to the underlying microdata.  %
In addition to its use at ONS, applying swapping and then CKM perturbation is also recommended by Eurostat's Centre of Excellence on Statistical Disclosure Control for EU censuses \citep{glessingRecommendationsBestPractices2017}.

That the CKM has already been analyzed through the lens of DP \citep{rinottConfidentialityDifferentialPrivacy2018, bailieABSPerturbationMethodology2019, sadeghiConnectionABSPerturbation2023} suggests the possibility that its use as a post-swap perturbation method may deliver a formal guarantee of protection that is stronger than that provided by swapping or CKM alone. Nevertheless, we leave to future work a full investigation of the protections supplied by probabilistic matching and %
by swapping combined with pre- or post-swap perturbation.
Note that compared to standard swapping algorithms such as the PSA, all three of the above procedures introduce strictly more auxiliary randomness into the data product. It would therefore be reasonable to expect the resulting algorithms to enjoy DP guarantees while supplying fewer and more flexible choices of invariants. One salient question for this line of research is to determine the DP specification for the `chaining' (i.e., sequential composition) of two mechanisms (e.g., swapping followed by tabular perturbation), when these mechanisms satisfy different DP specifications. %

\vspace{1em}

\bookmarksetup{startatroot}
\subsection*{Disclosure Statement}
JB gratefully acknowledges partial financial support from the Australian-American Fulbright Commission and the Kinghorn Foundation; RG and XLM acknowledge partial financial support from the NSF; and XLM acknowledges partial financial support from Harvard University's Office of the Vice Provost for Research.

\bookmarksetup{startatroot}
\subsection*{Acknowledgments}
We thank Cory McCartan for his assistance with the 2010 U.S. Census data, as well as for many stimulating discussions; and Xiaodong Yang, Nathan Cheng, and Souhardya Sengupta for their help in proving Lemma \ref{lemmaDSteps}. We are grateful to the participants of the National Bureau of Economic Research's conference \emph{Data Privacy Protection and the Conduct of Applied Research: Methods, Approaches and Their Consequences} (May 4-5, 2023); the \nth{36} New England Statistical Symposium's invited session \emph{A Private Refreshment on Statistical Principles and Senses} (June 6, 2023); the 2023 Joint Statistical Meetings session \emph{Methodological Approaches to Privacy Concerns Across Multiple Domains} (August 10, 2023); the CA Census retreat at the Boston University Center for Computing and Data Science (September 26, 2023); and the Statistics Canada Methodology Seminar (October 31, 2023) for their thoughtful comments and questions. We appreciate enormously the detailed feedback provided by Daniel Kifer, John Abowd, Philip Leclerc, Ryan Cummings, Rolando Rodriguez, Robert Ashmead, Sallie Keller, and Michael Hawes at the USCB, which greatly improved and corrected all three parts of this trio of papers. All remaining errors are purely our own. 
 
\printbibliography

\appendix

\section{Other Related Work}\label{secRelatedWork}

In this appendix, we briefly review some related work that the main body of this article does not discuss in sufficient depth. Firstly, there is existing literature examining differential privacy (DP) under invariants. One branch of this literature develops DP mechanisms that report invariants without noise. In addition to the United States Census Bureau's work on the TopDown Algorithm (TDA) \citep{abowd2022topdown}, other papers in this branch include \citet{gong2020congenial}, \citet{gao2022subspace}, and \citet{dharangutte2022integer}. Because invariants can be viewed from an attacker's perspective as background knowledge, work addressing how to incorporating this knowledge into DP \citep{kifer2014pufferfish, he2014blowfish, kiferNoFreeLunch2011, desfontainesDifferentialPrivacyPartial2020b} is also relevant. In particular, \citet{seemanFormalPrivacyPartially2022} applies the Pufferfish privacy framework to construct a DP formulation that can handle invariants, although with the additional complication that the data must be modeled. Although not specifically addressing invariant-respecting DP, \citet{protivashReconstructionAttacksAggressive2022} demonstrate that related DP formulations---which, like invariants, also restrict the data universes considered by DP---may not provide sufficient SDC. More recent work on invariants can be found in \citet{choFormalPrivacyGuarantees2024} and \citet{bailieBigDataDifferential2020}. %

Many of these works, including \citet{choFormalPrivacyGuarantees2024}, condition on the specific, realized value of the invariants, which we showed in Section~\ref{secInvariants} is not a valid way to incorporate invariants into DP. Further, those works above that incorporate invariants into DP only do so for specific flavors of DP, while in this work, we use a system that can integrate invariants into any DP flavor. This unified way to handle invariants, which is missing in other work, is needed for our comparisons of the Permutation Swapping Algorithm's and the TDA's DP specifications, as these specifications vary on other dimensions, not just on their invariants.

There is also related literature studying SDC for the U.S. Decennial Census. \citet{ashmead2019effective} and \citet{kiferBayesianFrequentistSemantics2022} describe DP semantics for the 2020 Census, with the former focusing on the impact of invariants. %
\citet{abowd2010CensusConfidentiality2023} examines the 2010 DAS, using a reconstruction attack to demonstrate that aggregation did not provide SDC, as has traditionally been assumed. \citet{christDifferentialPrivacySwapping2022} compares data swapping with standard DP-based mechanisms. 
Finally, the paper that first proposed data swapping \citep{daleniusDataswappingTechniqueDisclosure1982} includes arguments for the SDC provided by data swapping, which were reviewed by \citet{fienbergDataSwappingVariations2004}.

An extensive body of literature inspired and contributed to the development of the system of DP specifications outlined in Section~\ref{secSystemDPSpecifications}. A review of this literature is given in \citet{bailieRefreshmentStirredNot2026}.

\section{Background on Data Swapping}\label{appendixBackgroundSwapping}

Invented by Dalenius and Reiss (\citeyear{daleniusDataswappingTechniqueDisclosure1978, daleniusDataswappingTechniqueDisclosure1982}) and further expanded upon by Fienberg and McIntyre (\citeyear{fienbergDataSwappingVariations2004}), data swapping (also called record swapping, particularly in Europe) refers to a family of statistical disclosure control (SDC) methods that select some subset of records and permute the values these records take for a subset of variables. These methods differ on which variables are swapped, how records are selected to be swapped, and how the interchanging of the values of the swapping variables between the selected records is conducted. (See \cite{depersioNcycleSwappingAmerican2012, kimEffectDataSwapping2015, shlomoDataSwappingProtecting2010, fienbergDataSwappingVariations2004}, for examples of different data swapping methods.) Traditionally, claims of SDC protection provided by swapping methods have been based on the intuition that a successful disclosure requires linking inferred information about a \emph{sensitive variable} to an individual entity using some \emph{quasi-identifying variables}. By sensitive variable, we mean a variable that is plausibly of interest to an attacker---for example, a person's race or a household's income. Learning the value of a sensitive variable for an individual record may not be problematic on its own since the attacker does not know to whom the record belongs. Thus, an attacker has two goals: 1) to infer the value of a sensitive variable for an individual record and 2) to determine, using quasi-identifying variables, the individual entity associated with that record. Since the sensitive variable and the quasi-identifiers must belong to the same record, the attacker needs to infer them jointly. The idea behind data swapping is to hinder such joint inference by randomly permuting the records' quasi-identifiers while keeping the sensitive variables fixed (or visa versa). In this way, there are multiple plausible values for the original data set that are compatible with the swapped data set---thereby adding uncertainty to the relationship between any record's sensitive variables and its quasi-identifiers.

It is important to emphasize that the above discussion is only an intuitive justification for data swapping. A major motivation for this article is to supplement such intuitive arguments with mathematical SDC guarantees. Some such guarantees are provided by the Permutation Swapping Algorithm's differential privacy specification. In fact, Theorem~\ref{thmSwapDP} can be interpreted as a formalization of the above intuitive argument because it provides a bound on how plausible the true confidential data set is compared to other compatible data sets. This bound ensures a degree of uncertainty in the relationship between $\Vswap$ and $\Vhold\setminus\Vswap$. Taking $\Vswap$ to be the quasi-identifiers and $\Vhold\setminus\Vswap$ the sensitive variables (or visa versa), this recovers the above argument. 
However, theoretically any set of variables can function as quasi-identifiers, depending on the attacker's auxiliary knowledge and the context of the data collection (see, e.g., \cite{sweeneySimpleDemographicsOften2000, sweeneyKAnonymityModelProtecting2002, machanavajjhalaLdiversityPrivacyKanonymity2007, cohenAttacksDeidentificationDefenses2022}). As such, arguments that rely on knowing what variables are quasi-identifiers may have limited utility outside the scope of context-specific SDC analyses.

Data swapping is widely utilized---typically in combination with other SDC methods---by statistical offices across the globe. As we remark in the main body of this article, it has been used and studied extensively by the United States
Census Bureau \citep{mckennaLegacyTechniquesCurrent2019, steelEffectsDisclosureLimitation2003, zayatzDisclosureAvoidanceCensus2010, zayatzDisclosureAvoidancePractices2007, laugerDisclosureAvoidanceTechniques2014, lemonsMeasuringDegreeDifference2015}. Further, the Office of National Statistics (ONS) has employed it for the 2001, 2011, and 2021 UK Censuses \citep{spicerStatisticalDisclosureControl2020, officefornationalstatisticsonsProtectingPersonalData2023, shlomoDataSwappingProtecting2010}. It is also one of the two protection methods recommended by Eurostat's Centre of Excellence on Statistical Disclosure Control \citep{glessingRecommendationsBestPractices2017} and was used (or is intended to be used) for protecting census data by 15 of 30 European Union states surveyed by \citet{devriesOverviewUsedMethods2023}. The Australian Bureau of Statistics uses it as one of their primary SDC methods for releasing microdata \citep{australianbureauofstatisticsTreatingMicrodata2021}, and it has been explored as a method for protecting the Japanese Population Census \citep{itoDataSwappingMore2014}.

While we largely focus on the swapping procedure used in the 2010 U.S. Decennial Census, much of this article applies to other statistical agencies, especially when their swapping mechanisms are similar to the Disclosure Avoidance System (DAS) used in 2010. In particularly, the ONS's Targeted Record Swapping \citep{ukstatisticsauthorityTransparencySDCMethods2021} closely aligns with the 2010 DAS and hence the current article is also relevant for the 2021 UK Census.	

\section{The 2010 U.S. Census Disclosure Avoidance System}\label{appendix2010Swap}

This appendix collates information about the 2010 Disclosure Avoidance System (DAS) that have been made public by the United States Census Bureau (USCB). Most of this information also applies to the 2000 DAS---as it was very similar to the 2010 DAS---but likely not to the 1990 DAS, which used a significantly different data swapping procedure \citep{mckennaDisclosureAvoidanceTechniques2018}.

The main references are \citet{mckennaDisclosureAvoidanceTechniques2018}, \citet{ mckennaLegacyTechniquesCurrent2019} and \citet{abowdDeclarationJohnAbowd21}, with additional information spread across various other USCB publications \citep{zayatzDisclosureAvoidanceCensus2010, zayatzDisclosureLimitationCensus2003, hawalaProducingPartiallySynthetic2008, u.s.censusbureau2010DemonstrationData2022, u.s.censusbureauComparingDifferentialPrivacy2021, hawesUnderstanding2020Census2021, zayatzDisclosureAvoidancePractices2007, laugerDisclosureAvoidanceTechniques2014, garfinkelFormalPrivacyMaking2019a, hawesCensusBureauSimulated2021, steelEffectsDisclosureLimitation2003, lemonsMeasuringDegreeDifference2015}. However, the publicly available documentation on the 2010 DAS is deliberately incomplete as some implementation details have been deemed confidential by the USCB due to concerns that they may allow the privacy protections of the 2010 DAS to be undermined.  We are not the only researchers external to the USCB who have attempted to reproduce the 2010 DAS \citep{kimEffectDataSwapping2015, radwayImpactDeIdentificationSingleYearofAge2023, christDifferentialPrivacySwapping2022, keyesHowCensusData2022}; however, we believe this documentation is the most comprehensive of those that are currently publicly available.

The primary protection method of the 2010 DAS was data swapping. Special tabulations had additional rules-based protections (see Appendix~A of \citeauthor{mckennaDisclosureAvoidanceTechniques2018}, \citeyear{mckennaDisclosureAvoidanceTechniques2018} for these rules). Synthetic data methods were used to protect the confidentiality of group quarters (GQs) since swapping was infeasible for GQs due to their sparsity and the consequent lack of matching records \citep{hawalaProducingPartiallySynthetic2008}. These synthetic data methods involved replacing some GQ data with predicted values from a generalized linear model \citep[Section~6.5]{mckennaDisclosureAvoidanceTechniques2018}. 

The data swapping procedure for the 2000 and 2010 DAS had three main steps: 
\begin{enumerate}[label=Step \arabic*:, leftmargin=\widthof{[Step-III]}+\labelsep]
\item A random set $S$ of household records was selected.
\item Each record in $S$ was paired with a similar, nearby household.
\item The location of each household in $S$ was swapped with the location of its pair.
\end{enumerate}

We will describe each of these steps in detail below. The data swapping procedure was applied only to households \citep[i.e., `occupied housing units',][]{2010CensusSummary} and not to unoccupied housing units or group quarters. At the end of step 3, the DAS swapping procedure outputs a data set ---called the post-swapped data set---which differs from its input (the Census Edited File [CEF]) only on the locations of the selected households and their pairs. All publications from the 2010 Census were derived from this post-swapped data set \citep{zayatzDisclosureAvoidanceCensus2010, laugerDisclosureAvoidanceTechniques2014}. (The post-swapped data set is called the Hundred-percent Detailed File by the USCB.)

{\bf Step 1}: Household records were randomly selected into the set $S$ with a probability that, for 2010, depended on (possibly among other factors): 
\begin{enumerate}[label = (\Alph*)]
\item The size of the household's block (larger blocks decreased the probability of selection)
\item Whether the household contained individuals of a race category not found elsewhere in its block (unique race categories increased the probability of selection)
\item The imputation rate within the household's block (higher imputation rates decreased the probability of selection)
\item Whether the household was unique within their geographical area on some set of variables (such households were always included in $S$). (It is not clear what geographical area was used, but we speculate that it may have been either the household's block group, tract, or county.)
\item Whether the household record was imputed and what proportion of the record was imputed \citep{mckennaDisclosureAvoidanceTechniques2018, abowdDeclarationJohnAbowd21, mckennaLegacyTechniquesCurrent2019}.
\end{enumerate}
Note that (at least in 2000) the selection of records into $S$ was not mutually independent. This was because the number of records in $S$ was capped so that the proportion of swapped records (i.e., the swap rate) was controlled at prespecified thresholds at the state level \citep{steelEffectsDisclosureLimitation2003}. (The swap rates for each state were approximately equal \citep{steelEffectsDisclosureLimitation2003}.) There may have also been other dependencies between households' selection into $S$.

Exactly how a household's probability of selection was calculated is not public information. However, the USCB has confirmed that in 2010, the marginal selection probability (unconditional on other selections) was zero for totally imputed households, and was nonzero for all other households \citep{abowdDeclarationJohnAbowd21, mckennaDisclosureAvoidanceTechniques2018, hawesCensusBureauSimulated2021}.\footnote{However, this appears to be contradicted by another statement from the USCB: ``there was a threshold value for not swapping in blocks with a high imputation rate'' \citep{mckennaLegacyTechniquesCurrent2019}. Assuming that this imputation rate threshold was under 100\%, there would be not-totally-imputed household records with zero selection probability.

There is a possible explanation of this contradiction. All not-totally-imputed households may have had the possibility of being swapped (in Step~2) even though some of them had zero probability of being selected into $S$. Yet this would require that, for all the not-totally-imputed households $h$ with zero selection probability, there was a household $h'$ with nonzero selection probability that matched $h$ on the five criteria in Step~2 and furthermore that $h$ and $h'$ could possibly be matched given the DAS's prioritization of certain matches over other matches (e.g., 1.--3. in Step~2). It seems infeasible to guarantee such requirements for all possible CEFs.}

{\bf Step 2}: For each household record in $S$, the DAS swapping procedure found a household that
\begin{itemize}%
\item had the same number of adults (over 18 years of age);
\item had the same number of minors (under 18 years of age);\footnote{As a consequence, the paired housing units also matched on the occupancy status (occupied versus unoccupied), and total number of persons.}
\item had the same tenure status;\footnote{The 2010 Census classified households' tenure as either owner-occupied (owned outright), owner-occupied (with a loan or mortgage), renter occupied, or occupied without payment of rent. It is unclear if the swapping procedure matched households on these categories, or only on the broader categories of A. owner-occupied vs B. renter-occupied (including without payment of rent) \citep{2010CensusSummary, 2010SampleCensusForm}.}
\item was located in the same state; but
\item was located in a different block \citep{abowdDeclarationJohnAbowd21, garfinkelFormalPrivacyMaking2019a, u.s.censusbureauComparingDifferentialPrivacy2021}.\footnote{The public documentation from the USCB is contradictory on whether there were additional requirements beyond the five listed here \citep{hawesCensusBureauSimulated2021, abowdDeclarationJohnAbowd21, u.s.censusbureau2010DemonstrationData2022}.}
\end{itemize}
A household that satisfies these requirements is called a matching household. Each record in $S$ was paired with a matching household. In 2000 (and hence plausibly in 2010 as well), the swapping procedure prioritized pairings where
\begin{enumerate}
\item the matching record was also in $S$; or
\item both records were geographically close (e.g., they were in the same tract or county); or
\item the matching record had a high ``disclosure risk'' \citep{steelEffectsDisclosureLimitation2003}.
\end{enumerate}
It is possible that there were other criteria for deciding the pairing when there were multiple matching households. It is unclear how these criteria were ranked in their importance. (For example, how did the swapping procedure decide between I. a pair where both records were in $S$ but were in different counties; and II. a pair where both records were in the same block group but only one record was in $S$?) However, it is likely that criteria 1. was considered the most important since it minimizes the number of swaps \citep{steelEffectsDisclosureLimitation2003}.

{\bf Step 3}: Steps 1 and 2 produced pairs of household records. These pairs consisted of one record from $S$ along with its matching record found in Step 2 (which may also be in $S$). In Step 3, all pairs had their locations swapped. More exactly, for each household in $S$, the value of its block, block group, tract, and county were swapped with the corresponding values of its paired record. (Note that a pair of records might have had the same block group, tract, or county, in which case these values did not change. \citeauthor{garfinkelFormalPrivacyMaking2019a}, \citeyear{garfinkelFormalPrivacyMaking2019a} states that the paired households were always in the same state, so this location variable was never swapped.)

\subsection{Comparing the 2010 DAS With the PSA}\label{appendixCompareDAS2010Swapping}

In this section, we compare the 2010 data swapping procedure with the PSA. The PSA is a general algorithm (in the sense that its parameters---such as the swapping and matching variables---are not set but must be chosen). Thus, for the purposes of this comparison, we will consider the PSA using the implementation choices that attempt to mirror the 2010 DAS, as given in Section~\ref{sec:2010PrivacyLoss}.

There are a number of key similarities between the data swapping procedure in the 2010 DAS and the PSA from Section~\ref{sec:2010PrivacyLoss}:
\begin{enumerate}
\item The \emph{swapping units} (i.e., the records that are swapped) are household-records for both the 2010 DAS and the PSA.

\item \emph{Swap rates:} The swap rate is defined as the fraction of records that were swapped. For the 2010 DAS, the swap rate is the fraction of records that were selected into $S$ or were paired with a record in $S$. This rate was explicitly controlled by the USCB at the state level and all states had approximately the same swap rate  \citep{zayatzDisclosureLimitationCensus2003}. Although the USCB has not released the value of the 2010 DAS's swap rate, at the national level it is purported to be between 2-4\% \citep{boyd2022differential}.

In comparison, the PSA controls the expected swap rate (where the expectation is over the randomness in the PSA). An implementer of the PSA cannot precisely fix the swap rate---but only the expected swap rate (via the parameter $p$). However, when the number of records $n$ is large, the swap rate is typically very close to $p$, since its variance is approximately $p(1-p)/n \approx 0$. 

Hence, one may set the PSA's parameter $p$ so that the swap rates for the PSA and the 2010 DAS are similar at the state and national levels.

\item The \emph{matching variables} of the 2010 DAS include the household's state, the number of adult occupants, the number of child occupants, and the household's tenure status. There may be other matching variables (which have not been disclosed by the USCB), but \citet{abowdDeclarationJohnAbowd21} implicitly suggests that this is not the case. The PSA could be implemented with exactly the same matching variables. However, the matching variables of the PSA implementation in Section~\ref{sec:2010PrivacyLoss} are the household's state and counts of adults and children---the household's tenancy status was not included. By excluding a matching variable, the PSA from Section~\ref{sec:2010PrivacyLoss} has fewer invariants and its protection loss budget (PLB) $\epsilon$ is a conservative estimate, compared to a PSA implementation that mirrored exactly the 2010 DAS matching variables. (The reasoning here mirrors the discussion in Section~\ref{sec:2010PrivacyLoss} regarding the exclusion of the household count of voting age persons from the swap key.)

\item The \emph{swapping variables} of the 2010 DAS and the PSA from Section~\ref{sec:2010PrivacyLoss} are the same: the households' county, tract, block group, and block are swapped by the 2010 DAS and the PSA. (As we will discuss later in this section, the 2010 DAS sometimes used the households' county, tract, or block group as matching variables in an adaptive matching procedure. For our purposes, they can still be considered as swapping variables; matching variables can always be swapped since swapping them does not change the data.)
\end{enumerate}

There are a number of significant differences between the PSA and the 2010 swapping procedure:
\begin{enumerate}
\item The 2010 DAS \emph{swapped} pairs of records, whereas the PSA \emph{permutes} multiple records. While any permutation is equal to a sequence of multiple pairwise swaps, the 2010 DAS does not allow for such arbitrary swaps. However, permutation swapping (under the name $n$-Cycle swapping) was actively being investigated by the USCB \citep{depersioNcycleSwappingAmerican2012, laugerDisclosureAvoidanceTechniques2014} before this work was supplanted by their shift toward differential privacy (DP) \citep{mckennaLegacyTechniquesCurrent2019}. The USCB found that permutation swapping provided both better data utility and better data protection than the pairwise swapping used in 1990-2010 Censuses; their second finding is corroborated by our DP analysis of permutation swapping.

\item \emph{Swap probabilities:} The probability of a given household being swapped was not constant in the 2010 DAS. In fact, swapping was highly targeted to households that were ``vulnerable to re-identification'' \citep{hawesCensusBureauSimulated2021}. Moreover, the probability of a household being swapped was dependent on whether other households were selected for swapping (for example, because the absolute statewide swap rates were controlled). In comparison, in the PSA, the probability of a household being swapped is constant and independent of other households.

\item \emph{Adaptive matching:} The 2010 DAS paired households according to a complicated matching procedure. For example, they prioritized matching households that shared the same county or tract. (More details on their matching procedure is given in Step 2 of the 2010 DAS description above.) In essence, this means that sometimes the household's county or tract were included as matching variables, and sometimes they were not; and whether they were included was a function of the household as well as its matching households. The matching procedure for the PSA is much simpler by comparison: the matching variables are static and the choice of how to swap the selected matching households is made uniformly at random.

\item \emph{Nonvacuous swaps:} A swap is vacuous if it does not change the data set, except (possibly) by reordering the records. A pairwise swap is not vacuous if and only if the paired records have different values for both their swapping variables $\Vswap$ and their holding variables $\Vhold$. It is unclear whether the 2010 DAS prohibited vacuous swaps but we suspect so. On the other hand, vacuous swaps are allowed by the PSA.
\end{enumerate}

\subsection{Modifying the PSA to Further Align With the 2010 DAS}\label{appendixModifyPSA2010}

We discuss some possible extensions to the PSA in Section~\ref{secMitigateInvariants}. Those extensions aimed to reduce the PSA's invariants without foregoing its DP guarantee. In this section, we propose four additional extensions to the PSA that also provide a DP guarantee while being more faithful to the 2010 DAS. These extensions address the differences between the 2010 DAS and the PSA identified in the previous section. We show that these differences can be bridged without losing the guarantee of DP---at the cost of greatly complicating the calculation of the PLB. We do not attempt these calculations; we only argue why the PLBs for these extensions remain bounded away from infinity.

First, we address one aspect of the 2010 DAS that cannot be incorporated into a DP swapping mechanism. %
The 2010 DAS used disjoint, pairwise swapping \citep{laugerDisclosureAvoidanceTechniques2014}. This is not a transitive action---its orbit space does not equal the universe induced by the swapping invariants---and hence it cannot satisfy differential privacy. (Roughly speaking, a necessary condition for a mechanism $T$ to be pure DP is that $\sfP \big(T(\xd, U) \in\cdot\ \big)$ and $\sfP\big( T(\xd', U) \in\cdot\ \big)$ have common support for all $\xd$ and $\xd'$ in the same data universe with $\dX(\xd, \xd') < \infty$. When the statistical disclosure control method is a random group action on $\xd$, as is the case for permutation swapping, this condition is equivalent to the group action being transitive.)

\emph{Variable Swapping Probabilities}: The PSA uses the same swapping probability $p$ for all records. However, we can modify the PSA to use a different swapping probability $p_i$ for each record $i$. As long as these probabilities are uniformly bounded away from zero and one (so that $p_i^{-1}$ and $(1-p_i)^{-1}$ are bounded away from infinity), this modification will satisfy the same DP flavor as the original PSA (with a finite budget). The proof would follow the same strategy as in Appendix~\ref{sectionAppendixProofMainTheorem}; only the final computations would change, as one would need to optimize over $o_i = p_i/(1-p_i)$ for all $i$.

\emph{Nonuniform Permutations}: The PSA samples derangements of the selected records uniformly at random, whereas the 2010 DAS prioritizes certain swaps over others. We can mirror this aspect of the 2010 DAS by sampling from a nonuniform distribution over the derangements. This would allow for some derangements to be selected with higher probability than other derangements. The advantage here is that some derangements, which result in poor data utility (such as when geographically distant records are swapped), can be undersampled; while other, more desirable derangements can be oversampled. This would mimic the adaptive matching of the 2010 DAS. By reasoning that is analogous to the previous extension, this extension will also retain the PSA's DP flavor, provided that
$\sfPx (\sigma = g)/\sfPxdash (\sigma = g')$
is uniformly bounded by %
$\exp[O(\abs{k_g - k_{g'}})]$, for all derangements $g$ and $g'$ (of $k_g$ and $k_{g'}$ records, respectively).

\emph{Prohibiting Imputed Records From Being Swapped}: The 2010 DAS never swaps records that have been completely imputed. (The rationale is that imputed records do not require privacy protection.) We can modify the PSA so that $p_i = 0$ for all records $i$ that are imputed. Suppose that the records that are imputed are constant. If the PSA satisfies \SpecSwap, then this modification would satisfy $\epsilon'_{\D}$-DP$(\Xcef,\allowbreak \DSwap,\allowbreak \dX',\allowbreak \Mult)$, where
\[\dX' (\xd, \xd') = \begin{cases}
\dHamSh (\xd, \xd') & \text{if $\xd, \xd'$ do not differ on any imputed record,} \\ 
\infty & \text{otherwise'}
\end{cases}\]
and $\epsilon_{\mathcal D}' \le \epsilon_{\mathcal D}$ since the maximum stratum size $b$ is reduced when the imputed records are removed.

\emph{Prohibiting Vacuous Swaps}: A swap (or more generally a permutation) is vacuous if it does not change the data set, except perhaps by reordering the records. %

We assume that the 2010 DAS does not allow vacuous swaps. We can similarly prohibit the PSA from allowing vacuous swaps. Instead of sampling derangements uniformly at random, we would put zero probability on vacuous derangements. Under the action of nonvacuous derangements, the orbit space is still the entire data universe. Hence, this modification will still satisfy the same DP flavor as the PSA, however, the calculation of the PLB $\epsilon_{\mathcal D}$ will be difficult as one must optimize $\sfPx (\sigma = g)/\sfPxdash (\sigma = g')$ over all permutations $g$ and $g'$ that are nonvacuous with respect to $\xd$ and $\xd'$, over all $\xd, \xd' \in \mathcal D$.

\section{Appendix to Section~\ref{secInvariants}}\label{secInvariantsProofs}

\subsection{An Additional Result on the Impact of Invariants}

\begin{proposition}\label{propNestedscD}
    Suppose that $\scD'$ is refinement of $\scD$. %
    Then, for all PLBs $\epsilon_{\D'}' : \scD' \to [0,\infty]$, any mechanism that satisfies \DPSpec\ also satisfies $\varepsilon_{\D'}'$-DP$(\X,\allowbreak \scD',\allowbreak \dX,\allowbreak \dPr)$, 
    whenever $\epsilon_{\D} \le \inf \{ \epsilon_{\D'}' : \D' \in \scD' \mathrm{\ with\ } \D' \subset \D\}$. Furthermore, for all PLBs $\epsilon_{\D} : \scD \to [0,\infty]$, any mechanism that satisfies \DPSpec\ also satisfies $\varepsilon_{\D'}'$-DP$(\X,\allowbreak \scD',\allowbreak \dX,\allowbreak \dPr)$, whenever $\epsilon'_{\D'} \ge \inf \{ \epsilon_{\D} : \D \in \scD \mathrm{\ with\ } \D' \subset \D \}$.%
\end{proposition}

This result shows how protection loss budgets (PLBs) under one flavor of differential privacy (DP) can be transferred into PLBs under another flavor, when the two flavors are related by the fact that one's multiverse is refinement of the other (such as is the case when one has more invariants than the other). It will be used in the proofs of Proposition~\ref{propNestedscD2} and Theorem~\ref{thmTDASatisfiesDP}. Its second part is a stronger and more general result than Proposition~2.7 of \citet{gao2022subspace}, which was restricted to nested linear subspaces and did not consider shrinking of the PLB.

\subsection{Proofs for Section~\ref{secInvariants}'s Results}

\begin{proof}[Proof of Proposition~\ref{propReleaseF}]
    $T$ is constant within any universe $\mathcal D \in \scD_{\bm c}$. Therefore,
    \[\dPr\left[ \sfP(T(\xd, U) \in \cdot\ ), \sfP(T(\xd', U) \in \cdot\ ) \right] = 0,\]
    for all $\xd,\xd' \in \mathcal D$. This proves the first half of the proposition. 
    
    To prove the second half, observe that $\sfP(T(\xd, U) \in \cdot\ )$ and $\sfP(T(\xd', U) \in \cdot\ )$ have disjoint support (the former's support is concentrated on $\bm c(\xd)$, the latter's on $\bm c(\xd')$). Hence their total variation distance is one. This implies 
    \[\dPr\left[ \sfP(T(\xd, U) \in \cdot\ ), \sfP(T(\xd', U) \in \cdot\ ) \right] = \infty,\]
    Because $\dX(\xd_1, \xd_2) < \infty$, the DP condition (Equation~\ref{eqTSatisfiesDPDefn}) can therefore only be satisfied if $\epsilon_{\D_0} = \infty$.
\end{proof}

\begin{proof}[Proof of Proposition~\ref{propReleaseFConverse}]
    This proposition relies on the metric axiom $\dPr(\sfP,\sfQ) > 0$ if $\sfP \ne \sfQ$. This implies 
    \[\dPr\left[ \sfP(T(\xd, U) \in \cdot\ ), \sfP(T(\xd', U) \in \cdot\ ) \right] > 0.\]
    Because $\dX(\xd, \xd') < \infty$, the DP condition can therefore only be satisfied if $\epsilon_{\D_0} > 0$. (Note that on the right hand side of Equation~\ref{eqTSatisfiesDPDefn}, we define $0 \times \infty$ to be $\infty$.)
\end{proof}

\begin{lemma}\label{lemmaMultiverseCompare}
    Given two DP specifications \DPSpec\ and $\epsilon'_{\D'}$-DP$(\X,\allowbreak \scD',\allowbreak \dX,\allowbreak \dPr)$, suppose that, for all $\D' \in \scD'$ and all $\delta > 0$, there exists $\D \in \scD$ such that $\D' \subset \D$ and $\epsilon_{\D} \le \epsilon_{\D'}' + \delta$. Then
    \[\M (\X, \scD, \dX, \dPr, \epsilon_{\D}) \subset \M (\X, \scD', \dX, \dPr, \epsilon_{\D'}').\]
\end{lemma}

In the above lemma, $\M (\X,\allowbreak \scD,\allowbreak \dX,\allowbreak \dPr,\allowbreak \epsilon_{\D})$ denotes the set of mechanisms satisfying the DP specification \DPSpec. 

\begin{proof}
    Suppose that $T \in \M (\X, \scD, \dX, \dPr, \epsilon_{\D})$. Let $\D' \in \scD'$ and $\delta > 0$. Suppose $\xd, \xd' \in \D'$. By assumption, there exists $\D \in \scD$ such that $\D' \subset \D$ and
    \[\dPr\left[ \sfP(T(\xd, U) \in \cdot\ ), \sfP(T(\xd', U) \in \cdot\ ) \right] \le \epsilon_{\D} \dX(\xd, \xd') \le (\epsilon'_{\D'}+\delta) \dX(\xd, \xd').\]
    Since $\dPr\left[ \sfP(T(\xd, U) \in \cdot\ ), \sfP(T(\xd', U) \in \cdot\ ) \right] \le (\epsilon'_{\D'}+\delta) \dX(\xd, \xd')$ holds for all $\delta > 0$, it follows that $\dPr\left[ \sfP(T(\xd, U) \in \cdot\ ), \sfP(T(\xd', U) \in \cdot\ ) \right] \le \epsilon'_{\D'} \dX(\xd, \xd')$. This proves $T \in \M (\X,\allowbreak \scD',\allowbreak \dX,\allowbreak \dPr,\allowbreak \epsilon_{\D'}')$.
\end{proof}

\begin{proof}[Proof of Proposition~\ref{propNestedscD}]
    Because $\scD'$ is a refinement of $\scD$, the assumption of Lemma~\ref{lemmaMultiverseCompare} holds for both choices of the budgets $\epsilon_{\D}$ and $\epsilon'_{\D'}$. The results then follow by this lemma.
\end{proof}

\begin{proof}[Proof of Proposition~\ref{propNestedscD2}]
    Suppose for contradiction that the result does not hold. Then there exists some $\D'_0 \in \scD'$ and some $\D_0 \in \scD$ such that $\PLDdashz < \PLDz$ and $\D_0 \subset \D'_0$.

    Given a budget $\epsilon_{\D'}'^{(i)} : \scD' \to [0,\infty]$, define the budget $\epsilon_{\D}^{(i)} : \scD \to [0,\infty]$ as follows: For $\D_0$, define $\epsilon_{\D_0}^{(i)} = \epsilon_{\D_0'}'^{(i)}$. For all other $\D_1 \in \scD$, fix a $\D'_1 \in \scD'$ (which does not depend on $i$) with $\D' \supset \D$ and define $\epsilon_{\D_1}^{(i)} = \epsilon_{\D_1'}'^{(i)}$. 

    Now, if $T$ satisfies $\epsilon_{\D'}'^{(i)}$-DP\DPFlavDiffMultiverse, then it also satisfies $\epsilon_{\D}^{(i)}$-DP\DPFlav\ by Proposition~\ref{propNestedscD}. Yet we can construct a sequence of budgets $\epsilon_{\D'}'^{(1)}, \epsilon_{\D'}'^{(2)}, \ldots$ such that 
    \begin{enumerate}
        \item $\epsilon_{\D'_0}'^{(i)}$ converges to $\PLDdashz$ as $i \to \infty$; and
        \item $T$ satisfies $\epsilon_{\D'}'^{(i)}$-DP\DPFlavDiffMultiverse\ for all $i$.
    \end{enumerate}
    Because $\PLDdashz < \PLDz$, this implies there exists some $N$ such that $T$ satisfies $\epsilon_{\D}^{(N)}$-DP\DPFlav\ with $\epsilon_{\D_0}^{(N)} < \PLDz$. This contradicts the definition of $\PLDz$.
\end{proof}

\section{Proof of Theorem~\ref{thmSwapDP}}\label{sectionAppendixProofMainTheorem}

In this appendix, we prove that Algorithm~\ref{algoPermute} (the Permutation Swapping Algorithm [PSA]) satisfies \SpecSwapGen\ %
for the value of $\epsilon_{\mathcal D}$ given in Theorem~\ref{thmSwapDP}. Assume throughout this appendix the conditions of Theorem~\ref{thmSwapDP}: that all $\xd \in \X$ share a common set of variables $\bm V$, which is partitioned into subsets $\Vswap$ and $\Vhold$; and that the PSA swaps records at the same resolution as $\dHamSr$. %

By Proposition~\ref{propInvariantsOfSwap}, we may also assume that there is exactly one matching variable, one nonmatching holding variable, and one swapping variable. For ease of exposition, we assume that each of these variables can take on a finite number of values, which we denote by $m = 1,\ldots, \mathcal M$ and $h = 1, \ldots, \mathcal H$ and $s = 1, \ldots, \mathcal S$, respectively, although the proof immediately generalizes beyond this assumption. Recall that $n_{mhs}^{\xd}$ is the count of records in $\xd$, which take the value $(m,h,s)$. Replacing a category $m,h,s$ with $\cdot$ denotes a marginal count---for example, $n_{m \cdot s}^{\xd} = \sum_{h=1}^{\mathcal H} n_{mhs}^{\xd}$. We will drop $\xd$ in the superscript and $\cdot$ in the subscript when this does not cause ambiguity.

Write $(M_i^{\xd}, H_i^{\xd}, S_i^{\xd})$ for the $i$-th record in $\xd$, so that we can write $\xd$ as the vector $[(M_i, H_i, S_i)]_{i=1}^{n}$, %
where $n = n_{\cdot\cdot\cdot}^{\xd} = \abs{\xd}$ is the number of records in $\xd$. With this notation, 
\[n_{mhs}^{\xd} = \sum_{i=1}^n 1_{M_i^{\xd} = m} 1_{H_i^{\xd} = h} 1_{S_i^{\xd} = s}.\]

Let $\ell_1^r(\xd, \xd')$ be the $\ell_1$-distance on the interior cells of the fully saturated contingency table
\begin{equation}\label{eqDefl1Distance}
	\ell_1^r(\xd, \xd') := \sum_{m,h,s} \abs{ n_{mhs}^{\xd} - n_{mhs}^{\xd'}}.
\end{equation}

\begin{lemma}\label{lemmaEll1HamDistEqual}
	$\ell_1^r(\xd, \xd') = 2 \dHamSr(\xd, \xd')$ if $\abs{\xd} = \abs{\xd'}$.
\end{lemma}

\begin{lemma}
	$\Mult$ is a metric on the space of a.e. equal random variables (over the same probability space $\mathcal T$).
\end{lemma}

\begin{proof}
	It is easy to see that $\Mult$ is symmetric and $\Mult (X, Y) = 0$ if and only if $X = Y$ a.e. All that remains is to verify the triangle inequality. Let $\{E_n\} \subset \mathcal F$ such that
	\[ \abs{ \ln \frac{\sfP (X \in E_n)}{\sfP (Z \in E_n)}} \to \Mult(X,Z),\]
	as $n \to \infty$. Then
	\begin{align*}
		\abs{ \ln \frac{\sfP (X \in E_n)}{\sfP (Z \in E_n)}} &\le \abs{ \ln [\sfP (X \in E_n)] - \ln [\sfP (Y \in E_n)]} + \abs{\ln [\sfP (Y \in E_n)] - \ln [\sfP (Z \in E_n)]} \\
		&\le \Mult (X,Y) + \Mult (Y,Z).\qedhere
	\end{align*}
\end{proof}

Recall that $\sigma_m$ is the random perturbation sampled by the PSA, which deranges the selected records in matching stratum $m$. Let $\sigma$ be the permutation defined by $\sigma(i) = \sigma_{M_i}(i)$. Since $\sigma_m$ fixes $i$ whenever $M_i \ne m$, it is the case that $\sigma = \sigma_{\mathcal M} \circ \cdots \circ \sigma_{1}$. %
(Note that $\sigma$ is a random function of the input data set $\xd$, although we leave this dependence implicit.) For a permutation $g$, write $g(\xd)$ as shorthand for the data set %
in which the the values of the swapping variables have been permuted according to $g$. %
That is, if $\xd = [(M_i, H_i, S_i)]_{i=1}^{n}$ then $g(\xd) = [(M_i, H_{i}, S_{g(i)})]_{i=1}^n$. Given an input data set $\xd$, the swapped data set $\sigma(\xd)$ generated by the PSA is denoted by $\bm Z$.

Let $\sfPx$ denote the probability %
induced by the randomness in the PSA (i.e., the randomness in selecting records and in sampling the permutation $\sigma$), taking the input data set $\xd$ as fixed. Recall that the output of the PSA is the fully saturated contingency table $C(\bm Z) = [n_{jkl}^{\bm Z}]$. %

\begin{lemma}\label{lemmaReorderDoesntChangeMult}
	If $\xd$ and $\xd'$ differ only by reordering of rows (i.e., $\dHamSr(\xd, \xd') = 0$), then 
	\[\Mult\left[\sfPx(C(\bm Z) \in \cdot\ ), \sfPxdash(C(\bm Z) \in \cdot\ )\right] = 0.\]		
\end{lemma}

\begin{proof}
	The contingency table $[n_{mhs}^{\bm Z}]$ is invariant to reordering of rows of $\bm Z$. Thus $\sfPx(C(\bm Z) \in \cdot\ ) = \sfPxdash(C(\bm Z) \in \cdot\ )$.
\end{proof}

\begin{lemma}\label{lemmaDSteps}
	Fix some data universe $\mathcal D \in \scD_{\cswap}$ and some $\xd, \xd' \in \mathcal D$ with $\dHamSr(\xd, \xd') = \Delta$. Then there exists a permutation $\rho$ that fixes exactly $n - \Delta$ records such that $C(\rho(\xd)) = C(\xd')$.%
\end{lemma}

\begin{proof}
	We have that $\Delta < \infty$ since the invariants $\cswap$ imply that all data sets in $\mathcal D$ have the same number of records. Hence the symmetric difference $\xd \ominus \xd'$ contain $2\Delta$ records, with $\Delta$ records from $\xd$ and $\Delta$ records from $\xd'$. %
	Denote the records in $\xd \ominus \xd'$ that come from $\xd$ by $\xd_0$ and the records from $\xd'$ by $\xd'_0$, so that $\xd \ominus \xd'$ is the disjoint union of $\xd_0$ and $\xd_0'$.  
	
	Without loss of generality, we may assume that there is a single matching category ($\mathcal M = 1$). (If there is more than one matching category, apply the following argument to each category separately.) Then the data set $\xd$ (disregarding the order of the records) can be represented as the matrix $C(\xd) = [n_{hs}^{\xd}]$.
	
	We will need the following result ($*$) whose proof is straightforward: For any $\xd'', \xd''' \in \X$, the matrix $C(\xd'') - C(\xd''') =
	[n_{hs}^{\xd''} - n_{hs}^{\xd'''}]$ has zero row- and column-sums if and only if $\xd'' \in \D_{\cswap}(\xd''')$. Moreover, $\xd'' \in \D_{\cswap}(\xd''')$ implies $\xd_0'' \in \D_{\cswap}(\xd'''_0)$ and $C(\xd'') - C(\xd''') = C(\xd''_0) - C(\xd'''_0)$.
	
	By the above result $(*)$, the marginal counts of $\xd_0$ and $\xd_0'$ agree: $n_h^{\xd_0} = n_h^{\xd_0'}$ and $n_s^{\xd_0} = n_s^{\xd_0'}$ for all $h$ and $s$. But the interior cells disagree: if $n_{hs}^{\xd_0} > 0$ then $n_{hs}^{\xd_0'} = 0$ (and visa versa, swapping $\xd_0$ and $\xd_0'$). Further $C(\xd_0) - C(\xd_0')$ has positive entries that sum to $\Delta$ and negative entries that sum to $-\Delta$, and zero row- and column-sums. %
	
	By construction of $\xd_0$ and $\xd_0'$, if we can permute $\xd_0$ to produce $\xd_0'$, then we can use the same permutation to produce $\xd'$ from $\xd$ (up to reordering of records). Critically, permutations of $\xd_0$ can only derange $\Delta$ records (since there are only $\Delta$ records in $\xd_0$) and indeed must derange $\Delta$ records to produce $\xd_0'$ (since there are no records in common between $\xd_0$ and $\xd_0'$). Therefore we have reduced the problem: we need to find a permutation $\rho$ (regardless of the number of records it fixes) such that $C(\rho(\xd_0)) = C(\xd_0')$.
	
	We construct this permutation $\rho$ by induction on $\Delta =\dHamSr(\xd, \xd') = \dHamSr(\xd_0, \xd_0')$. There are two %
	base cases: The case $\Delta = 1$ is vacuous since $\dHamSr(\xd, \xd') = 1$ implies that $\xd, \xd'$ are not in the same data universe. Why? If $\ell_1^r(\xd, \xd') = 2$ then $C(\xd) - C(\xd')$ only has one or two nonzero cells. But this implies $C(\xd) - C(\xd')$ has a row or column with nonzero sum.
	
	For $\Delta = 2$, the result ($*$) implies that the $2 \times 2$ top-left submatrix of $\bm A = C(\xd_0) - C(\xd'_0)$ looks like
	\[\bm A_{1:2,1:2} = \begin{bmatrix}
		1 & -1 \\ -1 & 1
	\end{bmatrix},\]
	(up to reordering of rows and columns). Therefore (up to reordering of records), $\xd_0$ and $\xd'_0$ differ by a single swap: if $h,h',s,s'$ are indices such that $A_{hs}= A_{h's'} = 1$ then define $\rho$ to be the swap of the records $(k,l)$ and $(k',l')$ in $\xd_0$. We have $C(\rho(\xd_0)) = C(\xd'_0)$ as desired.

	This completes the base cases. Now we will prove the induction step. By ($*$), we can always reorder the rows and columns of $\bm A = C(\xd_0) - C(\xd'_0)$ such that the $2 \times 2$ top-left submatrix looks like
	\[\bm A_{1:2,1:2} = \begin{bmatrix}
		A_{11} & A_{12} \\
		A_{21} & A_{22}		
	\end{bmatrix},\]
	with $A_{11}, A_{22} > 0$ and $A_{21} < 0$. Define $\xd_1$ by swapping the records $(1,1)$ and $(2,2)$ in $\xd_0$. Then the top-left submatrix of $\bm A' = C(\xd_1) - C(\xd'_0)$ looks like
	\[\bm A_{1:2,1:2}' = \begin{bmatrix}
		A_{11}-1 & A_{12}+1 \\
		A_{21}+1 & A_{22}-1		
	\end{bmatrix},\]
	and the rest of $\bm A'$ is the same as $\bm A$. If $A_{12} < 0$, then $\ell_1^r (\xd_1, \xd_0') = \ell_1(\bm A') = \ell_1(\bm A) - 4$. If $A_{12} \ge 0$ then $\ell_1(\bm A') = \ell_1(\bm A) - 2$. In both cases, we can use the induction hypothesis to give us a permutation $\rho_1$ of $\xd_1$, which produces $\xd_0'$ (up to reordering of records). Define the permutation $\rho$ as the composition of $\rho_1$ with the swap of $(1,1)$ and $(2,2)$. Then $C(\rho(\xd_0)) = C(\xd_0')$ as desired.
\end{proof}

\begin{proof}[Proof of Theorem \ref{thmSwapDP}]
	Fix $\xd$ and $\xd'$ in the same data universe $\mathcal D \in \scD_{\cswap}$. Let $\Delta = \dHamSr(\xd, \xd')$.
	We need to prove that $\Mult[\sfPx(C(\bm Z)), \sfPxdash(C(\bm Z))] \le \Delta \epsilon_{\mathcal D}$ or equivalently
	\begin{equation*}
		\sfPx \left[C(\sigma(\xd)) = C(\bm z) \right] \le \exp(\Delta\epsilon_{\mathcal D}) \sfPxdash \left[C(\sigma(\xd')) = C(\bm z)\right],
	\end{equation*}
	for all possible swapped data sets $\bm z$, %
	where the probability is over the random permutation $\sigma$ sampled by the PSA. Since the output $C(\bm Z)$ does not depend on the ordering of the records in the input $\xd$, we may without loss of generality reorder the records in $\xd'$. Hence, there exists a permutation $\rho$, which fixes exactly $n - \Delta$ records such that $\rho(\xd') = \xd$ by Lemma~\ref{lemmaDSteps}. %
	
	Since
	\[\sfPx \left[C(\sigma(\xd)) = C(\bm z) \right] = \sum_{\bm z'} \sfPx \left[\sigma(\xd) = \bm z' \right],\]
	where the sum is over data sets $\bm z'$ with $\dHamSr(\bm z, \bm z') = 0$, it suffices to show
	\begin{equation}\label{eqToProveThm0}
		\sfPx \left[\sigma(\xd) = \bm z \right] \le \exp(\Delta\epsilon_{\mathcal D}) \sfPxdash \left[\sigma(\xd') = \bm z\right],
	\end{equation}
	for all possible swapped data sets $\bm z$.
	
	Recall 
	\[b = \max \{0, n_{m\cdot\cdot} \mid \mathrm{there\ are\ two\ records\ with\ different\ values\ in\ matching\ stratum\ } m\}.\]
	If $b = 0$, then $\xd$ and $\xd'$ only differ by reordering of rows and hence $\epsilon_{\mathcal D} = 0$ satisfies the differential privacy (DP) condition (Equation~\ref{eqToProveThm0}) by Lemma~\ref{lemmaReorderDoesntChangeMult}. Having taken care of the case $b = 0$, from herein we may assume $b \ge 2$. (The case $b = 1$ is not possible.)
	
	If $p \in \{0,1\}$, then $\epsilon_{\D} = \infty$ and the DP condition holds vacuously.

All that remains is to prove Equation~\ref{eqToProveThm0} holds in the case where $0 < p < 1$. %
Since $\xd$ and $\xd'$ themselves differ by the permutation $\rho$, %
we can permute $\xd$ to produce $\bm z$ if and only if we can permute $\xd'$ to produce $\bm z$. %
Thus, either $\sfPx(\sigma(\xd) = \bm z)$ and $\sfPxdash(\sigma(\xd') = \bm z)$ are both zero, or they are both nonzero. We need only focus on the case where both probabilities are nonzero. 

Recall that any permutation $\sigma$ selected with nonzero probability by the PSA can be decomposed as $\sigma = \sigma_{\mathcal M} \circ \ldots \circ \sigma_1$, where $\sigma_m$ will leave any unit $i$ with matching category $M_i \ne m$ fixed. Write $\xd_m$ for the records of $\xd$ with $M_i = m$. Because we perform random selection and permutation independently for each stratum $m$,
\[\frac{\sfPx(\sigma(\xd)= \bm z)}{\sfPxdash(\sigma(\xd')= \bm z)}
= \frac{\prod_{m=1}^{\mathcal M}\sfPx(\sigma_m(\xd_m)= \bm z_m)}{\prod_{m=1}^{\mathcal M}\sfPxdash(\sigma_m(\xd_m')= \bm z_m)}.\]
Thus, to prove Equation~\ref{eqToProveThm0} it suffices to show
\begin{equation}\label{eqToProveThm}
	\frac{\sfPx(\sigma_m(\xd_m)= \bm z_m)}{\sfPxdash(\sigma_m(\xd_m')= \bm z_m)} \le \exp (\Delta_m \epsilon_{\mathcal D}),
\end{equation}
for all $m$ where $\Delta_m = \dHamSr(\xd_m, \xd_m')$. 

Fix some $m$. For notation simplicity, whenever it is not essential to indicate the role of $m$, we will drop the subscript $m$ from herein (until the end of the proof when we need to optimize over $m$). (This is the same as assuming $\Vmat$ is empty.)

Let $G_{\xd \to \bm z} = \{ \text{permutation } g : g(\xd) = \bm z\}$. We use the notation $g$ instead of $\sigma$ to emphasize that $g$ is not random, while the permutation $\sigma$ chosen by Algorithm \ref{algoPermute} is random. There is a bijection between $G_{\xd \to \bm z}$ and $G_{\xd' \to \bm z}$ given by $g \mapsto g \circ \rho$. %
Since 
\[\sfPx (\sigma (\xd) = \bm z) = \sum_{g \in G_{\xd \to \bm z}} \sfPx (\sigma = g),\]
we will prove (\ref{eqToProveThm}) by showing
\[ \sfPx (\sigma = g) \le \exp(\Delta \epsilon_{\mathcal D}) \sfPxdash (\sigma = g \circ \rho),\]
for all $g \in G_{\xd \to \bm z}$. (Note that this may not obtain the best possible bound for specific $\xd$ and $\xd'$, but it is mathematically easier to bound ${\sfPx (\sigma = g)}/{\sfPxdash (\sigma = g \circ \rho)}$ than bound the desired ratio 
\[\frac{\sum\limits_{g \in G_{\xd \to \bm z}} \sfPx (\sigma = g)}{\sum\limits_{g \in G_{\xd \to \bm z}}\sfPxdash (\sigma = g \circ \rho)}\]
directly. Yet in the case where $G_{\xd \to \bm z}$ and $G_{\xd' \to \bm z}$ are singletons, this approach gives tight bounds.)

Let $k_{g}$ be the number of records (in category $m$) that were deranged (i.e., not fixed) by $g$ and let $d(k)$ denote the $k$-th derangement number (i.e., the number of derangements of size $k$):
\begin{alignat}{2}
	d(k) &= k! \sum_{j=0}^k \frac{(-1)^j}{j!} \notag\\
	&= k d(k-1) + (-1)^k \qquad &\text{ for } k \ge 0.\label{eqRecurrence}
\end{alignat}

Fix $g \in G_{\xd \to \bm z}$ and $g' = g \circ \rho$. We now compute $\sfPx (\sigma = g)$. The permutation $g$ is sampled in Algorithm \ref{algoPermute} via a two-step procedure. Firstly records are independently selected for derangement with probability $p$. Suppose that $g$ deranges records $\{i_1,\ldots,i_{k_g}\}$. Since we disallow the possibility of selecting only one record,
\[ \sfPx (\text{the selected records are } \{i_1,\ldots,i_{k_g}\}) = \frac{p^{k_g} (1-p)^{n - k_g}}{1-\sfPx(\text{exactly 1 record selected})}.\]
Secondly we sample uniformly from the set of all derangements of $k_g$ records. Hence, we sample $g$ with probability $[d(k_g)]^{-1}$and therefore,
\begin{equation*}%
	\sfPx (\sigma = g) = \frac{p^{k_g} (1-p)^{n - k_g}}{[1-\sfPx(\text{exactly 1 record selected})] d(k_g)}.
\end{equation*}
This gives
\begin{equation}\label{eqProbRatio}
	\frac{\sfPx (\sigma = g)}{\sfPxdash (\sigma = g')} = o^\delta \frac{d(k_g-\delta)}{d(k_g)},
\end{equation}
where $o =p/(1-p)$ and $\delta = k_g - k_{g'}$. 

Our aim is now to bound the right hand side of (\ref{eqProbRatio}) by $\exp(\Delta\epsilon_{\mathcal D})$. Since $g'$ and $g$ differ only by the permutation $\rho$ (which fixes $n - \Delta$ records), we must have $k_g - \Delta \le k_{g'} \le k_g + \Delta$. Therefore, there are at most $2\Delta+1$ possible cases: 
\begin{align*}
	\delta \in S &= \left\{ \delta \in \mathbb Z \mid  -\Delta \le \delta \le \Delta \text{ and } \big( k_g - \delta = 0 \text{ or } 2 \le k_g - \delta \le n\big)\right\} \\
	&= \left\{ \delta \in \mathbb Z \mid  \max(-\Delta, k_{g}-n) \le \delta \le \min (\Delta, k_g) \text{ and } \delta \ne k_g - 1 \right\}.
\end{align*}

Suppose $0 < p \le 0.5$. Since $d(k)$ is nondecreasing (except at $k = 1$, which is not realizable by $g$ or $g'$) and $(1-p)/p \ge 1$, the right hand side of (\ref{eqProbRatio}) is maximized when $k_{g_m'} = n_m$ and $k_{g_m} = n_m- \Delta_m$ (i.e., $\delta = - \Delta_m$), in which case 
\begin{align}
	\frac{\sfPx (\sigma = g)}{\sfPxdash (\sigma = g')} &= o^{-\Delta} \prod_{m=1}^{\mathcal M} \frac{d(n_m)}{d(n_m-\Delta_m)} \notag\\
	&\le o^{-\Delta} \prod_{m=1}^{\mathcal M} (n_m+1)^{\Delta_m} \notag\\
	&\le o^{-\Delta} (b+1)^{\Delta} \notag\\
	&= \exp (\Delta \epsilon_{\mathcal D}),\label{eqProofMainResultPLessOneHalf}
\end{align}
for $\epsilon_{\mathcal D} = \ln (b+1) - \ln o$. (The second line uses Lemma \ref{lemmaBoundDisarrangement}, which is given below this proof.)

Now suppose $0.5 < p < 1$. In the case of $\delta_m = \Delta_m$, the ratio (\ref{eqProbRatio}) is maximized at $o^{\Delta_m}$ when $k_{g_m} = \Delta_m = 2$. Moreover, $o^{\Delta_m}$ also dominates $o^{\delta_m} \frac{d(k_{g_m}-\delta_m)}{d(k_{g_m})}$ for all $0 \le \delta_m \le \Delta_m$ and all possible $k_{g_m}$. Thus,
\begin{align*}
	\frac{\sfPx (\sigma = g)}{\sfPxdash (\sigma = g')} &\le \prod_{m=1}^{\mathcal M} \max \left\{ o^{\Delta_m}, o^{\delta_m} \frac{d(k_{g_m} - \delta_m)}{d(k_{g_m})} : \delta_m \in S_m \text{ and } \delta_m < 0 \right\} \\
	&\le \prod_{m=1}^{\mathcal M} \max \left\{ o^{\Delta_m}, o^{\delta_m} (k_{g_m}-\delta_m + 1)^{-\delta_m} : \delta_m \in S_m \text{ and } \delta_m < 0 \right\} \\
	&\le \prod_{m=1}^{\mathcal M} \max \left\{ o^{\Delta_m}, o^{-\delta_m} (n_m + 1)^{\delta_m} : 0 < \delta_m \le \Delta_m \right\} \\
	&\le \max \left\{ o^{\Delta}, o^{-\delta} (b + 1)^{\delta} : 0 < \delta \le \Delta \right\}.
\end{align*}
If $o^{-1}(b+1) \ge 1$ then $o^{-\delta}(b + 1)^{\delta}$ is maximised at $\delta = \Delta$. Otherwise $o^{-\delta}(b + 1)^{\delta} < 1 < o^{\Delta}$. Hence
\begin{equation}\label{eqProofMainResultPGreaterOneHalf}
	\frac{\sfPx (\sigma = g)}{\sfPxdash (\sigma = g')} \le \exp (\Delta \epsilon_{\mathcal D}),
\end{equation}
for $\epsilon_{\mathcal D} = \max \big\{\hspace{-0.15em} \ln o, \ln (b+1) - \ln o \big\}$. Combining Equations~\ref{eqProofMainResultPLessOneHalf} and~\ref{eqProofMainResultPGreaterOneHalf}, we have 
\begin{equation}\label{eqProofMainResultNearlyThere}
	\epsilon_{\D} = \begin{cases}
		\ln (b+1) - \ln o &\mathrm{if\ } 0 < p \le 0.5 \mathrm{\ and\ } b > 0,\\
		\max \big\{\hspace{-0.15em} \ln o, \ln (b+1) - \ln o \big\} &\mathrm{if\ } 0.5 < p < 1 \mathrm{\ and\ } b > 0.
	\end{cases}
\end{equation}
When $b > 0$, we have $b \ge 2$ and hence also $\max \big\{\hspace{-0.15em} \ln o, \ln (b+1) - \ln o \big\} = \ln (b+1) - \ln o$ for $0.5 < p \le \sqrt{b+1}/(\sqrt{b+1}+1)$. Thus, Equation~\ref{eqProofMainResultNearlyThere} simplifies to 
\[\epsilon_{\D} = \begin{cases}
	\ln (b+1) - \ln o &\mathrm{if\ } 0 < p \le \frac{\sqrt{b+1}}{\sqrt{b+1}+1} \mathrm{\ and\ } b > 0,\\
	\ln o &\mathrm{if\ } \frac{\sqrt{b+1}}{\sqrt{b+1}+1} < p < 1 \mathrm{\ and\ } b > 0.
\end{cases}\]
as required.
\end{proof}

\begin{lemma}\label{lemmaBoundDisarrangement}
For any $k \in \mathbb N$ and any $a \in \mathbb N$ satisfying $0 \le a \le k$ and $a \ne k - 1$,
\[\frac{d(k)}{d(k-a)} \le (k+1)^a,\]
where $d(k)$ is the number of derangements of $k$ elements (see Equation~\ref{eqRecurrence}).
\end{lemma}

\begin{proof}
We use induction on $k$. The base cases $k = 0,1,2$ are straightforward to verify since $d(0) = d(2) = 1$ and $d(1) = 0$. For the induction step, we can assume $k \ge 3$ so that $d(k-1) \ge 1$ and hence 
\begin{align*}
	\frac{d(k)}{d(k-a)} &= \frac{d(k)}{d(k-1)} \frac{d(k-1)}{d(k-a)} \\
	&\le \frac{d(k)}{d(k-1)} k^{a-1}
\end{align*}
by the induction hypothesis. The result then follows by the identity given in Equation~\ref{eqRecurrence}:
\begin{align*}
	\frac{d(k)}{d(k-1)} &= \frac{k d(k-1) + (-1)^k}{d(k-1)} \\
	&\le k+1.\qedhere
\end{align*}
\end{proof}

\section{Optimality of Theorem~\ref{thmSwapDP}}\label{appendixOptimality}

Throughout this appendix we make the following assumptions. Following Proposition~\ref{propInvariantsOfSwap}, we may assume there is a single matching variable, a single nonmatching holding variable, and a single swapping variable. Let $\mathscr M$, $\mathscr H$, and $\mathscr S$ be the domains for the matching variable, the nonmatching holding variable, and the swapping variable, respectively. Define $\Xprod = \bigcup_{k=1}^\infty (\mathscr M \times \mathscr H \times \mathscr S)^k$. (Note $\Xcef \subset \Xprod$, but we cannot assume the reverse inclusion.)

Recall that $b = \max \{0, n_{m\cdot\cdot} \mid \mathrm{there}\allowbreak\ \mathrm{are}\allowbreak\ \mathrm{two}\allowbreak\ \mathrm{records}\allowbreak\ \mathrm{with}\allowbreak\ \mathrm{different}\allowbreak\ \mathrm{values}\allowbreak\ \mathrm{in}\allowbreak\ \mathrm{matching}\allowbreak\ \mathrm{stratum}\allowbreak\ m \in \mathscr M\}$; that $o = p/(1-p)$; and that $d(k)$ denotes the $k$-th derangement number (see Equation~\ref{eqRecurrence}).

\begin{theorem}\label{thmTightSwapBudget}
Assume that $\abs{\mathscr H}, \abs{\mathscr S} \ge 2$ (so that $\D \in \DSwap$ are not all singletons and swapping is not completely vacuous).

Suppose that the PSA satisfies \SpecSwapProdGen. Then:
\begin{enumerate}[label=(\Alph*)]
	\item If $p \in \{0,1\}$, then there exists a universe $\D_0 \in \DSwap$ such that $\epsilon_{\D_0} = \infty$. \label{thmTightSwapBudgetStatementA}
	\item If $0 < p < 1$, then there exists a universe $\D_0 \in \DSwap$ such that $\epsilon_{\D_0} \ge \ln o$. %
	\label{thmTightSwapBudgetStatementB}
	\item If $0 < p < 1$, then there exists a universe $\D_0 \in \DSwap$ such that $\epsilon_{\D_0} \ge 0.5 \ln [d(b)/d(b-2)] - \ln (o)$. %
	\label{thmTightSwapBudgetStatementC}
\end{enumerate}
\end{theorem}

The above values of $\epsilon_{\mathcal D_0}$ describe lower bounds on the protection loss budget (PLB) of the Permutation Swapping Algorithm (PSA); any differential privacy (DP) specification for the PSA must have a PLB at least equal to these values. Comparing these lower bounds to the PLB $\epsilon_{\D}^{(1)}$ given in Theorem~\ref{thmSwapDP} shows that $\epsilon_{\D}^{(1)}$ is optimal in the weak sense that there exists universes $\D_0$ for which $\epsilon_{\D_0}^{(1)}$ is arbitrarily close to the best possible budget $\epsinf_{\D_0}$.

\begin{theorem}\label{thmTightSwapBudget2}
Assume $\abs{\mathscr H}, \abs{\mathscr S} \ge 4$. For each $\D_0 \in \DSwap$, define
\[\epsilon_{\D_0}^{(\inf)} = \inf \{ \epsilon_{\D_0} \mid \mathrm{the\ PSA\ satisfies\ } \epsilon_{\D}\text{-DP}(\Xprod, \DSwap, \dHamSr, \Mult)\}.\]
(That is, $\epsinf_{\D}$ is the pointwise infimum over all PLBs $\epsilon_{\D}$ satisfied by the PSA.) Then $\epsinf_{\D}$ is the smallest budget under which the PSA satisfies the DP flavor \FlavSwapProdGen. 

Let $\epsilon_{\D}^{(1)}$ be the PLB given in Theorem~\ref{thmSwapDP}. There exists $\D_0 \in \DSwap$ such that 
\[\epsilon_{\D_0}^{(1)} - \epsinf_{\D_0} \le \begin{cases}
	f(b) &\mathrm{if\ } 0 < p < \frac{\sqrt{b+1}}{\sqrt{b+1}+1} \mathrm{\ and\ } b > 0, \\
	0 &\mathrm{otherwise,}
\end{cases}\]
where 
\[f(b) = \frac{1}{2} \ln \left[ \frac{(b+1)^2}{b(b-1)} \frac{1+ \tfrac{e}{2(b-2)!}}{1-\tfrac{e}{2b!}}\right],\]
is a positive, monotonically decreasing function for $b \ge 2$ that converges to zero, and satisfies, for example, $f(b) \le 0.148$ for all $b \ge 10$.
\end{theorem}

We emphasize that this is a weak form of optimality. A budget $\epsilon_{\D}$ can be tight at the level of the output (in the sense that $\frac{\sfPx(C(\sigma)(\xd) = \bm z)}{\sfPxdash(C(\sigma)(\xd') = \bm z)} = \exp[\epsilon_\D \dHamSr(\xd, \xd')]$ for all $\xd, \xd' \in \D$, all $\bm z$ and all $\D$); or at the level of the data (in the sense that $\Mult(\sfPx, \sfPxdash) = \epsilon_{\D} \dHamSr(\xd, \xd')$ for all $\xd, \xd' \in \D$ and all $\D$); or at the level of the universe (in the sense that $\Mult(\sfPx, \sfPxdash) = \epsilon_{\D} \dHamSr(\xd, \xd')$ for some $\xd, \xd' \in \D$, and all $\D \in \DSwap$). The optimality of Theorem~\ref{thmSwapDP} is weaker than any of these notions; all we have shown is that, for all $\delta > 0$, there exists some $\D_0 \in \DSwap$ and some $\xd, \xd' \in \D_0$ such that $\epsilon^{(1)}_{\D_0} - \Mult(\sfPx, \sfPxdash) < \delta$. Part of the suboptimality arises from the fact that $\epsilon^{(1)}_{\D}$ is a function only of $p$ and $b$. We could perform a tighter analysis of the PSA by allowing $\epsilon_{\D}$ to depend on $\D$ in more complex ways (i.e., by allowing $\epsilon_{\D}$ to be a function of other properties of $\D$, not just $b$).

\begin{proof}[Proof of Theorem~\ref{thmTightSwapBudget}]
Result~\ref{thmTightSwapBudgetStatementA} follows from Propositions~\ref{propTightSwapBudgetpZero} and~\ref{propTightSwapBudgetpOne}. Result~\ref{thmTightSwapBudgetStatementB} follows from Proposition~\ref{propTightSwapBudgetLnO}. Result~\ref{thmTightSwapBudgetStatementC} follows from Propositions~\ref{propTightSwapBudgetLnDb1} and~\ref{propTightSwapBudgetLnDb1Point5}.%
\end{proof}

\begin{proof}[Proof of Theorem~\ref{thmTightSwapBudget2}]
Because the multiverse $\DSwap$ partitions $\Xprod$, the DP constraint imposed on each universe $\D$ is independent of the constraint on another universe $\D' \ne \D$. Hence the PSA does indeed satisfy $\epsinf_{\D}$-DP$(\Xprod,\allowbreak \DSwap,\allowbreak \dHamSr,\allowbreak \Mult)$. Clearly, $\epsinf_\D \le \epsilon_{\D}$ holds for all $\D$ and all budgets $\epsilon_{\D}$ for which the PSA satisfies \SpecSwapProdGen. Hence $\epsinf_{\D}$ is the smallest budget for which the PSA satisfies the DP flavor \FlavSwapProdGen.

Moving on to the second half of the theorem, we have by Theorem~\ref{thmTightSwapBudget} that 
\[\epsilon_{\D_0}^{(1)} - \epsinf_{\D_0} = 0,\]
if $b = 0$ or $p = 0$ or $\frac{\sqrt{b+1}}{\sqrt{b+1}+1} \le p \le 1$. On the other hand, if $0 < p < \frac{\sqrt{b+1}}{\sqrt{b+1}+1}$ and $b > 0$, then
\begin{align*}
	\epsilon_{\D_0}^{(1)} - \epsinf_{\D_0} &\le \ln(b+1) - \frac{1}{2} \ln [d(b)/d(b-2)] \\
	&= \frac{1}{2} \ln \left[ (b+1)^2 \frac{\floor{\tfrac{(b-2)!}{e}+\tfrac{1}{2}}}{\floor{\tfrac{b!}{e}+\tfrac{1}{2}}} \right] \\
	&\le \frac{1}{2} \ln \left[ \frac{(b+1)^2}{b(b-1)} \frac{1+ \tfrac{e}{2(b-2)!}}{1-\tfrac{e}{2b!}}\right] \\
	&= f(b),
\end{align*}
where the first line follows by Proposition~\ref{propTightSwapBudgetLnDb1} and the second line by the identity $d(k) = \floor{\tfrac{k!}{e}+\tfrac{1}{2}}$. The second term inside the logarithm 
\[\frac{1+ \tfrac{e}{2(b-2)!}}{1-\tfrac{e}{2b!}}\]
has a numerator that decreases with $b$ and a denominator that increases. Hence this term is monotonically decreasing. The first term inside the logarithm $\frac{(b+1)^2}{b(b-1)}$ has negative first derivative and hence is also decreasing. Therefore, $f(b)$ is monotonically decreasing. Moveover, $f(b)$ is positive and converges to zero because both terms inside the logarithm are greater than one and converge to one.
\end{proof}

\begin{proposition}\label{propTightSwapBudgetpZero}
Suppose $p = 0$ and $\abs{\mathscr H}, \abs{\mathscr S} \ge 2$. %
Then there exists $\D_0 \in \DSwap$ such that $C(\xd) \ne C(\xd')$ for some $\xd, \xd' \in \D_0$. Hence, 
the PSA does not satisfy \SpecSwapProdGen\ for any finite $\epsilon_{\mathcal D_0}$ and any such $\D_0$.   
\end{proposition}

\begin{proof}
First we show that such a universe $\D_0 \in \DSwap$ exists. Given $\abs{\mathscr H}, \abs{\mathscr S} \ge 2$, the data sets $[(m,h,s), (m,h',s')]$ and $[(m,h,s'), (m,h',s)]$ (for any choice of $m \in \mathscr M$, $h \ne h' \in \mathscr H$ and $s \ne s' \in \mathscr S$) are in the same universe $\D_0 \in \DSwap$ and satisfy $C(\xd) \ne C(\xd')$. 

Let $\xd, \xd' \in \Xprod$ be data sets that are in the same universe $\D_0$. Suppose $C(\xd) \ne C(\xd')$. If $p = 0$ then the permutation $\sigma$ sampled by the PSA must be the identity. Thus, $\sfPxdash (C(\sigma(\xd')) = C(\xd)) = 0$ but $\sfPx (C(\sigma(\xd)) = C(\xd)) = 1$. Since $\dHamSr(\xd, \xd') < \infty$, the DP condition
\[\sfPx(C(\sigma(\xd)) = C(\xd)) \le \exp [\dHamSr(\xd, \xd') \epsilon_{\mathcal D_0}] \sfPxdash(C(\sigma(\xd')) = C(\xd)),\]
cannot be satisfied by a finite $\epsilon_{\D_0}$. %
\end{proof}

\begin{proposition}\label{propTightSwapBudgetpOne}
Suppose $p = 1$ and $\abs{\mathscr H}, \abs{\mathscr S} \ge 2$. %
Then there exists $\mathcal D_0 \in \DSwap$ with $n_{m_0h_0} = n_{m_0h'_0} = n_{m_0s_0} = n_{m_0s'_0} = 1$ for some $m_0 \in \mathscr M$, $h_0 \ne h'_0 \in \mathscr H$, and $s_0 \ne s'_0 \in \mathscr S$. Hence the PSA does not satisfy \SpecSwapProdGen\ for any finite $\epsilon_{\D_0}$ and any such $\D_0$.
\end{proposition}

\begin{proof}
The universe given in the proof of Proposition~\ref{propTightSwapBudgetpZero} satisfies the property: $n_{m_0h_0} = n_{m_0h'_0} = n_{m_0s_0} = n_{m_0s'_0} = 1$ for some $m_0 \in \mathscr M$, $h_0 \ne h'_0 \in \mathscr H$ and $s_0 \ne s'_0 \in \mathscr S$.

Now take any $\D_0 \in \DSwap$ that satisfies this property. Then there exists $\xd, \xd' \in \D_0$ that differ by a single swap between $(m_0,h_0,s_0)$ and $(m_0,h'_0,s'_0)$---that is,
\begin{align*} 
	{\xd}&=\left[ (m_0, h_0, s_0),  (m_0, h'_0, s'_0),  \xd_{3:n}\right], \\
	{\xd'}&=\left[ (m_0, h_0, s'_0),  (m_0, h'_0, s_0), \xd_{3:n}\right],
\end{align*}
where $\xd_{3:n}=[(M_i, H_i, S_i), i=3, \ldots, n]$. Then $n_{m_0h_0s_0}^{\xd} = n_{m_0h'_0s'_0}^{\xd} = 1$ and $n_{m_0hs_0}^{\xd} = n_{m_0h_0s}^{\xd} = 0$ for all $h \ne h_0$ and all $s \ne s_0$. Since no records can be fixed by $\sigma$ when $p = 1$, we have $n_{m_0h_0s_0}^{\sigma(\xd)} = 0$ for any possible $\sigma$ and hence $\sfPx(C(\sigma(\xd)) = C(\xd)) = 0$ but $\sfPx(C(\sigma(\xd')) = C(\xd)) > 0$.
\end{proof}

\begin{proposition}\label{propTightSwapBudgetLnO}
Suppose that $0 < p < 1$ and $\abs{\mathscr H}, \abs{\mathscr S} \ge 2$. Then there exists $\D_0 \in \DSwap$ and $m_0 \in \mathscr M$ such that $n_{m_0} \ge 2$ and $n_{m_0h}, n_{m_0s} \in \{0,1\}$ for all $h \in \mathscr H$ and $s \in \mathscr S$. A necessary condition for the PSA to satisfy \SpecSwapProd\ is that $\epsilon_{\D_0} \ge \ln o$ for any such $\D_0$.%
\end{proposition}

\begin{proof}
Let $\xd, \xd' \in \D_0$ with $\dHamSr(\xd_{m_0}, \xd'_{m_0}) = 2$ and $\dHamSr(\xd_{m}, \xd'_{m}) = 0$ for all $m \ne m_0$. (Such a pair of data sets exist because $n_{m_0} \ge 2$.) Reorder the records in $\xd'$ so that there exists a permutation $\rho$ which deranges exactly two records and satisfies $\rho(\xd') = \xd$. (Such a permutation exists by Lemma~\ref{lemmaDSteps}.)

Because $n_{m_0h}, n_{m_0s} \in \{0,1\}$ for all $m \in \mathscr M$ and $s \in \mathscr S$, there are no vacuous swaps in the $m_0$ stratum. That is, $g(\xd_{m_0}) \ne \xd_{m_0}$ for all permutations $g$ that are not the identity $\id$. Hence $G_{\xd_{m_0} \to \xd_{m_0}}=\{\id\}$. Thus,
\begin{align*}
	\frac{\sfPx(C(\sigma(\xd)) = C(\xd))}{\sfPxdash(C(\sigma(\xd')) = C(\xd))} &= \frac{\sfPx(C(\sigma_{m_0}(\xd_{m_0})) = C(\xd_{m_0}))}{\sfPxdash(C(\sigma_{m_0}(\xd'_{m_0})) = C(\xd_{m_0}))} \\
	&= \frac{\sfPx(\sigma_{m_0} = \id)}{\sfPxdash(\sigma_{m_0} = \rho)} \\
	&= o^{-2}.
\end{align*}
Hence, $\sfPxdash(C(\sigma(\xd')) = C(\xd)) \le \exp [\dHamSr(\xd, \xd') \epsilon_{\mathcal D_0}] \sfPx(C(\sigma(\xd)) = C(\xd))$ if and only if $\epsilon_{\D_0} \ge \ln o$.
\end{proof}

\begin{proposition}\label{propTightSwapBudgetLnDb1}
Suppose that $0 < p < 1$ and $\abs{\mathscr H}, \abs{\mathscr S} \ge 4$. %
Then there exists $\mathcal D_0 \in \DSwap$ that has the following properties: 
\begin{equation}\label{eqDivAssume}
	\max_h n_{m_0h} \le \frac{b}{2} - 1 \text{ and } \max_s n_{m_0s} \le \frac{b}{2}  - 1,
\end{equation}
for some $m_0 \in \mathscr M$ with $n_{m_0} = b$, and there exists $h_1 \ne h_2 \in \mathscr H$ and $s_1 \ne s_2 \in \mathscr S$ such that
\begin{equation}\label{eqDiv2Assume}
	n_{m_0 h_1} = n_{m_0 h_2} = n_{m_0 s_1} = n_{m_0 s_2} = 1.
\end{equation}
A necessary condition for the PSA to satisfy \SpecSwapProdGen\ is that 
\[\epsilon_{\D_0} \ge 0.5 \ln [d(b)/d(b-2)] - \ln (o),\] 
for any $\D_0$ satisfying the above properties.
\end{proposition}

We will use the following two lemmata in the proof of Proposition~\ref{propTightSwapBudgetLnDb1}.

\begin{lemma}\label{lemmaSwapDist}
For any $\xd$ and any permutation $g$,
\[ \dHamSr(\xd, g(\xd)) \le k_g,\]
where $k_g$ is the number of records that are deranged by $g$.
\end{lemma}

\begin{proof}
For every record $(M_i, H_i, S_i)$ permuted by $g$, the counts in the fully saturated contingency table can change by at most 2: the count $n_{M_i H_i S_i}$ will decrease by (at most) 1 and the count $n_{M_i, H_i, S_{g(i)}}$ will increase by (at most) 1. Thus, in sum, the counts $n_{mhs}$ can change by at most $2k_g$. That is,
\[\ell_1^r(\xd, g(\xd)) = \sum_{m,h,s} \abs{n_{mhs}^{\xd} - n_{mhs}^{g(\xd)}} \le 2k_g.\]
The desired result then follows by Lemma~\ref{lemmaEll1HamDistEqual}.
\end{proof}

\begin{lemma}\label{lemmaDiverseVacuous}
Suppose that $\D_0 \in \DSwap$ satisfies Equation~\ref{eqDivAssume}. Then there exists $\xd \in \D_0$ and a derangement $g$ of $\xd_{m_0}$ such that
\[ \dHamSr(\xd_{m_0}, g(\xd_{m_0})) = b.\]
(In fact, such an $\xd$ and $g$ exist if and only if $\D_0$ satisfies Equation~\ref{eqDivAssume}.)
\end{lemma}

\begin{proof}
We suppress the subscript $m_0$ in $\xd_{m_0}$ throughout the proof.

We begin by consider the cases $b = 1$ and $b = 0$ individually. Equation~\ref{eqDivAssume} implies that $b \ne 1$. Similarly, no derangement of $\xd$ exists when $b = 1$. In the case of $b = 0$, the result is also trivial. Hence we may assume throughout that $b \ge 2$.

``$\Rightarrow$'': Suppose that $\D_0$ does not satisfy Equation~\ref{eqDivAssume}. Then $n_{m_0} < b$ or there exists (WLOG) a swapping category $s_0$ such that $n_{m_0s_0} \ge b/2$. In the first case, any permutation $g$ of $\xd$ deranges at most $n_{m_0}$ records and hence $\dHamSr(\xd, g(\xd)) < b$ by Lemma~\ref{lemmaDiverseVacuous}. By the pigeonhole principle, the second case implies every derangement $g$ of $\xd_{m_0}$ must send a record with swapping value $s_0$ to a record that also has value $s_0$. Yet the counts $n_{mhs}$ are unaffected by permutations of records within the same swapping category $s$. Hence
\[\sum_{h,s} \abs{n_{m_0hs}^{\xd} - n_{m_0hs}^{g(\xd)}} \le 2(k_g-1) < 2b,\]
(where the first inequality follows by the reasoning in the proof of Lemma~\ref{lemmaSwapDist}). The desired result then follows by Lemma~\ref{lemmaEll1HamDistEqual}.

``$\Leftarrow$'': Assume for now that $b$ is even. By Equation~\ref{eqDivAssume}, there exists $\xd \in \D_0$ whose records are ordered so that every odd record has a different $\Vswap$ and $\Vhold$ compared to the subsequent record. That is, $H_i \ne H_{i+1}$ and $S_i \ne S_{i+1}$ for all odd $i$. (One can construct $\xd$ by picking any $\xd' \in \D_0$, ordering the records of $\xd' \in \D_0$ so that the values of $\Vhold$ differ between consecutive records, and then permuting $\Vswap$ so that their values also differ between consecutive records.)

Construct $g$ by swapping odd and even records:
\[g(i) = \begin{cases}
	i+1 & \text{if } i \text{ odd}, \\
	i-1 & \text{if } i \text{ even}.
\end{cases}\]
Then $k_g = b$ and $\dHamSr(g(\xd), \xd) = b$. 

Now suppose that $b$ is odd. Then Equation~\ref{eqDivAssume} implies that there exists $\xd \in \D_0$ such that 
\begin{enumerate}[label=\arabic*)]
	\item $H_i \ne H_{i + 1}$ and $S_i \ne S_{i+1}$ for all odd $i < n_{m_0}$; and
	\item $H_{n_{m_0}} \notin \{H_{n_{m_0}-1}, H_{n_{m_0}-2}\}$ and $S_{n_{m_0}} \notin \{S_{n_{m_0}-1}, S_{n_{m_0}-2}\}$.
\end{enumerate}
Why is this true? We already know that 1) must be true by the proof for even $b$. Suppose that 2) is not true for any $\xd$. Then it must not be true for any $\xd'$ that are just reorderings of the records of $\xd$. Hence, for every adjacent pair $(i,i+1)$ (with $i < n_{m_0}$ odd), we must have $H_{n_{m_0}} \in \{ H_i, H_{i+1}\}$ or $S_{n_{m_0}} \in \{S_{i}, S_{i+1}\}$. Yet this would contradict Equation~\ref{eqDivAssume}.

Construct $g$ by swapping odd and even records, bar the final three records, which are permuted. That is,
\[g(i) = \begin{cases}
	i+1 & \text{if } i < n_{m_0} \text{ odd}, \\
	i-1 & \text{if } i < n_{m_0}-1 \text{ even},\\
	n_{m_0} & \text{if } i = n_{m_0}-1,\\
	n_{m_0}-2 & \text{if } i = n_{m_0}.
\end{cases}\]
As before, $k_g = b$ and $\dHamSr(g(\xd), \xd) = b$.
\end{proof}

\begin{proof}[Proof of Proposition \ref{propTightSwapBudgetLnDb1}]
Fix some $\D_0 \in \DSwap$, which satisfies the Equations~\ref{eqDivAssume} and~\ref{eqDiv2Assume}. Such a universe exists when $\abs{\mathscr H}, \abs{\mathscr S} \ge 4$ because, for example, 
\[\xd = 
[(m_0, h_1, s_1), (m_0, h_2, s_2), (m_0, h_3, s_3), (m_0, h_3, s_3), (m_0, h_4, s_4), (m_0, h_4, s_4)],\]
satisfies these properties.

Let $\epsilon_0 = 0.5 \ln [d(b)/d(b-2)] - \ln (o)$. %
We want to prove that
\begin{equation}\label{eqWTSForTight}
	\frac{\sfPx(C(\sigma (\xd)) = C(\bm z))}{\sfPxdash (C(\sigma(\xd')) = C(\bm z))} = \exp\left[\dHamSr(\xd, \xd')\epsilon_0\right],
\end{equation}
for some $\xd, \xd', \bm z \in \D_0$. %

We will construct $\xd$ and $\xd'$ so that they are identical except within the matching category $m_0$. Then by independence between matching categories,
\[\frac{\sfPx(C(\sigma (\xd)) = C(\bm z))}{\sfPxdash (C(\sigma(\xd')) = C(\bm z))} = \frac{\sfPx(C(\sigma_{m_0} (\xd_{m_0})) = C(\bm z_{m_0}))}{\sfPxdash (C(\sigma_{m_0}(\xd'_{m_0})) = C(\bm z_{m_0}))}.\]
This justifies dropping the subscript $m_0$ from $\xd_{m_0}$ and ignoring records with matching categories not equal to $m_0$ throughout the remainder of the proof.

We construct $\xd$ as follows: The first two records of $\xd$ are $(m_0, h_1, s_1)$ and $(m_0, h_2, s_2)$. The remainder of the records satisfy Equation~\ref{eqDivAssume}. Hence construct the remainder of $\xd$ according to the procedure given in the proof of Lemma~\ref{lemmaDiverseVacuous}. Let $\xd'$ be the same as $\xd$, except interchange the values of the swapping variable of the first two records. That is, $\xd' = [(m_0,h_1,s_2), (m_0,h_2,s_1), \xd_{3:n}]$.

Lemma \ref{lemmaDiverseVacuous} implies there exists a permutation $g_0$, which fixes the first two records and deranges the remaining records such that
\[\dHamSr(\xd, g_0(\xd)) = b-2.\]
Moreover, for $g'_0 = g_0 \circ (12)$, we have 
\[\dHamSr(\xd', g'_0(\xd')) = b.\]
Set $\bm z = g_0(\xd) = g'_0(\xd')$. 

Now we will prove Equation~\ref{eqWTSForTight} holds for these choices of $\xd, \xd'$ and $\bm z$. We have
\[\frac{\sfPx(C(\sigma (\xd)) = C(\bm z))}{\sfPxdash (C(\sigma(\xd')) = C(\bm z))} = \frac{\sum\limits_{\bm z' \text{ re-ordering of } \bm z} \sfPx(\sigma (\xd) = \bm z')}{\sum\limits_{\bm z' \text{ re-ordering of } \bm z} \sfPxdash (\sigma(\xd') = \bm z')}. \]%

Fix some $\bm z'$, which is a reordering of $\bm z$---i.e., some $\bm z'$ with $C(\bm z') = C(\bm z)$. We will show that $\frac{\sfPx(\sigma (\xd) = \bm z')}{\sfPxdash (\sigma(\xd') = \bm z')} = \exp(2\epsilon_0)$, when assuming that one of the numerator or the denominator is nonzero (which implies the other is also nonzero, since $\xd$ and $\xd'$ differ by a single swap). Since both the numerator and denominator are nonzero when $\bm z' = \bm z$, this result will prove Equation~\ref{eqWTSForTight}.

We know that $\dHamSr(\bm z, \bm z') = 0$ and $\dHamSr(\xd', \bm z) = b$. Then using the triangle inequality (twice, once for $\le$ and once for $\ge$), $\dHamSr(\xd', \bm z') = b$. Lemma~\ref{lemmaSwapDist} implies that $k_{g} = b$ for all $g \in G_{\xd' \to \bm z'}$. 

By the same reasoning, $\dHamSr(\xd, \bm z') = b-2$. This implies $k_g \ge b-2$ for all $g \in G_{\xd \to \bm z'}$ by Lemma~\ref{lemmaSwapDist}. We now show that, in fact, $k_g = b-2$. By construction,
\[n_{m_0 h_1 s_1}^{\xd} = n_{m_0 h_2 s_2}^{\xd} = 1 \text{ and } n_{m_0 h_1 s}^{\xd} = n_{m_0 h_2 s}^{\xd} = n_{m_0 h s_1}^{\xd} = n_{m_0 h s_2}^{\xd} = 0,\]
for all $h \notin \{h_1, h_2\}$ and $s \notin \{s_1, s_2\}$. These equations also hold for $\bm z$ and hence also for $\bm z'$. Thus, all $g \in G_{\xd \to \bm z'}$ must fix the first two records and hence $k_g \le b-2$.

In the proof of Theorem~\ref{thmSwapDP}, we showed that $\sfPx(\sigma = g)$ only depends on $k_g$ and, furthermore, that
\[\frac{\sfPx (\sigma = g)}{\sfPxdash (\sigma = g')} = \frac{(1-p)^2 d(b)}{p^2 d(b-2)},\]
when $k_{g} = b-2$ and $k_{g'} = b$. Thus,
\[\frac{\sfPx(\sigma (\xd) = \bm z')}{\sfPxdash (\sigma(\xd') = \bm z')} = \frac{\sum_{g \in G_{\xd \to \bm z'}} \sfPx (\sigma = g)}{\sum_{g' \in G_{\xd' \to \bm Z'}} \sfPxdash (\sigma = g')} = \frac{(1-p)^2 d(b)}{p^2 d(b-2)} = \exp(2 \epsilon_0),\]
since $k_g = b-2$ for all $g \in G_{\xd \to \bm z'}$ and $k_{g'} = b$ for all $g' \in G_{\xd' \to \bm z'}$.
\end{proof}

\begin{proposition}\label{propTightSwapBudgetLnDb1Point5}
Suppose that $0 < p \le 0.5$ and $\abs{\mathscr H}, \abs{\mathscr S} \ge 2$. Then there exists $\D_0 \in \DSwap$ such that $b = 2$ and 
\begin{equation}\label{eqDiv2AssumeRepeat}
	n_{m_0 h_1} = n_{m_0 h_2} = n_{m_0 s_1} = n_{m_0 s_2} = 1,
\end{equation}
for some $m_0 \in \mathscr M$ with $n_{m_0} = b$ and some $h_1 \ne h_2$ and $s_1 \ne s_2$.

A necessary condition for the PSA to satisfy \SpecSwapProdGen\ is that 
\[\epsilon_{\D_0} \ge 0.5 \ln [d(b)/d(b-2)] - \ln (o),\]
for any such $\D_0$.
\end{proposition}

\begin{proof}
Because $\abs{\mathscr H}, \abs{\mathscr S} \ge 2$, any data set of the form $[(m_0, h_1, s_1), (m_0, h_2, s_2)]$ satisfies Equation~\ref{eqDiv2AssumeRepeat}. Moreover, $\xd$ is in some universe $\D_0$, thereby proving the first half of the proposition. The second half of the proposition follows by the same reasoning as the proof of Proposition~\ref{propTightSwapBudgetLnDb1} applied to $\xd = [(m_0, h_1, s_1), (m_0, h_2, s_2)]$ and $\xd = [(m_0, h_1, s_2), (m_0, h_2, s_1)]$.  
\end{proof}

\section{Zero-Concentrated Differential Privacy}\label{appZCDP}

The normalized R\'enyi metric $\Dnor$ is defined as:
\[\Dnor (\sfP, \sfQ) = \sup_{\alpha > 1} \frac{1}{\sqrt\alpha} \max \left\{ \sqrt{D_\alpha (\sfP||\sfQ)}, \sqrt{D_\alpha (\sfQ||\sfP)} \right\},\]
where $D_\alpha$ is the R\'enyi divergence of order $\alpha$:
\[D_{\alpha} (\sfP||\sfQ) = \begin{cases}
    \frac{1}{\alpha-1} \ln \int \left[\frac{d\sfP}{d\sfQ} \right]^\alpha d\sfQ, &\mathrm{if\ } \sfP \mathrm{\ is\ absolutely\ continuous\ wrt.\ } \sfQ, \\
    \infty &\mathrm{otherwise.}
\end{cases}\]
Here $\frac{d\sfP}{d\sfQ}$ is the Radon-Nikodym derivative of $\sfP$ with respect to $\sfQ$. %

The term `zCDP' refers to the class of DP flavors whose output premetric is the normalized R\'enyi metric, much as `pure DP' refers to the class of flavors whose output premetric is the multiplicative distance (Equation~\ref{eqMultDistDefnPart2}). Note that we reparameterize $\rho$ so that $\Dnor$ is a metric \citep{bailieRefreshmentStirredNot2026}. This is similar to the parameterization of zCDP given in \citet{canonneDiscreteGaussianDifferential2022} and \citet{kairouzDistributedDiscreteGaussian2021}. The standard formulation of zCDP, as originally given in \citet{bunConcentratedDifferentialPrivacy2016a}, uses the square of the normalized R\'enyi metric as its output premetric. Consequently, the standard parameterization of zCDP is equal to $\rho^2$ under our formulation of zCDP.

\section{Proof and Discussion of Theorem~\ref{thmTDASatisfiesDP}}\label{appProofTDASatisfiesDP}

\begin{proof}[Proof of Theorem~\ref{thmTDASatisfiesDP}]
We first analyze the TopDown Algorithm (TDA) for producing the P.L. 94-171 Redistricting Summary File (PL). \citet{abowd2022topdown} prove that the mechanism $\Th$ that produces the household Noisy
Measurement File (NMF) satisfies $\rho$-DP$(\Xcef,\allowbreak \{\Xcef\},\allowbreak d_{r_{bs}}^{hh},\allowbreak \Dnor)$, where $\rho^2 = 0.07$ and $d_{r_{bs}}^{hh}$ is the input premetric corresponding to bounded differential privacy (DP) on household-records. %
But $(\Xcef,\allowbreak \{\Xcef\},\allowbreak d_{r_{bs}}^{hh},\allowbreak \Dnor)$ and $(\Xcef,\allowbreak \{\Xcef\},\allowbreak \dHamSh,\allowbreak \Dnor)$ are equivalent DP flavors (\cite[see][]{bailieRefreshmentStirredNot2026}). %
Hence $\Th$ satisfies \SpecTDA\ by the second half of Proposition~\ref{propNestedscD} %
with $\rho^2 = 0.07$. We can similarly conclude that $\Tp$ satisfies \SpecTDA\ with $\rho^2 = 2.56$. Then by composition, the mechanism $\bm T_{ph} = [\Tp, \Th]$ has protection loss budget (PLB) $\rho^2 = 0.07 + 2.56 = 2.63$. Proposition~\ref{propReleaseF} implies the invariants $\cTDA(\Xp, \Xh)$---considered as a data release mechanism---satisfies \SpecTDA\ with $\rho^2 = 0$. Therefore, the composed mechanism $\bm T = [\bm T_{ph}, \cTDA]$ has budget $\rho^2 = 2.63$. The second step of the TDA is postprocessing on $\bm T$ and hence has the same budget.

The argument for producing the Demographic and Housing Characteristics File is almost analogous. The composed mechanism $\bm T_{ph} = [\bm T_{p}, \Th]$ has budget $\rho^2 = 7.70+4.96=12.66$. Now the second step of the TDA also uses the PL $\bm P$. Hence, this second step is postprocessing on the composed mechanism $[\bm T_{ph}, \bm P, \cTDA]$. This composed mechanism has budget $\rho^2 = 12.66 + 2.63 + 0 = 15.29$. 

The second half of the theorem follows from Proposition~\ref{propReleaseF}. (Hence it can be generalized from $(\Xcef, \scD_{\bm c'}, \dHamS^p, \Dnor)$ to any DP flavor $(\Xcef, \scD, \dX, \dT)$ satisfying the assumptions of this proposition.)
\end{proof}

The second step of the TDA requires access to both the NMFs $[\Tp(\Xp), \Th(\Xh)]$ and the invariant statistics $\cTDA(\Xp, \Xh)$ computed on the Census Edited File (CEF). Under the DP flavor $(\Xcef,\allowbreak \{\Xcef\},\allowbreak \dX,\allowbreak \dT)$, the invariant statistics $\cTDA(\Xp,\Xh)$ cannot be released with finite budget. So the second step of the TDA is not actually postprocessing under this flavor---it is only postprocessing when conditioning on the invariants. %
In fact, the second half of Theorem~\ref{thmTDASatisfiesDP} shows that any argument that relies on the TDA's second step being postprocessing must necessarily use a DP flavor that conditions on the invariants $\cTDA$.

To avoid inflating the PLB by a factor of 99,999, it is necessary to use person-records as the resolution of the Hamming distance in the TDA's DP specification. While the household mechanism $\Th$ satisfies \FlavhTDA, the sensitivity of the person-level query $\Qp$ due to a single change in a household-record is 99,999. This is because the maximum possible household size in the CEF is 99,999 \citep{populationreferencebureauDisclosureAvoidance20202023}. This means $\Tp$ can satisfy \FlavhTDA\ only if the PLB is amplified by 99,999. %

\end{document}